\documentclass{article}
\usepackage{JalexMath}
\usepackage{authblk}
\usepackage[toc,page]{appendix}
\usepackage{rotating}
\usepackage{caption}
\usepackage[margin=1in]{geometry}
\usepackage[showonlyrefs]{mathtools}

\mathtoolsset{showonlyrefs=true} 

\newcommand{\nocontentsline}[3]{}
\newcommand{\tocless}[2]{\bgroup\let\addcontentsline=\nocontentsline#1{#2}\egroup}

\newcommand{\dagg}{^\dagger}
\newcommand{\eps}{\e}
\newcommand{\eprp}{\ket{\r{EPR}_d}}
\newcommand{\epr}{{\r{EPR}_d}}

\newcommand{\eprrho}{{\r{EPR}_\s}}
\newcommand{\pauli}{\m P_d^{\otimes n}}
\newcommand{\paulin}[1]{\m P_d^{\otimes #1}}
\newcommand{\drho}[3]{\mathop{D}\limits_{#1}\left(#2\middle\| #3\right)}

\complexity{LCS}
\DeclareMathOperator*\can{\sf{can}}

\begin{document}
\title{Robust self-testing for linear constraint system games}

\author{Andrea Coladangelo}
\author{Jalex Stark}
\date{}
\affil{
Computing and Mathematical Sciences, Caltech\\ 
{\{acoladan,jalex\}@caltech.edu}
}
\maketitle

\abstract{
	We study linear constraint system (LCS) games over the ring of arithmetic modulo $d$. 
    We give a new proof that certain LCS games (the Mermin--Peres Magic Square and Magic Pentagram over binary alphabets, together with parallel repetitions of these) have unique winning strategies, where the uniqueness is robust to small perturbations.
	In order to prove our result, we extend the representation-theoretic framework of Cleve, Liu, and Slofstra \cite{cleve2016perfect} to apply to linear constraint games over $\Z_d$ for $d\geq 2$. We package our main argument into machinery which applies to any nonabelian finite group with a ``solution group'' presentation. We equip the $n$-qubit Pauli group for $n\geq 2$ with such a presentation; our machinery produces the Magic Square and Pentagram games from the presentation and provides robust self-testing bounds.
	The question of whether there exist LCS games self-testing maximally entangled states of local dimension not a power of 2 is left open. A previous version of this paper falsely claimed to show self-testing results for a certain generalization of the Magic Square and Pentagram mod $d\neq 2$. We show instead that such a result is impossible.
}

\pagenumbering{gobble}
\pagebreak
\pagenumbering{arabic}

\tocless\section{Introduction}
	In \cite{peres1990incompatible,mermin1990simple}, Mermin and Peres discovered an algebraic coincidence related to the $3\times3$ ``Magic Square'' of operators on $\C^2\otimes \C^2$ in Figure \ref{fig:magic-square-operators-introduction}.

If we pick any row and take the product of the three operators in that row (note that they commute, so the order does not matter), we get the identity operator. Similarly, we can try this with the columns. Two of the columns give identity while the other  gives $-1$ times identity. Thus, the product of these nine operators depends on whether they are multiplied row by row or column by column. This can be exploited to define a two-player, one-referee game called the Mermin--Peres Magic Square game \cite{aravind2004quantum} (see Definition \ref{definition:linear-constraint-game} and Figure \ref{fig:magic-square-equations} for a formal definition). Informally, the Mermin--Peres Magic Square game mod $2$ is as follows. The players claim to have a $3\times 3$ square of numbers in which each row and each of the first two columns sums to $0\pmod 2$, while the third column sums to $1\pmod 2$. (The players are usually called ``provers'', since they try to prove that they have such a square.) The referee asks the first player to present a row of the supposed square and the second to present a column. They reply respectively with the $3$ entries of that row and column in $\{0,1\}$. They win if their responses sum to $0$ or $1$ as appropriate, and they give the same number for the entry where the row and column overlap. This game can be won with probability $1$ by provers that share two pairs of maximally entangled qubits of dimension $2$, but provers with no entanglement can win with probability at most $\frac89$. Games which are won in the classical case with probability $< 1$ but are won in the quantum case with probability $1$ are known as \emph{pseudotelepathy games}.

\begin{figure}[!b]
	\caption{On the left are the operators of the Magic Square. $X$ and $Z$ are the generalized Pauli operators, i.e. they are unitaries for which $X^2 = Z^2 = I$ and each permutes the eigenbasis of the other. Across any solid line, the three operators commute and their product is identity. Across the dashed line, the operators commute and their product is $-1$ times identity.}
	\label{fig:magic-square-operators-introduction}
	\begin{tabular}
		{m{0.5\textwidth}m{0.5\textwidth}}
		  \resizebox{0.5\textwidth}{!}{\includegraphics{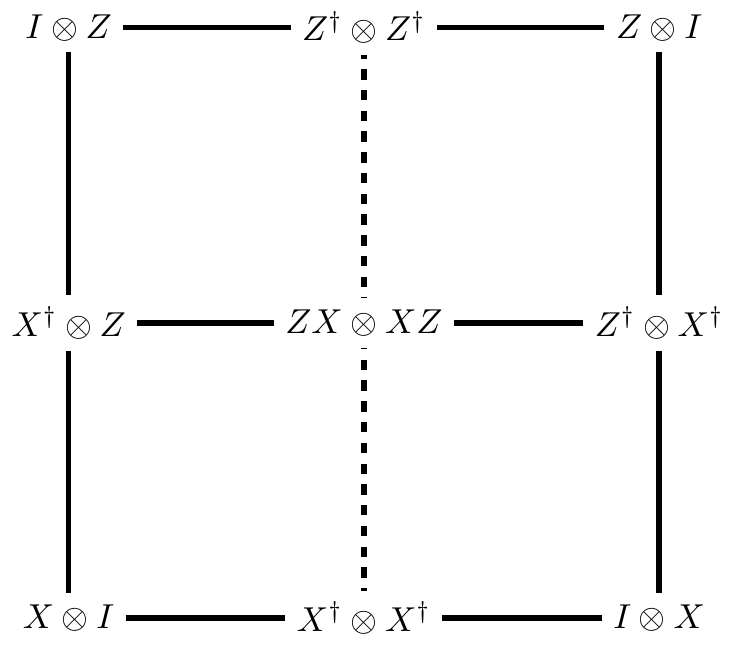}}
		 &\resizebox{0.5\textwidth}{!}{\includegraphics{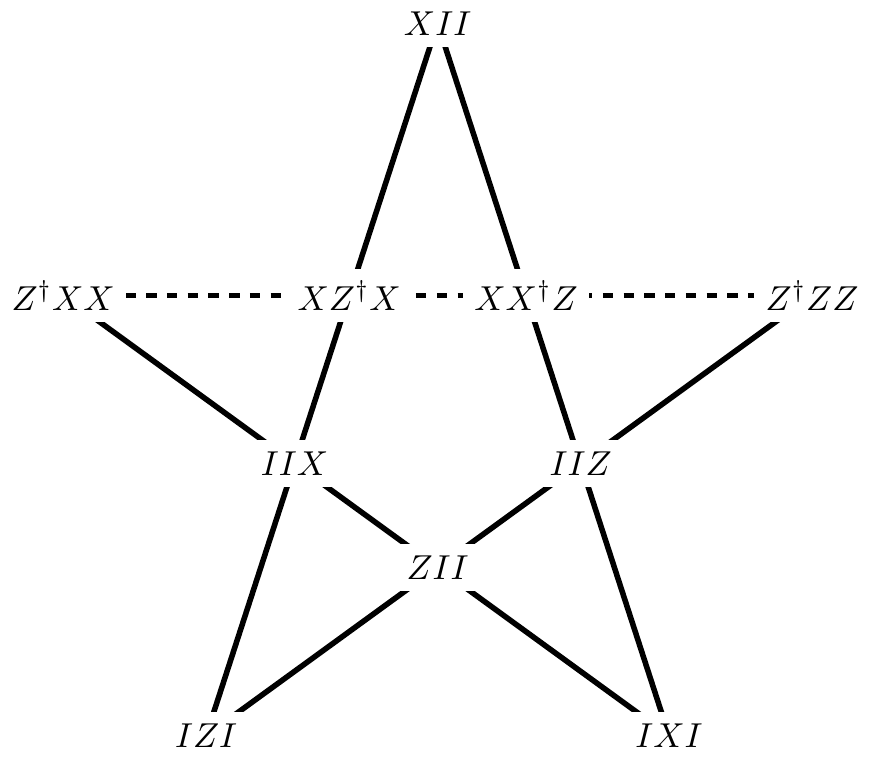}}
	\end{tabular}
	\caption{On the right are the operators of the Magic Pentagram. These are operators on $(\C^2)^{\otimes 3}$; the tensor product symbols are omitted. Across any line, the four operators commute. Across any solid line, the alternating product $AB^\dagger CD^\dagger$ of the four operators is identity. Across the dashed line, the alternating product (computed from left to right) is $-1$ times identity.}
	\label{fig:magic-pentagram-operators-introduction}
\end{figure}

How special is this ``algebraic coincidence'' and the corresponding game? We can refine this question into a few sub-questions.
\begin{question}\label{question:is-magic-square-setup-unique?}
	Are there other configurations of operators with similarly interesting algebraic relations? Do they also give rise to pseudotelepathy games?
\end{question}

Arkhipov \cite{arkhipov2012extending} gives a partial answer to this question by introducing the framework of \emph{magic games}. Starting from any finite graph, one can construct a magic game similar to the Magic Square game. Arkhipov finds that there are exactly two interesting such magic games: the Magic Square (derived from $K_{3,3}$, the complete bipartite graph with parts of size $3$) and the Magic Pentagram (derived from $K_5$, the complete graph on $5$ vertices). 
Subsequently, Cleve and Mittal \cite{cleve2014characterization} introduced \emph{linear constraint system games} (hereafter referred to as \emph{LCS games}), which can be thought of as a generalization of Arkhipov's magic games from graphs to hypergraphs. Moreover, they proved that any linear constraint game exhibiting pseudotelepathy requires a maximally entangled state to do so. Their result also suggested that there may be other interesting linear constraint games to find. Indeed, Ji showed \cite{ji2013binary} that there are families of linear constraint games requiring arbitrarily large amounts of entanglement to win.

\begin{question}\label{question:is-magic-square-solution-unique?}
	The easiest proof of correctness for a Magic Square game strategy uses the fact the observables measured by the players satisfy the appropriate algebraic relations. Is this a necessary feature of any winning strategy?
\end{question}

In order to answer questions like this, Cleve, Liu, and Slofstra \cite{cleve2016perfect} associate to each LCS game an algebraic invariant called the \emph{solution group} (see Section \ref{sec:linear-constraint-games} for a precise definition), and they relate the winnability of the game to the representation theory of the group. In particular, they show that any quantum strategy winning the game with probability $1$ corresponds to a representation of the solution group---in other words, that the observables in a winning strategy must satisfy the algebraic relations captured by the group. This reduces the problem of finding LCS games with interesting properties to the problem of finding finitely-presented groups with analogous representation-theoretic properties, while maintaining combinatorial control over their presentations. Slofstra used this idea together with techniques from combinatorial group theory to resolve the weak Tsirelson problem \cite{slofstra2016tsirelson}. By including some techniques from the stability theory of group representations, he improved this result to show that the set of quantum correlations is not closed \cite{slofstra2017set}. In words, he constructed an LCS game which can be won with probability arbitrarily close to $1$ with finite-dimensional quantum strategies, but cannot be won with probability $1$ by any finite (or infinite) dimensional quantum strategy (in the tensor product model).

\begin{question}\label{question:is-magic-square-game-rigid}
	We introduced the magic square operators and then noticed that they satisfy certain algebraic relations. Do these algebraic relations characterize this set of operators? Could we have picked a square of nine different operators, possibly of much larger dimension, satisfying the same relations?
\end{question}

This question was resolved by Wu et.\ al \cite{wu2016device}. They showed that any operators satisfying the same algebraic relations as those in the Magic Square game are equivalent to those in Figure \ref{fig:magic-square-operators-introduction}, up to local isometry and tensoring with identity. This is sometimes referred to as \emph{rigidity} of the Magic Square game. Moreover, they showed that the Magic Square game is \emph{robustly rigid}, or \emph{robustly self-testing}. Informally, we say that a game is \emph{rigid with $O(\d(\e))$-robustness and perfect completeness} if whenever Alice and Bob win the game with probability at least $1-\e$, then there is a local isometry taking their state and measurement operators $O(\d(\e))$-close to an ideal strategy, possibly tensored with identity.


\paragraph{Our contributions}
Our main result is a robust self-testing theorem which applies to any linear constraint game with sufficiently nice solution group; this is stated as Theorem \ref{thm:robust-self-testing}. Our proof employs the machinery of \cite{cleve2016perfect} and \cite{slofstra2016tsirelson}. We apply the general self-testing result to conclude robust rigidity for the Magic Square game, the Magic Pentagram game, and for a certain repeated product of these two games. We informally state these results now. We emphasize that these results are not new, but it is new that we can achieve all three as simple corollaries of the main self-testing machinery. The general result holds for LCS games mod $d$, but the only nontrivial application we have is for LCS games mod $2$.

\begin{thm}[Informal, c.f. Definition \ref{defn: robust-self-testing} and Theorem \ref{thm:robust-self-testing-square}]\label{thm:informal-magic-square}
	The Magic Square game is rigid with $O(\e)$-robustness and perfect completeness. The ideal state is two copies of the maximally entangled state of local dimension $2$, and the ideal measurements are onto the eigenbases of the operators in Figure \ref{fig:magic-square-operators-introduction}.
\end{thm}
\noindent
This recovers the same asymptotics as in \cite{wu2016device}. Note that they state their robustness as $O(\sqrt\e)$; this is because they use the Euclidean distance $\norm{\ket\psi - \ket{\text{ideal}}}$, while we use the trace-norm distance of density operators $\norm{\rho - \rho_{\text{ideal}}}_1$. 

\begin{thm}[Informal, c.f. Theorem \ref{thm:robust-self-testing-pentagram}]\label{thm:informal-magic-pentagram}
	The Magic Pentagram game (see Figure \ref{fig:magic-pentagram-equations} for a definition) is rigid with $O(\e)$-robustness and perfect completeness. The ideal state is three copies of the maximally entangled state of local dimension $2$, and the ideal measurements are onto the eigenbases of the operators in Figure \ref{fig:magic-pentagram-operators-introduction}.
\end{thm}
\noindent
This recovers the same asymptotics as \cite{kalev2017rigidity}, up to translation between distance measures.

Applying our general self-testing theorem to the LCS game product
\footnote{This is defined precisely in Definition \ref{def:lcs-game-product}. This is similar to but not the same as playing multiple copies of the game in parallel.} 
of many copies of the Magic Square game yields a self-test for $n$ maximally entangled pairs of qubits and associated $n$-qubit Pauli measurements. 

\begin{thm}[Informal, c.f. Theorem \ref{thm:self-testing-pauli-LCS}]\label{thm:informal-n-qudit-pauli-self-testing}
	For any $n\geq 2$, there is a linear constraint system game with $O(n^2)$ variables, $O(n^2)$ equations, and $\Z_2$-valued answers which is rigid with $O(n^{10}\e)$-robustness and perfect completeness.
	The ideal state is $n$ copies of the maximally entangled state of local dimension $2$. The ideal measurements are onto the eigenbases of certain Pauli operators of weight at most $5$. 
\end{thm}
The polynomial scaling in $n$ is similar to previous works that self-test $n$ pairs of maximally entangled qubits via copies of the magic square game \cite{coladan17, CN16}, but we obtain our bound by a simple application of our general self-testing theorem. 

\tocless\subsection{Proof Overview}
We step away from games and back towards algebra to discuss Question \ref{question:is-magic-square-game-rigid}. Suppose we wanted a $3\times 3$ square of operators, call them $e_1$ through $e_9$, with the same relations as those in the Magic Square. Concretely, those relations are as follows:\\
\begin{tabular}
	{m{0.2\textwidth}m{0.6\textwidth}}

	  \resizebox{0.2\textwidth}{!}{\includegraphics{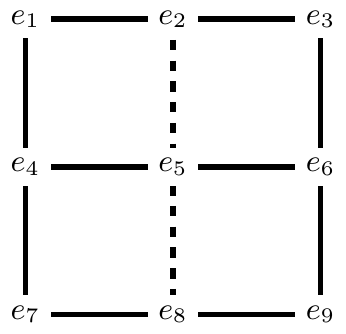}}
	& \begin{itemize}
		\item The linear constraints of each row and column: $e_2e_5e_8 = -I$,
		$e_1e_2e_3 = e_4e_5e_6 = e_7e_8e_9 = e_1e_4e_7 = e_3e_6e_9 = I$. 
		\item Commutation between operators in the same row or column: $e_1e_2 = e_2e_1$, $e_1e_3 = e_3e_1$, $e_2e_3=e_3e_2$, \ldots,
		$e_3e_6 = e_6e_3$, $e_3e_9=e_9e_3$, $e_6e_9=e_9e_6$.
		\item Associated unitaries have $2$ eigenspaces: $e_i^2= I$ for all $i$.
	\end{itemize}
\end{tabular}

These are just multiplicative equations. We can define an abstract \emph{group} whose generators are the $e_i$ and whose relations are those above. This is, in a sense, the most general object satisfying the Magic Square relations. More precisely, any square of operators satisfying these relations is a \emph{representation} of this group. 
It's not hard to compute that this group is isomorphic to the group of two-qubit Pauli matrices, a friendly object. (This is proven as Proposition \ref{prop:solution-group-square}.) 
This group is the \emph{solution group} of the magic square game. We study the representation theory of the solution group of the magic square game, and we apply \cite{cleve2016perfect} to deduce the exact version of our self-testing Theorem \ref{thm:informal-magic-square} (i.e. the $\e=0$ case). One might view our proof via solution groups as an ``algebrization'' of the proof in \cite{wu2016device}.


In order to get the robustness bounds, we must work significantly harder. Tracing through the proof of the main result of \cite{cleve2016perfect}, a finite number of equalities between various operators are applied. Knowing how many equalities are needed, one can get quantitative robustness bounds by replacing these with approximate equalities and then applying finitely many triangle inequalities. In order to carry out this counting argument, we introduce a measure of complexity for linear constraint games and then upper bound the robustness parameter as a function of this complexity.

This complexity measure depends on the use of van Kampen diagrams, a graphical proof system for equations in finitely-presented groups. Van Kampen diagrams are introduced in \S \ref{subsection:van-kampen-diagrams}. Several of our main proofs reduce to reasoning visually about the existence of such diagrams.
Manipulating the chains of approximate equalities requires us to develop familiarity with a notion of state-dependent distance; this is done in \S \ref{subsection:state-dependent-distance}.


\tocless\subsection{Organization}

In Section \ref{sec:preliminaries}, we establish basic tools that we'll use without comment in the main body of the paper. In Section \ref{sec:linear-constraint-games}, we give the definition and basic properties of linear constraint games over $\Z_d$. Those familiar with linear constraint games over $\Z_2$ will not find surprises here. In Section \ref{sec:self-testing}, we establish our measure of LCS game complexity and prove our general robust self-testing result, Theorem \ref{thm:robust-self-testing}. We warm up first by proving the $\e=0$ case of the theorem in \S\ref{subsection:exact-self-testing}. We then introduce two new ingredients to obtain a robust version. In \S\ref{subsection:stability-lemma}, we give a proof by Vidick \cite{vidick2017approx} of a so-called \emph{stability theorem for representations of finite groups} (Lemma \ref{lemma:vidick-gowers-hatami}). Such a result first appeared in \cite{gowers2015inverse}. In \S\ref{subsection:quant-van-kampen}, we show how to extract quantitative bounds on lengths of proofs from van Kampen diagrams, and in \S\ref{subsection:robust-self-testing}, we complete the proof of the general case. In Section \ref{sec:specific-games}, we specialize our robust self-testing theorem to the case of the Magic Square and Magic Pentagram games, establishing Theorems \ref{thm:robust-self-testing-square} and \ref{thm:robust-self-testing-pentagram}. We go on to exhibit a way to compose LCS games in parallel while controlling the growth of the complexity, proving Theorem \ref{thm:self-testing-pauli-LCS}.

\section*{Acknowledgements}

An early version of Theorem \ref{thm:informal-magic-square} used a more complicated linear constraint game. We thank William Slofstra for pointing out that the same analysis goes through for the Magic Square.

The arxiv version 1 of this paper falsely claimed that in a certain $3\times 3$ square of operators, every pair of operators sharing a row or column commute. We thank Richard Cleve, Nadish De Silva and Joel Wallman for pointing out that one pair of them did not. 
We thank Richard Cleve and Joel Wallman for sharing with the authors a proof that the magic square game mod $d$ for $d\neq 2$ is not a pseudotelepathy game. More details about this impossibility are provided in section \ref{sec:impossibility}.

We thank 
William Ballinger,
William Hoza,
Jenish Mehta,
Chinmay Nirkhe,
William Slofstra,
Thomas Vidick,
Matthew Weidner,
and
Felix Weilacher
for helpful discussions.
We thank Martino Lupini for pointing us to reference \cite{de2017operator} and Scott Aaronson for pointing us to reference \cite{cleve2004consequences}.

We thank
Arjun Bose,
Chinmay Nirkhe, 
and 
Thomas Vidick
for helpful comments on preliminary drafts of the paper.

We thank Thomas Vidick for various forms of guidance throughout the project. A.C. was supported by AFOSR YIP award number FA9550-16-1-0495.
J.S. was supported by NSF CAREER Grant CCF-1553477 and the Mellon Mays Undergraduate Fellowship. 
Part of this work was completed while J.S. was visiting UT Austin. 

\tableofcontents

\section{Preliminaries}\label{sec:preliminaries}
We assume a basic familiarity with quantum information, see e.g.\ \cite{nielsen2002quantum}. We introduce all necessary notions from the fields of nonlocal games and self-testing, but we don't reproduce all of the proofs.

\subsection{Notation}
We write $[n]$ to refer to the finite set $\set{1,\ldots, n}$ with $n$ elements. We write $[A,B]$ for $ABA\1B\1$, the group commutator of $A$ and $B$.
We use the Dirac delta notation 
\begin{equation}
	\delta_{x,y} :=\begin{cases}
		1, &\text{ if }x=y\\
		0, &\text{ otherwise}
	\end{cases}.
\end{equation}
$\mathbf H$ will refer to a hypergraph, while 
$\m H$ will refer to a Hilbert space.
$\m L(\m H)$ is the space of linear operators on the Hilbert space $\m H$. 
$\rho$ will always refer to a state on a Hilbert space, while $\s$ and $\tau$ are reserved for group representations. $\w_d:= e^{2\pi i /d}$ will always refer to the same $d\th$ root of unity. 
When we have multiple Hilbert spaces, we label them with subscripts, e.g.\ as $\m H_A, \m H_B$. In that case, we may also put subscripts on operators and states to indicate which Hilbert spaces they are associated with. 
When the Hilbert space is clear from context, $I$ refers to the identity operator on that space. 
$I_d$ will always refer to the identity operator on $\C^d$. 
$\eprp:=\frac{1}{\sqrt d}\sum_i^d\ket{ii}$ refers to the maximally entangled state on $\C^d\otimes \C^d$. We use the shorthand $\Tr_\rho(X) = \Tr X\rho$. 
We use the following notion of state-dependent distance, which we'll recall, and prove properties of, in \S\ref{subsection:state-dependent-distance}.
\begin{equation}\label{eq:state-dependent-distance-definition}
	\drho{\rho}{X}{Y} = \sqrt{\Tr_\rho(X-Y)^\dagger(X-Y)}.
\end{equation}
$\norm{X}_p$ denotes the $p$-norm of $X$, i.e.\ $\norm X_1 = \Tr\sqrt{XX^\dagger}$ and $\norm X_2 = \sqrt{\Tr XX^\dagger}$.

\subsection{Nonlocal games}

\begin{definition}[Nonlocal game]
	For our purposes, a nonlocal game $G$ is a tuple $(A,B,X,Y,V,\pi)$, where $A,B,X,Y$ are finite sets of answers and questions for Alice and Bob, $\pi:X\times Y\to [0,1]$ is a probability distribution over questions, and $V:A\times B \times X \times Y\to \pair01$ is \emph{the win condition}. 
\end{definition}

\begin{definition}[Strategies for nonlocal games]
	If $G$ is a nonlocal game, then a strategy for $G$ is a probability distribution $p:A\times B\times X\times Y\to [0,1]$.	The \emph{value} or \emph{winning probability} of a strategy is given by 
	\begin{equation}
		\w(G;p) := \sum_{a,b,x,y}\pi(x,y)p(a,b\|x,y)V(a,b,x,y).
	\end{equation}
	If the value is equal to $1$, we say that the strategy is \emph{perfect}. If the probability distribution is separable, i.e.\ $p(a,b\|x,y) = \sum \a_i p_i(a\|x)q_i(b\|y)$ for some probability distributions $\set{p_i}, \set{q_i}$, then we say that the strategy is \emph{local}.
\end{definition}
	We think of a local strategy as being implemented by using only the resource of public shared randomness. Alternatively, the local strategies are the strategies which are implementable by spacelike-separated parties in a hidden variable theory of physics.
\begin{definition}[Quantum strategies, projective measurement version]
	We say that a strategy $p:A\times B\times X\times Y\to [0,1]$ is \emph{quantum of local dimension $d$} if there exist projective measurements $\set{\set{A_x^a}_a}_x, \set{\set{B_y^b}_b}_y$ on $\C^d$ and a state $\rho \in \m L(\C^d\otimes \C^d)$ such that 
	\begin{equation}
		p(a,b\|x,y) = \Tr_\rho(A_x^a\otimes B_y^b)
	\end{equation}
	 (By \emph{projective measurement} we mean that for all $x,y,a,b$ we have $(A_x^a)^2 = A_x^a = (A_x^a)\dagg, (B_y^b)^2 = B^b_y= (B^b_y)\dagg$, and for all $x,y$, we have $\sum_a A_x^a = I = \sum_bB_y^b$.)
	 \\\noindent We say that a strategy is \emph{quantum} if it is quantum of local dimension $d$ for some $d$.
\end{definition}
We denote by $\w_*(G)$ the \emph{optimal quantum value} of $G$, i.e. the supremum over all quantum strategies of the winning probability. If the value of a strategy is $\w_*(G)$, we say that the strategy is \emph{ideal}. For quantum strategies, we use the term \emph{strategy} to refer interchangeably to the probability distribution or to the state and measurement operators producing it. 

	

	

\begin{definition}[Self-testing]
We say that a non-local game $G$ \emph{self-tests} a quantum strategy $S = (\set{\set{A_x^a}_a}_x, \set{\set{B_y^b}_b}_y, \ket{\Psi}) $ if any quantum strategy $S'$ that achieves the optimal quantum winning probability $w_*$ is equivalent up to local isometry to $S$.
\end{definition}
By \emph{local isometry} we mean a channel $\Phi: \m L(\m H_A\otimes \m H_B) \to \m L(\m H_A'\otimes \m H_B')$ which factors as $\Phi(\rho) = (V_A\otimes V_B) \rho(V_A\otimes V_B)^\dagger$, where $V_A: \m H_A\to \m H_A', V_B: \m H_B\to \m H_B'$ are isometries.

\begin{definition}[Robustness of self-tests] We say that a non-local game $G$ is $\left(\eps,\delta(\eps)\right)$-rigid if it self-tests a strategy $S = (\set{\set{A_x^a}_a}_x, \set{\set{B_y^b}_b}_y, \ket{\Psi})$, and, moreover, for any quantum strategy $\tilde{S} = \big(\set{\set{\tilde{A}_x^{a}}_a}_x, \set{\set{\tilde{B}_y^{b}}_b}_y, \rho \big)$ that achieves a winning probability of $w_*(G) - \eps$, there exists a local isometry $\Phi$ such that 
\begin{equation}
\norm {\Phi(\tilde{A}_x^{a}\otimes \tilde{B}_y^{b}\ \rho \tilde{A}_x^{a}\otimes \tilde{B}_y^{b}) - \big(A_x^a\otimes B_y^b \ket{\Psi}\bra{\Psi} A_x^a\otimes B_y^b \big)\otimes \rho_{\text{extra}}}_{2} \leq \delta(\eps)
\end{equation}
where $\rho_{\text{extra}}$ is some auxiliary state, and $\delta(\eps)$ is a function that goes to zero with $\eps$.
\end{definition}
\subsection{Groups}

We work with several groups via their presentations. For the basic definitions of group, quotient group, etc.\ see any abstract algebra text, e.g.\ \cite{dummit2004abstract}.
\begin{definition}
	Let $S$ be a set of letters. We denote by $\m F(S)$ the \emph{free group on $S$}. As a set, $\m F(S)$ consists of all finite words made from $\set{s,s\1\;s\in S}$ such that no $ss\1$ or $s\1s$ appears as a substring for any $s$. The group law is given by concatenation and cancellation.
\end{definition}

\begin{definition}[Group presentation]
	Let $S$ be finite and $R$ a finite subset of $\m F(S)$. Then $G = \Braket {S:R}$ is the \emph{finitely presented group} generated by $S$ with relations from $R$. Explicitly, 
	$G = \m F(S)/\braket R$, where $/$ is used to denote the quotient of groups, and $\braket R$ denotes the subgroup generated by $R$. 
	We say that an equation $w = w'$ is \emph{witnessed by $R$} if $w'w^{-1}$ (or some cyclic permutation thereof) is a member of $R$. 
\end{definition}
We emphasize that in this work, we sometimes distinguish between two presentations of the same group. If $G = \braket{S:R}, G' = \braket{S':R'}$ are two finitely presented groups, we reserve equality for the case $S = S'$ and $R = R'$, and in this case we'll say $G = G'$. We'll say that $G \cong G'$ if there is a group isomorphism between them.

\begin{definition}\label{definition:canonical-form}
	Let $G = \braket{S:R}$ be a finitely presented group and $\can: G\to \m F(S)$ be an injective function. We say that $\can$ is a \emph{canonical form} for $G$ if the induced map $\bar \can: G\to \m F(S)/\braket{R}$ is an isomorphism. In other words, we require that $\can(g)\can(h) = \can(gh)$ as elements of $G$, but not as strings.
\end{definition}

Now and throughout the paper, for a group $G$, we'll denote by $1$ its identity, and we'll let $[a,b] := aba\1b\1$ denote the commutator of $a$ and $b$. The group presentations of interest in this paper will take a special form extending the ``groups presented over $\Z_2$'' from \cite{slofstra2016tsirelson}. 

\begin{definition}[Group presentation over $\Z_d$]
	Let $d\in \N$ and let $\Z_d = \Braket{J:J^d}$ be the finite cyclic group of order $d$. A \emph{group presented over $\Z_d$} is a group $G = \Braket{S':R'}$, where $S'$ contains a distinguished element $J$ and $R'$ contains relations $[s,J]$ and $s^d$ for all $s\in S$.
	
	For convenience, we introduce notation that suppresses the standard generator $J$ and the standard relations.
	\begin{equation}
		G= \Braket{S:R}_{\Z_d} = \Braket{
		S\cup\set{ J } : R \cup \set{ s^d,J^d,[s,J] \; s\in S }
		}
	\end{equation}
\end{definition}
In the group representations of interest, we'll have $J \mapsto e^{2\pi i/d}$---we should always just think of $J$ as a $d\th$ root of unity. We'll think of relations of the form $J\1[a,b]$ as ``twisted commutation'' relations, since they enforce the equation $aba\1b\1 = e^{2\pi i/d}$.

\begin{example}\label{example:P_d}
	The Pauli group on one $d$-dimensional qudit can be presented as a group over $\Z_{d}$. 
	\begin{equation}
	\m P_d^{\otimes 1} = \Braket{x,z: J[x,z]}_{\Z_{d}}
	\end{equation}
\end{example}

\subsection{Group pictures}\label{subsection:van-kampen-diagrams}

Suppose we have a finitely presented group $G = \braket{S:R}$ and a word $w\in \m F(S)$ such that $w = 1$ in $G$. 
Then by definition, there is a way to prove that $w = 1$ using the relations from $R$. How complicated can such a proof get? Group pictures give us a way to deal with these proofs graphically, rather than by writing long strings of equations.
In particular, we will use group pictures to get quantitative bounds on the length of such proofs. (For a more mathematically rigorous treatment of group pictures, see \cite{slofstra2016tsirelson}. These are dual to what are usually known as van Kampen diagrams.) 

\begin{definition}[Group picture]
	Let $G = \Braket{S:R}_{\Z_d}$ be a group presented over $\Z_d$. A $G$-picture is a labeled drawing of a planar directed graph in the disk. Some vertices may lie on the boundary. The vertices that do not lie on the boundary are referred to as \emph{interior vertices}. A $G$-picture is \emph{valid} if the following conditions hold:
	\begin{itemize}
		\item Each interior vertex is labeled with a power of $J$. (We omit the identity label.)
		\item Each edge is labeled with a generator from $S$.
		\item At each interior vertex $v$, 
		the clockwise product of the edge labels (an edge labeled $s$ should be interpreted as $s$ if it is outgoing and as $s\1$ if it is ingoing) is equal to the vertex label, as witnessed by $R$. (Since the values of the labels are in the center of the group, it doesn't matter where you choose to start the word.)
	\end{itemize}
\end{definition}
Note that the validity of a $G$-picture depends on the presentation of $G$. Pictures cannot be associated directly with abstract groups.

If we collapse the boundary of the disk to a point (``the point at infinity''), then the picture becomes an embedding of a planar graph on the sphere (see Figure \ref{fig:sphere}). The following is a kind of ``Stoke's theorem'' for group pictures, which tells us that the relation encoded at the point at infinity is always valid.

\begin{figure}
	\begin{center}
	\begin{tabular}{m{0.5\textwidth}m{0.5\textwidth}}
	    \resizebox{0.5\textwidth}{!}{\includegraphics{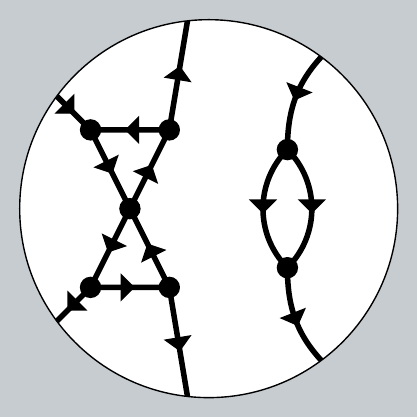}}
	    &
	    \resizebox{0.5\textwidth}{!}{\includegraphics{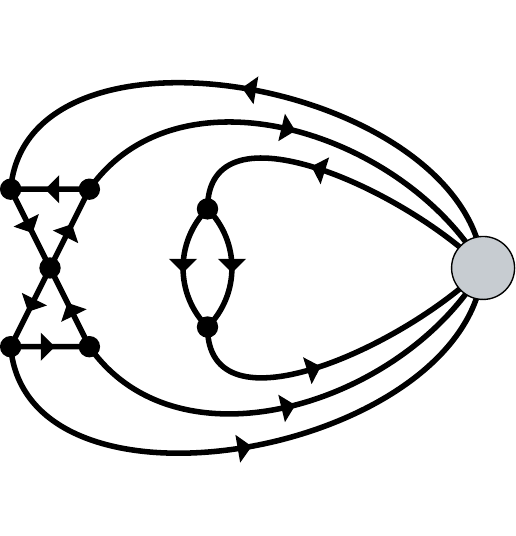}}
	\end{tabular}
    \end{center}
    \caption{This is a directed version of Figure 3 from \cite{slofstra2016tsirelson}
    .
    The interior vertices are drawn with dots, while the edge labels and the non-interior vertices are suppressed. 
    }
    \label{fig:sphere}
\end{figure}
\begin{definition}
	Suppose $\m P$ is a $G$-picture. The \emph{boundary word} $w$ is the product of the edge labels of the edges incident on the boundary of $\m P$, in clockwise order.
\end{definition}

\begin{lemma}[van Kampen]
\label{lemma:van-kampen}
	Suppose $\m P$ is a valid $G$-picture with boundary word $w$. Let $J^a$ be the product of the labels of the vertices in $\m P$. Then $w=J^a$ is a valid relation in $G$. Moreover, we say that the relation $w=J^a$ is \emph{witnessed} by the $G$-picture $\m P$.
\end{lemma}
The proof is elementary and relies on the fact that the subgroup $\braket J$ is abelian and central, so that cyclic permutations of relations are valid relations.  
By counting what goes on at each step in the induction of a proof of the above lemma, one can extract a quantitative version. This is stated and proved in \S \ref{subsection:quant-van-kampen}. 

\begin{example}
	Recall the group $\paulin 1$ from Example \ref{example:P_d}. It's easy to see that $(xz)^d = 1$ in this group. In Figure \ref{fig:xz-group-picture}, we give two proofs of this fact, for the case $d =3$. The examples are chosen to illustrate that shorter proofs are more natural than longer proofs in the group picture framework.

\begin{figure}[h]
	\begin{tabular}
	{m{0.2\textwidth}m{0.25\textwidth}m{0.25\textwidth}m{0.2\textwidth}}
		$\begin{aligned}[t]
				  &(zx)z(xz)x\\
				  &=(Jxz)z(J\1zx)x\\
				  &=x(zzz)xx\\
				  &=(xxx)\\
				  &=1
				\end{aligned}$
		& \resizebox{0.25\textwidth}{!}{\includegraphics{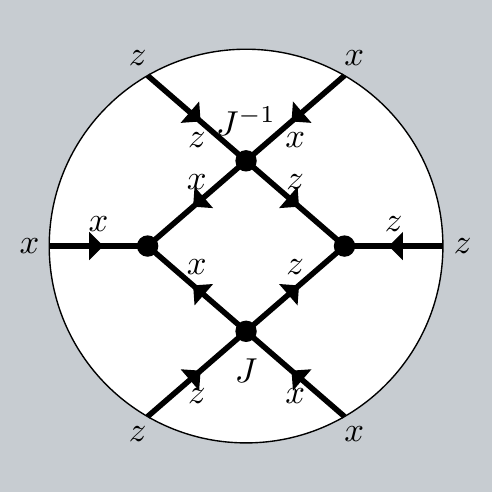}}
		& \resizebox{0.25\textwidth}{!}{\includegraphics{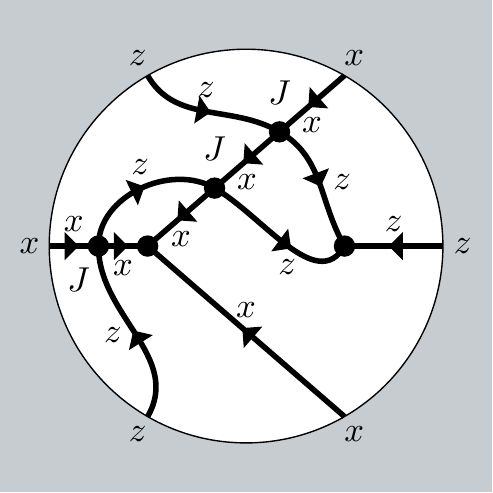}}
		&$\begin{aligned}[t]
				  &(zx)zxzx\\
				  &=Jxz(zx)zx\\
				  &=J^2x(zx)zzx\\
				  &=xx(zzz)x\\
				  &=(xxx)\\
				  &=1
				\end{aligned}$
		\end{tabular}
	\caption{The first picture uses a minimal number of relations, and corresponds (in an imprecise sense) to the equation manipulations on the left. The second picture corresponds to the equation manipulations on the right, in which each $z$ is commuted all the way to the end of the string.
}
	\label{fig:xz-group-picture}
\end{figure}
\end{example}

\subsection{Representation theory of finite groups}

We'll study groups through their representations. We collect here some basic facts about the representation theory of finite groups. For exposition and proofs, see e.g. \cite{dummit2004abstract}. Throughout, $G$ will be a finite group. It should be noted that some of these facts are not true of infinite groups. 


\begin{definition}
	A $d$-dimensional \emph{representation} of $G$ is a homomorphism from $G$ to the group of invertible linear operators on $\C^d$.
	A representation is \emph{irreducible} if it cannot be decomposed as a direct sum of two representations, each of positive dimension.
	A representation is \emph{trivial} if its image is $\set I$, where $I$ is the identity matrix.
	The \emph{character} of a representation $\s$ is the function defined by $g\mapsto \Tr(\s(g))$. 
	Two representations $\rho_1$ and $\rho_2$ are \emph{equivalent} if there is a unitary $U$ such that for all $g$, $U\rho_1(g) U^\dagger = \rho_2(g)$.
\end{definition}
Notice that a $1$-dimensional representation and its character are the same function, and that $1$-dimensional representations are always irreducible. We sometimes write ``irrep'' for ``irreducible representation.'' The next fact allows us to check equivalence of representations algebraically.

\begin{fact}
	$\rho_1$ is equivalent to $\rho_2$ iff they have the same character. 
\end{fact}

The following is immediate:
\begin{lemma}\label{fact:character-convex-combination}
	Let $\s = \bigoplus_i \s_i$ be a direct sum decomposition of $\s$ into irreducibles. Let $\circ$ denote composition of maps, and let $\chi = \Tr\circ \s, \chi_i = \Tr\circ\s_i$ be the characters corresponding to the representations $\s$. Then $\chi = \sum_i \chi_i$. 

	Furthermore, define $\tilde \chi = \frac1{\dim \s}\chi$ and $\tilde\chi_i = \frac1{\dim \s_i} \chi_i$ as the \emph{normalized characters} of $\s,\s_i$. 
	Then the normalized character of $\s$ is a convex combination of the normalized characters of $\s_i$. 
	\begin{equation}
		\tilde \chi = \sum_i \frac{\dim \s_i}{\dim \s}\tilde\chi_i.
	\end{equation}
\end{lemma}

There is a simple criterion to check whether a representation of a finite group is irreducible:
\begin{fact}\label{fact:irreducibility-criterion}
	$\s$ is an irreducible representation of $G$ iff
	\begin{equation}
		\abs G = \sum_{g\in G} \Tr \s(g) \Tr\s(g\1). 
	\end{equation}
\end{fact}

\begin{definition}
	The \emph{commutator subgroup} $[G,G]$ of $G$ is the subgroup generated by all elements of the form $[a,b] := aba\1b\1$ for $a,b\in G$. The \emph{index} $\abs{G:H}$ of a subgroup $H \leq G$ is the number of $H$-cosets in $G$. Equivalently for finite groups, the index is the quotient of the orders $\abs{G:H}= \frac{\abs G}{\abs{H}}$.
\end{definition}

\begin{fact}\label{fact:1-dim-irreps}
	$G$ has a number $\abs{G:[G,G]}$ of inequivalent $1$-dimensional irreducible representations, each of which restricts to the trivial representation on $[G,G]$.
\end{fact}

\begin{fact}\label{fact:character-dimension}
	For a finite group $G$, the size of the group is equal to the sum of the squares of the dimensions of the irreducible representations. In other words, for $R$ any set of inequivalent irreps, 
	\begin{equation}\label{eq:character-dimension}
			\abs G = \sum_{\s \in R} (\dim \s)^2 \text{ iff $R$ is maximal.}
		\end{equation}	
\end{fact}
By ``maximal'', we mean that any irreducible representation is equivalent to one from $R$. This fact can be used to check whether one has a complete classification of the irreducibles of $G$. This is a special case of the following for $x = 1$.
\begin{fact}[Second orthogonality relation for character tables]\label{fact:orthogonality}
	Let $x\in G$. Let $\s$ vary over a maximal set of inequivalent irreps of $G$, and let $n_\s$ be the dimension of $\s$. Then
	\begin{equation}
		\frac 1{\abs G}\sum_\s n_\s\Tr(\s(x)) = \delta_{x,1}.
	\end{equation}
\end{fact}

\begin{fact}[Schur's lemma]\label{fact:schur}
	Let $\tau: G\to U(\C^d)$ be an irrep and $X \in \m L(\C^d)$ be a linear operator. Suppose that $X\tau(g) = \tau(g)X$ for all $g\in G$. Then $X = \lambda I$ is a scalar multiple of identity. 
\end{fact}

\section{Linear constraint system games over $\Z_d$}\label{sec:linear-constraint-games}
	We recall several definitions from previous works of Cleve, Liu, Mittal, and Slofstra \cite{slofstra2016tsirelson,cleve2016perfect,cleve2014characterization}. Following a suggestion from \cite{cleve2016perfect}, we define the machinery over $\Z_d$ instead of $\Z_2$. 

\begin{definition}
	A \emph{hypergraph} $\mathbf H = (V,E,H)$ consists of a finite \emph{vertex set} $V$, a finite \emph{edge set} $E$ and an \emph{incidence matrix} $H: V\times E \to \Z$. 
\end{definition}
We think of $V$ as a set of $\Z$-linear equations, $E$ as a set of variables, and $H(v,e)$ as the coefficient of variable $e$ in equation $v$. Following Arkhipov \cite{arkhipov2012extending}, some of our hypergraphs of interest will be graphs. Unlike previous works, we introduce signed coefficients (outgoing edges have a positive sign in the incidence matrix, while ingoing edges have a negative sign). This is because previous works considered equations over $\Z_2$, where $1 = -1$.

\begin{definition}[\cite{cleve2014characterization}, \cite{slofstra2016tsirelson}]\label{definition:linear-constraint-game}
	Given hypergraph $\mathbf H$, vertex labelling $l: V \to \Z$, and some modulus $d\in \Z$, we can associate a nonlocal game which we'll call the \emph{linear constraint game} $\LCS(\mathbf H, l, \Z_d)$. Informally, a verifier sends one equation $x$ to Alice and one variable $y$ to Bob, demanding an assignment $a:E\to \Z_d$ to all variables from Alice and an assignment $b\in \Z_d$ to variable $y$ from Bob. The verifier checks that Alice's assignment satisfies equation $x \pmod d$, and that Alice and Bob gave the same assignment to variable $y$.
	\\\noindent
	Formally, we have the following question and answer sets: $X = V$, $Y = E$, $A = \Z_d^E$, $B = \Z_d$. The win condition selects those tuples $(a,b,x,y)$ satisfying:
	\begin{align}
	\label{eq:lcs-equation}
		a(y) &= b &&\text{(Consistency) }
	\\	\sum_{e\in E} H(x,e)a(e) &\equiv l(x) \pmod d.	&&\text{(Constraint satisfaction) }
	\end{align}
\end{definition}
We introduce the two primary LCS games of interest in this paper.
\begin{example}
\label{example:magic-square}
	The magic square LCS (mod $d$) has vertex set $\set{v_1,\ldots, v_6}$, edge set $\set{e_1,\ldots, e_9}$, vertex labeling $l(v_5) = 1, l(v_i) = 0$ for $i\neq 5$. See Figure \ref{fig:magic-square-equations} for the full description of the hypergraph and the associated set of linear equations.
	
\end{example}
\begin{figure}[h]
	\begin{center}
	\begin{multicols}{2}
	    \resizebox{0.2\textwidth}{!}{\includegraphics{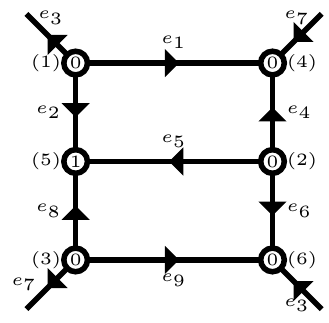}}
	    
	    \begin{align}
	    	(1)&& e_1 + e_2 + e_3 &= 0 & (4)&& -(e_1 + e_4 + e_7) &= 0 
	    \\	(2)&& e_4 + e_5 + e_6 &= 0 & (5)&& -(e_2 + e_5 + e_8) &= 1 
	    \\	(3)&& e_7 + e_8 + e_9 &= 0 & (6)&& -(e_3 + e_6 + e_9) &= 0 
	    \end{align}
	\end{multicols}
    \end{center}
    \caption{The magic square LCS, presented both in terms of equations (mod $d$) and in terms of a labelled hypergraph. The two line segments labeled $e_3$ are parts of the same edge, as are the pair of line segments labeled $e_7$. The underlying graph is $K_{3,3}$, the smallest bipartite non-planar graph. The direction of the edges emphasizes the bipartition.}
    \label{fig:magic-square-equations}
\end{figure}
\begin{example}
\label{example:magic-pentagram}
	The magic pentagram LCS (mod $2$) has vertex set $\set{v_1,\ldots, v_5}$, edge set $\set{e_1,\ldots, e_{10}}$, vertex labeling $l(v_5) = 1, l(v_i) = 0$ for $i\neq 5$. See Figure \ref{fig:magic-pentagram-equations} for the full description of the hypergraph and the associated set of linear equations.
\end{example}
\begin{figure}[h]
	\begin{center}
	\begin{multicols}{2}
	    \resizebox{0.3\textwidth}{!}{\includegraphics{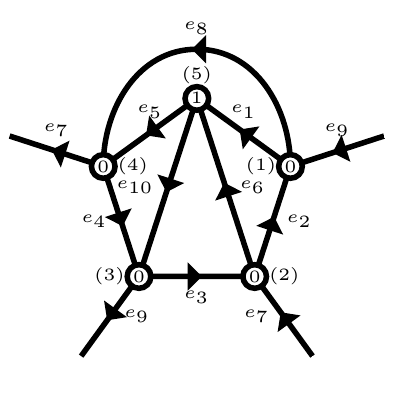}}

	    \begin{align}
	    (1)&& e_1 - e_2 + e_8 - e_9 	&= 0\\
	    (2)&& e_2 - e_3 + e_6 - e_7 	&= 0\\
	    (3)&& e_3 - e_4 + e_9 - e_{10} 	&= 0\\
	    (4)&& e_4 - e_5 + e_7 - e_{8} 	&= 0\\
	    (5)&& e_5 - e_6 + e_{10} - e_1 	&= 1\\
	    \end{align}
	\end{multicols}
    \end{center}
    \caption{The magic pentagram LCS, presented both in terms of equations (mod $2$) and in terms of a labelled hypergraph. The two line segments labeled $e_7$ are parts of the same edge, as are the pair of line segments labeled $e_9$. The underlying graph is $K_5$, the smallest complete non-planar graph.} 
    \label{fig:magic-pentagram-equations}
\end{figure}

The following is the main tool we use to understand linear constraint system games.

\begin{definition}[Solution group over $\Z_d$, \cite{cleve2016perfect}]
	For an LCS game $\LCS(\mathbf H, l, \Z_d)$ with $\mathbf H = (V,E,H)$, the \emph{solution group} $\G(\mathbf H, l, \Z_d)$ has one generator for each edge of $\mathbf H$ (i.e.\ for each variable of the linear system), one relation for each vertex of $\mathbf H$ (i.e.\ for each equation of the linear system), and relations enforcing that the variables in each equation commute. Formally, define the sets of relations $R_c$, the local commutativity relations, and $R_{eq}$, the constraint satisfaction relations as
	\begin{align}
		R_{c} 		&:= \set{[e,e'] \; H(v,e) \neq 0 \neq H(v,e')\text{ for some } v\in V}
		\\R_{eq}	&:= \set{ J^{-l(v)}\prod_{e\in E}e^{H(v,e)} \; v\in V}.
	\end{align}
Then define the solution group as
\begin{equation}
	\G(\mathbf H, l, \Z_d) := \Braket{E:R_{c}\cup R_{eq}}_{\Z_d}.
\end{equation}
(Notice that the order of the products defining $R_{eq}$ is irrelevant, since each pair of variables appearing in the same $R_{eq}$ relation also have a commutation relation in $R_c$.)
\end{definition}
When the LCS game is clear from context, we'll just write $\G$ to denote its solution group.

Our aim is to prove that for some specific linear constraint system games, strategies that win with high probability are very close to some ideal form.
We start by observing that for any LCS game, \emph{any} strategy already has a slightly special form.

\begin{lemma}[Strategies presented via observables]\label{lemma:unitary-observable-strategy-exact}
	Suppose that $p(a,b\|v,e) =\Tr_\rho \tilde A_v^a\otimes \tilde B_e^b$ is a quantum strategy for an LCS game over $\Z_d$ with hypergraph $\mathbf H = (H,V,E)$. Then there are unitaries $\set{A_e^{(v)} \; e\in E, v\in V}$ and $\set{B_e \; e\in E}$ such that for all $v,e$, $(A_e^{(v)})^d = I = B_e^d$; for any fixed $v$, the $A_e^{(v)}$ pairwise commute; moreover, the provers win with probability $1$ iff 
	\begin{equation}\label{eq:exact-win-criterion-con}
		\text{for all $v,e$, }\Tr_\rho A_e^{(v)}\otimes B_e  = 1\text{, and}
	\end{equation}
	\begin{equation}\label{eq:exact-win-criterion-sat}
		\text{for all }v\text{, }\Tr_\rho \prod_{e}\left(A_e^{(v)}\right)^{H(v,e)}\otimes I_B = \w_d^{l(v)}.
	\end{equation}
\end{lemma}

We refer to the operators $\set{A_e^{(v)}}, \set{B_e}$ together with the state $\rho$ as a \emph{strategy presented via observables.}
Typically the word ``observable'' is reserved for Hermitian operators. Nonetheless, we call our operators observables because they capture properties of the projective measurements from which they're built in a useful way. Operationally, we think of Bob as measuring the observable $B_e$ and reporting the outcome when asked about variable $e$ and of Alice measuring the observables $A_e^{(v)}$ and reporting the outcome for each $e$ when asked about equation $v$. The fact that Alice's observables pairwise commute at each equation means that Alice can measure them simultaneously without ambiguity.

A version of this lemma is proved in the course of the proof of Theorem 1 of \cite{cleve2014characterization}. We give essentially the same proof, just over $\Z_d$.
\begin{proof}[Proof of Lemma \ref{lemma:unitary-observable-strategy-exact}]
	Define the observables as
	\begin{align}
		B_e:= \sum_j\w_d^{-j}\tilde B^{i}_{e}
		&& 
		A_e^{(v)}:= \sum_i \w_d^i\sum_{a: a(e) =i} \tilde A_{v}^a.
	\end{align}

	It's clear that each of these operators is a unitary whose eigenvalues are $d\th$ roots of unity. To see that $A_e^{(v)}$ commutes with $A_{e'}^{(v)}$, notice that they are different linear combinations of the same set of projectors. Now we compute, for any $v,e$,
	\begin{align}
		\Tr_\rho A_e^{(v)}\otimes B_e 
			&= \sum_{i,j}\w_d^{i-j}\Tr_\rho \left(\sum_{a: a(e) =i} \tilde A_{v}^a\right)\otimes \tilde B^{j}_{e} 
		\\	&= \sum_k \w_d^k\Pr[a(e) - b \equiv k\mid \text{questions } x=v, y=e].
	\end{align}
	Notice that the last line is a convex combination of the $d\th$ roots of unity. Hence, it equals $1$ if and only if $\Pr[a(e) \equiv b \mid \text{questions } x=v, y=e] = 1$.

	A similar computation reveals:
	\begin{align}
		 &\w_d^{-l(v)} \Tr_\rho \prod_{e} \left(A_e^{(v)}\right)^{H(v,e)}\otimes I \\
		= & \sum_k\w_d^{k-l(v)} \Tr_\rho \sum_{\substack{a\\ \sum_e H(v,e)a(e)\equiv k}} \tilde A_{v}^a\otimes I \\
		= & \sum_k \w_d^{k-l(v)}\Pr\left[\sum_{e}H(v,e)a(e)\equiv k\middle|  \text{question } x=v\right]
	\end{align}
Again, the last line is a convex combination of the $d\th$ roots of unity. Hence it equals $1$ if and only if $\Pr\left[\sum_{e}H(v,e)a(e)\equiv l(v) \middle|  \text{question } x=v\right] = 1$. 
\end{proof}

Note that we can always recover the original strategy in terms of projective measurements by looking at the eigenspaces of the observables. Therefore, we restrict our attention to strategies presented by observables without loss of generality. 

Next, we state a simple sufficient condition for the existence of a perfect quantum strategy for an LCS game.

\begin{definition}[Operator solution]
	An \emph{operator solution} for the game $\LCS(\mathbf H, l, \Z_d)$ is a unitary representation $\s$ of the group $\G(\mathbf H, l, \Z_d)$ such that $\s(J) = \w_dI$. A \emph{conjugate operator solution} is a unitary representation sending $J\mapsto \bar{\w_d}I$. 
\end{definition}

Notice that if $\s$ is an operator solution, then for any choice of basis the complex conjugate $\bar\s: g\mapsto \bar{\s(g)}$ is a conjugate operator solution. The existence of an operator solution is sufficient to construct a perfect quantum strategy. 

\begin{example}[Operator solution for magic square]
See the square of group generators in Figure \ref{fig:magic-game-generators}. Let $\G_2$ be the solution group of the Magic Square. Consider the map $\G_2 \to U(\C^d\otimes \C^d)$ generated by sending each generator in this square to the operator in the corresponding location of Figure \ref{fig:magic-square-operators-introduction}. This map is an operator solution.
\end{example}

\begin{figure}[h]
	\begin{tabular}{m{0.4\textwidth}m{0.1\textwidth}m{0.4\textwidth}}
	    \resizebox{0.3\textheight}{!}{\includegraphics{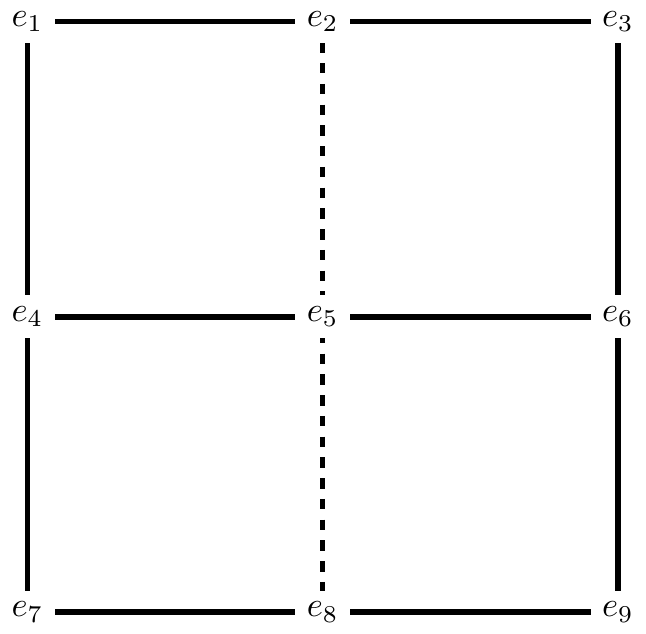}}
	    &
	    &\resizebox{0.3\textheight}{!}{\includegraphics{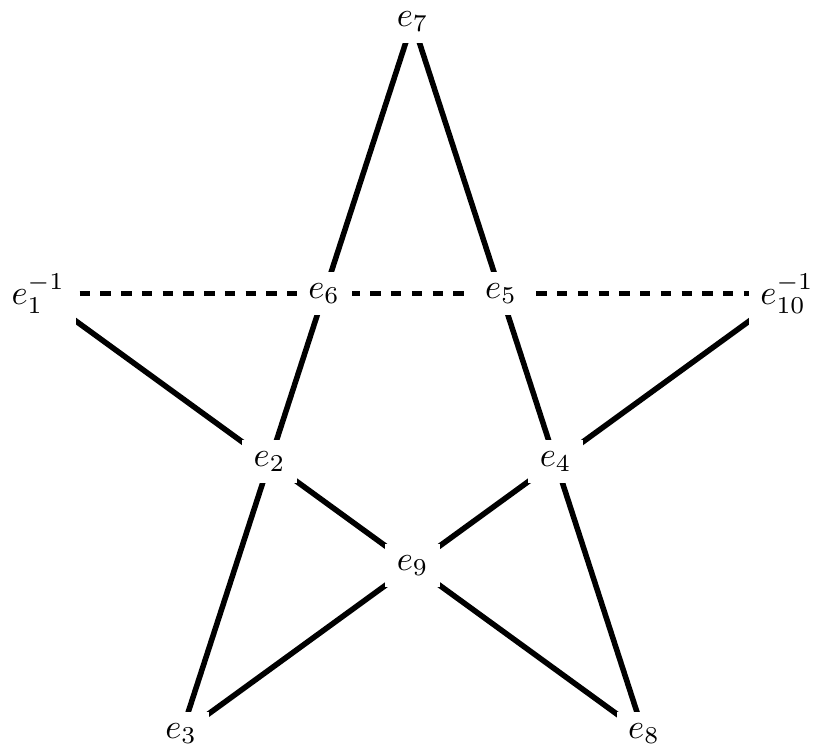}}
   	\end{tabular}
\caption{On the left-hand figure, the product of the generators on any solid line is equal to $1$ in the solution group of the magic square. The product of the operators on the dashed line is equal to $J$. Similarly, on the right-hand figure, the alternating product $ab\1cd\1$ is equal to $1$ on the solid lines and $J$ on the dashed line.}
	\label{fig:magic-game-generators}
\end{figure}

\begin{example}[Operator solution for magic pentagram]
See the pentagram of group generators in Figure \ref{fig:magic-game-generators}. Let $\G_3$ be the solution group of the Magic Pentagram. Consider the map $\G_3 \to U(\C^d\otimes \C^d\otimes\C^d)$ generated by sending each generator in this pentagram to the operator in the corresponding location of Figure \ref{fig:magic-square-operators-introduction}. This map is an operator solution.
\end{example}

\begin{prop}
\label{prop: perfect strategy}
	Let $\s: \G \to U(\C^D)$ be an operator solution. Define a strategy by setting $\ket\psi = \ket{\r{EPR}_D}$, $A_e^{(v)} = \s(e)$ for all $e,v,$ and $B_e = \bar{\s(e)}$ for all $e$. Provers using this strategy win with probability $1$.
\end{prop}
\begin{proof}
	By a well-known property of the maximally entangled state, we have
	\begin{equation}
		\braket{ \psi | \s(e) \otimes \bar{\s(e)} | \psi} = 
		\braket{ \psi | \s(e)\bar{\s(e)}^T \otimes I  | \psi} = 1,
	\end{equation}
	where $^T$ denotes the transpose. 
	Therefore, the consistency criterion \eqref{eq:exact-win-criterion-con} is satisfied.
	Since $\s$ is an operator solution, we have 
	\begin{align}
		\prod_{e} \left(A_e^{(v)}\right)^{H(v,e)} 
		\\&= \s(\prod_e \s(e)^{H(v,e)})
		\\&= \s(J^{l(v)}) 
		\\&=\w_d^{l(v)}I,
	\end{align}
	so the constraint satisfaction criterion \eqref{eq:exact-win-criterion-sat} is satisfied.
\end{proof}

We'll see both exact and approximate converses to this proposition in Section \ref{sec:self-testing}.

\section{General self-testing}\label{sec:self-testing}

In this section, we introduce our main robust self-testing theorem for linear constraint system games with solution groups of a certain form. In \S \ref{subsection:exact-self-testing}, to ease understanding, we start by stating and proving an exact version of the theorem. In \S \ref{subsection:state-dependent-distance} through \S \ref{subsection:quant-stabilizer-state-bounds}, we introduce the necessary tools to prove an approximate version of the self-testing theorem. \S \ref{subsection:state-dependent-distance} introduces the state-dependent distance and some of its properties. \S \ref{subsection:stability-lemma} proves a stability lemma for representations of finite groups, which allows us to deduce that the action of a strategy winning with high probability is close to the action of a representation of the solution group.  \S \ref{subsection:quant-van-kampen} presents a quantitative version of the van Kampen Lemma from Section \S \ref{subsection:van-kampen-diagrams}, which is key in bounding the robustness of the main theorem. \S \ref{subsection:quant-stabilizer-state-bounds} shows that if a joint state is approximately stabilized by the action of the Pauli group on two tensor factors, then it is close to the maximally entangled state on the two tensor factors. In \S \ref{subsection:robust-self-testing} we combine these tools to prove our robust self-testing theorem. 
 
\subsection{Exact self-testing}
\label{subsection:exact-self-testing}

Throughout, let $\LCS(\mathbf H, l, \Z_d), \mathbf H = (V,E,H)$ be an LCS game with solution group $\G$.

\begin{thm}[Rigid self-testing of observables]\label{thm:rigid-self-testing-observables}
Suppose $\G$ is finite and all of its irreducible representations with $J\mapsto \w_d I$ are equivalent to a fixed irrep $\s:\G\to U(\C^d)$. Suppose $\set{A_e^{(v)}}, \set{B_e}, \rho\in \m L(\m H_A\otimes \m H_B)$ is a perfect strategy presented via observables for the game. 
Then there are local isometries $V_A, V_B$ such that
	\begin{itemize}
		\item for all $e,v$, $V_AA_e^{(v)}V_A^\dagger = \s(e)\otimes I \oplus \hat A^{(v)}_e$, where $\hat A^{(v)}_e V_A\rho V_A^\dagger = 0$, and 
		\item for all $e$ $V_BB_eV_B^\dagger =\bar{\s(e)}\otimes I \oplus \hat B_e$, where $\hat B_e V_B\rho V_B^\dagger = 0$.
	\end{itemize}
\end{thm}
Awkwardly, we must pick a basis to take the complex conjugate in. Fortunately, we only care about our operators up to isometry. So to make sense of the theorem statement, we pick the basis for complex conjugation first, and then the isometry $V_B$ depends on this choice.
 We break the proof into two lemmas.
\begin{lemma}\label{lem:unique-operator-solution}
	Suppose $\G$ is finite and all of its irreducible representations with $J\mapsto \w_dI$ are equivalent to a fixed irrep $\s:\G\to U(\C^d)$. Then every operator solution is equivalent to $\s\otimes I$ and every conjugate operator solution is equivalent to $\bar{\s}\otimes I$, where the complex conjugate can be taken in any basis. 
\end{lemma}

\begin{lemma}[Adapted from Lemma 8, \cite{cleve2016perfect}]\label{lemma:CLS}
Suppose $\set{A_e^{(v)}}, \set{B_e}, \rho\in \m L(\m H_A\otimes \m H_B)$ is a perfect strategy presented via observables for the game. 
	Then, there are orthogonal projections $P_A, P_B$ such that 
	\begin{enumerate}
	\item $(P_A\otimes P_B)\rho(P_A\otimes P_B) = \rho$;
	\item for each $e$, $P_AA_e^{(v)}P_A = P_AA_e^{(v')}P_A$, provided that $H(v,e) \neq 0 \neq H(v',e)$ (we now write $P_AA_eP_A$ without ambiguity);
	\item the map $\s_A:\G\to \ran P_A$ generated by $e\mapsto P_AA_eP_A$ (and $j \mapsto \w_dI$) is an operator solution;
	\item the map $\s_B:\G\to \ran P_B$ generated by $e\mapsto P_BB_eP_B$ (and $j \mapsto \bar{\w_d}I$) is a conjugate operator solution.
	\end{enumerate}
\end{lemma}

\begin{proof}[Proof of Theorem \ref{thm:rigid-self-testing-observables}, assuming the lemmas]

	Take the maps $\s_A$ and $\s_B$ from Lemma \ref{lemma:CLS}; note that their ranges are the subspaces determined by $P_A, P_B$. From Lemma \ref{lem:unique-operator-solution} we get partial isometries $W_A$, $W_B$ such that $W_A\s_A(e)W_A^\dagger = \s(e)\otimes I$ and $W_B\s_B(e)W_B^\dagger = \bar{\s(e)}\otimes I$. To complete the proof, let $V_A$ and $V_B$ be any isometric extensions of $W_A$ and $W_B$, and set 
	$\hat A_e^{(v)} = V_A(I-P_A)A_e^{(v)}(I-P_A)V_A^\dagger, \hat B_e = V_B(I-P_B)B_e(I-P_B)V_B^\dagger$. Checking that these operators satisfy the equations in the theorem is a simple computation. 
\end{proof}

\begin{proof}[Proof of Lemma \ref{lem:unique-operator-solution}]
	Let $\tau$ be an operator solution, i.e.\ a representation of $\G$ with $\tau(J) = \w_d I$. Let $\tau = \oplus_{i=1}^k\tau_i$ be a decomposition of $\tau$ into $k$ irreducibles. 
	As in Lemma \ref{fact:character-convex-combination}, let $\tilde \chi: g\mapsto \frac1{\dim \tau}\Tr\tau(g)$ be the normalized character of $\tau$ and $\tilde \chi_i$ be the same for $\tau_i$. One can check that $\abs{\tilde\chi_i(g)} \leq 1$ for all $g\in \G$. Furthermore, $\tilde\chi(g)$ is a convex combination of the $\chi_i(g)$. Therefore, $\tilde\chi_i(J) = \w_d$ for each $i$. Then also $\tau_i(J) = \w_dI$ for each $i$, since this the only $d$-dimensional unitary with trace $d\w_d$. We conclude that  $\tau$ is equivalent to $\bigoplus_{i=1}^k\s = \s\otimes I_k$. 

	Now suppose that $\tau'$ is a conjugate operator solution. Then taking the complex conjugate in any basis, $\bar{\tau'}$ is an operator solution. By the above, $\bar{\tau'}$ is equivalent to $\s\otimes I$. Therefore, $\tau'$ is equivalent to $\bar{\s}\otimes I$.
\end{proof}

\begin{proof}[Proof of Lemma \ref{lemma:CLS}]
	This is essentially the same proof as given in \cite{cleve2016perfect} (their treatment is a bit more complicated since they wish to cover the infinite-dimensional case).

	Let $\m A$ be the set of finite products of unitaries from $\set{A_e^{(v)}}$, and similarly let $\m B$ be the set of finite products of unitaries from $\set{B_e}$. 
	Let $\rho_A = \Tr_B \rho$ and $\rho_B = \Tr_A\rho$. Define 
	\begin{equation}
		\hat{\m H_A} = \supp\rho_A\text{, and }
		\hat{\m H_B} = \supp\rho_B\text,
	\end{equation}
	and let $P_A$ and $P_B$ be the projectors onto these spaces. Notice that $(P_A\otimes P_B) \rho (P_A\otimes P_B) = \rho$. From the consistency criterion \eqref{eq:exact-win-criterion-con}, we have
	\begin{equation}\label{eq:CLS-conjugate}
		1 	= \Tr_\rho A_e^{(v)}\otimes B_e,\text{ so } A_e^{(v)}\ket\phi = B_e^\dagger\ket\phi\text{ for }\ket\phi\in\supp\rho.
	\end{equation}
	Let $A\in \m A$ be arbitrary. Then, the above implies that there is $B\in \m B$ be such that $(A\otimes I) \rho (A^\dagger\otimes I) = (I\otimes B^\dagger)\rho(I\otimes B)$. We compute
	\begin{align}
		A\rho_AA^\dagger
		= \Tr_B (A\otimes I)\rho(A^\dagger\otimes I)
		= \Tr_B (I\otimes B^\dagger)\rho(I\otimes B)
		= \Tr_B \rho
		=\rho_A,
	\end{align}
	from which we conclude that $\m A$ fixes $\hat{\m H_A}$. This implies that $(PA_1P)(PA_2P) = PA_1A_2P$ for $A_1,A_2\in \m A$. 
	Next, we compute
	\begin{equation}
		1 = \Tr_\rho A_e^{(v)}(A_e^{(v')})^\dagger \otimes I
		= \Tr_{\rho_A} A_e^{(v)}(A_e^{(v')})^\dagger,
	\end{equation}
	from which we conclude that $P_AA_e^{(v)}P_A=P_AA_e^{(v')}P_A$. We now write $P_AA_eP_A$ without ambiguity. Finally, we compute
	\begin{align}
		1 = \Tr_\rho w_d^{-l(v)}\prod_{e:H(v,e)\neq 0} A_e\otimes I = \Tr_{\rho_A} w_p^{-l(v)}\prod_{e:H(v,e)\neq 0} A_e\otimes I,
	\end{align}
	from which we conclude that the map $e\mapsto P_AA_eP_A$ is an operator solution. The same argument shows that $e\mapsto P_BB_eP_B$ is a conjugate operator solution. (The conjugation comes from equation \eqref{eq:CLS-conjugate}.)

\end{proof}
Here we constructed representations directly, projecting onto the support of a known state. In the approximate case, this work will be subsumed by an application of the stability lemma \ref{lemma:vidick-gowers-hatami}.

\subsection{State-dependent distance}\label{subsection:state-dependent-distance}
We now begin to collect the necessary tools to generalize the previous subsection to the approximate case. To start, we need a convenient calculus for manipulating our notion of state-dependent distance. Recall the definition of $\drho\rho\cdot\cdot$ as 
\begin{equation}
\drho{\rho}{X}{Y} = \sqrt{\Tr_\rho(X-Y)^\dagger(X-Y)}
\end{equation} 
We use the same notation as the Kullback-–Leibler divergence despite the fact that our $\drho\rho\cdot\cdot$ is symmetric in its arguments. We do this because we will write complicated expressions in the place of $X$ and $Y$; the notation becomes harder to parse if the symbol $\|$ is replaced by a comma.
Notice that if $\rho_{AB}$ is the maximally entangled pure state, then $\drho{\rho_{AB}}{X\otimes I_B}{Y\otimes I_B}$ is exactly the usual $2$-norm distance $\norm{X-Y}_2$. Much like the fidelity of quantum states, the squared distance $\drho\rho\cdot\cdot^2$ is often more natural than the distance. We collect computationally useful properties of $\drho\rho\cdot\cdot$ in the following lemma.

\begin{lemma}\label{lemma:state-dependent-distance}
	Let $\m H = \m H_A\otimes \m H_B$ be a Hilbert space. Let $U,U_i$ be unitary operators on $\m H$. Let $Z,Z_i$ be arbitrary operators on $\m H$. Similarly, let $A_i, B_i$ be unitary operators on $\m H_A, \m H_B$ respectively. Let $X_i, Y_i$ be arbitrary operators on $\m H_A, \m H_B$, respectively.
	Let $\rho$ be a state on $\m H_A\otimes \m H_B$. Let $V: \m H\to \m H'$ be an isometry and $U'$ a unitary operator on $\m H'$. Then
	\begin{enumerate} [(a)]
		\item 
		\label{item:state-dependent-distance-square} 
		$\drho{\rho}{U}{I}^2 =2 - 2\Re\Tr_\rho U$. More generally, $\drho{\rho} ZI = 1 + \Tr_\rho Z\dagg Z - 2\Re\Tr_\rho Z$.
		\item
		\label{item:state-dependent-distance-inverse} 
		$\drho{\rho}{UZ}{I} = \drho{\rho}{Z}{U^\dagger}$. In particular, $\drho{\rho}{U}{I}=\drho{\rho}{U^\dagger}{I}$.
		\item\label{item:state-dependent-distance-triangle}
		$\drho{\rho}{Z_1}{Z_3} \leq \drho{\rho}{Z_1}{Z_2} + \drho{\rho}{Z_2}{Z_3}$.
		\item\label{item:state-dependent-distance-right-multiplication} $\drho{\rho}{ZU_2}{U_3} \leq \drho{\rho}{Z}{I} + \drho{\rho}{U_2}{U_3}$. If $U_2$ commutes with $U_3$ (in particular if $U_3 = I$), then also $\drho{\rho}{U_1U_2}{U_3} \leq \drho{\rho}{U_2}{I} + \drho{\rho}{U_1}{U_3}$.
		\item\label{item:state-dependent-distance-chaining} $\drho{\rho}{\prod_{i}{A_i\otimes I_B}}{\prod_{i}I_A \otimes B_i} \leq \sum_i \drho{\rho}{A_i\otimes I_B}{ I_A\otimes B_i}$.
		\item\label{item:state-dependent-distance-conjugation} If $\drho{\rho}{I_A\otimes WB}{I} \leq \nu$ and $\drho{\rho}{A\otimes B}{I}\leq \eta$, then \mbox{$\drho{\rho}{I_A\otimes BW}{I}\leq \nu + 2\eta$.}
		\item\label{item:state-dependent-distance-jensen}
		$\drho{\rho}{\E{i}{ U_i}} {I} \leq \E{i}\drho{\rho}{U_i}{ I}$.
		\item \label{item:state-dependent-distance-partial-trace}
		$\drho{\rho}{A\otimes I_B}{I_{AB}} = \drho{\rho_A}{A}{I_A}$, where $\rho_A = \Tr_B \rho$.
		\item \label{item:state-dependent-distance-isometry}
		$\drho{\rho}{Z_1}{Z_2} = \drho{V\rho V^\dagger}{VZ_1V^\dagger}{VZ_2V^\dagger}$. 
		\item\label{item:state-dependent-distance-projection-is-identity} 
		If $P$ is a projection such that $P\rho = P$, then $\drho{\rho}{XP}{I} = \drho{\rho}{X}{I} = \drho\rho XP$.
		\item \label{item:state-dependent-distance-isometry-switching} $\drho\rho U{V^\dagger U' V} = \drho{V\rho V\dagg}{VUV\dagg}{U'}$.
	\end{enumerate}
\end{lemma}
We'll use \eqref{item:state-dependent-distance-conjugation} and \eqref{item:state-dependent-distance-right-multiplication} to convert proofs of group relations into proofs of approximate relations between operators which try to represent the group.

The reader interested in following the $\drho\rho\cdot\cdot$ computations in the rest of the paper may find it useful to find their own proofs of the preceeding facts. For completeness, we provide detailed arguments in the sequel.
\begin{proof}~
\begin{enumerate}[(a)]
	\item We complete the square.
	\begin{align}
		\drho{\rho}{X}{I}^2 
		&=
		\Tr_\rho(X-I)^\dagger(X-I)
		\\&= \Tr_\rho(2I - X - X^\dagger)
		\\&= 2 - (\Tr_\rho X + \bar{\Tr_\rho X})
		\\&= 2 - 2\Re\Tr_\rho X.
	\end{align}
	\item In the second equality, we use that the map $\rho \mapsto X^\dagger\rho X$ is trace-preserving.
	\begin{align}
		\drho{\rho}{XY}{I}^2
		&= \Tr (XY - I)\rho (XY - I)^\dagger
		\\&= \Tr (Y - X^\dagger)\rho (Y - X^\dagger)^\dagger
		\\&= \drho{\rho}{Y}{X^\dagger}^2.
	\end{align}
	\item First, suppose $\rho = \proj \psi$ is pure. Then $\drho{\rho}{Z_1}{Z_3} = \norm{Z_1\ket \psi - Z_3\ket\psi}$ and the triangle inequality for the Hilbert space norm applies. Next, notice that $\drho{\rho}{Z_1}{Z_3}^2$ is linear in $\rho$. 

	Let $\rho = \sum_i \a_i\proj i$ be a convex combination of pairwise orthogonal pure states. Then we apply linearity and Cauchy-Schwarz:	
	\begin{align}
		\drho{\rho}{{Z_1}}{Z_3}^2 
		&= \sum_i \a_i \drho{\proj i}{Z_1}Z_3^2
		\\&\leq \sum_i \a_i \left[\drho{\proj i}{Z_1}{Z_2}^2+\drho{\proj i}{Z_2}Z_3^2+2\drho{\proj i}{Z_1}{Z_2}\drho{\proj i}{Z_2}Z_3\right]
		\\&= \drho{\rho}{{Z_1}}{{Z_2}}^2+\drho{\rho}{{Z_2}}{Z_3}^2
		+2\sum_i 
		\sqrt{\a_i\Braket{i|({Z_1}-{Z_2})^\dagger({Z_1}-{Z_2})|i}}
		\sqrt{\a_i\Braket{i|({Z_2}-Z_3)^\dagger({Z_2}-Z_3)|i}}
		\\&\leq \drho{\rho}{{Z_1}}{{Z_2}}^2+\drho{\rho}{{Z_2}}{Z_3}^2
		+2\sqrt{\sum_i \a_i
		\Braket{i|({Z_1}-{Z_2})^\dagger({Z_1}-{Z_2})|i}
		\sum_j\a_j
		\Braket{j|({Z_2}-Z_3)^\dagger({Z_2}-Z_3)|j}}
		\\&= \drho{\rho}{{Z_1}}{{Z_2}}^2+\drho{\rho}{{Z_2}}{Z_3}^2
		+2\drho{\rho}{{Z_1}}{{Z_2}}\drho{\rho}{{Z_2}}{Z_3} 
		\\&= \left(\drho{\rho}{{Z_1}}{{Z_2}} + \drho{\rho}{{Z_2}}{Z_3}\right)^2.
	\end{align}
	\item Applying \eqref{item:state-dependent-distance-inverse} and \eqref{item:state-dependent-distance-triangle},
	\begin{align}
		\drho{\rho}{XY}{Z}
		&= \drho{\rho}{XYZ^\dagger}{I}
		\\&= \drho{\rho}{X}{ZY^\dagger}
		\\&\leq \drho{\rho}{X}{I} + \drho{\rho}{I}{ ZY^\dagger}
		\\&= \drho{\rho}{X}{I} + \drho{\rho}{Y}{ Z}.
	\end{align}
	If $Y$ commutes with $Z$, then we have
	\begin{align}
		\drho{\rho}{XY}{Z}
		&= \drho{\rho}{XYZ^\dagger}{I}
		\\&= \drho{\rho}{XZ^\dagger Y}{I}
		\\&= \drho{\rho}{Y}{ZX^\dagger}
		\\&\leq \drho{\rho}{Y}{I} + \drho{\rho}{X}{ Z}.
	\end{align}
	\item We apply \eqref{item:state-dependent-distance-inverse} and then apply \eqref{item:state-dependent-distance-right-multiplication} once for each $i$.
	\begin{align}
		\drho{\rho}{\prod_{i}A_i\otimes I_B}{\prod_{i}I_A\otimes B_i} 
		&= \drho{\rho}{\prod_{i}A_i\otimes B_i^\dagger}{I}
		\\&\leq \sum_i \drho{\rho}{A_i\otimes B_i^\dagger}{I}
		\\&= \sum_i \drho{\rho}{A_i\otimes I_B}{I_A \otimes B_i}.
	\end{align}
	\item This follows from \eqref{item:state-dependent-distance-inverse} and \eqref{item:state-dependent-distance-chaining} by writing $I_A \otimes BW = (A)(I_A)(A^\dagger)\otimes (B)(WB)(B^\dagger)$.
	\item By linearity and \eqref{item:state-dependent-distance-square}, we have
	\begin{align}
		\drho{\rho}{\E{i}{U_i}}{I}^2 
		&= 2 - 2\Re\Tr_\rho\left[\E{i} U_i\right]
		\\&= \E{i}\left[ 2 - 2\Re\Tr_\rho U_i\right]
		\\&= \E{i} \drho{\rho}{U_i}{I}^2.
	\end{align}
	Then Jensen's inequality completes the proof. 
	\item We use that the trace of the partial trace is the trace.
	\begin{align}
		\drho{\rho}{A\otimes I_B}{I}^2 
		&= 2 - 2\Re\Tr_\rho A\otimes I_B 
		\\&= 2 - 2\Re\Tr_{\rho_A}A 
		\\&= \drho{\rho_A}{A}{I}^2.
	\end{align}
	\item We apply cyclicity of trace and unitarity, i.e.\ $V\dagg V = I$. 
	\begin{align}
		\drho{V\rho V\dagg}{VZ_1V\dagg}{VZ_2V\dagg}^2
		&= \Tr V(Z_1-Z_2)\dagg V\dagg V(Z_1-Z_2)V\dagg V\rho V\dagg
		\\&= \Tr V(Z_1-Z_2)\dagg (Z_1-Z_2)\rho V\dagg
		\\&= \Tr (Z_1-Z_2)\dagg (Z_1-Z_2)\rho.
	\end{align}
	\item Again, we apply cyclicity of trace.
	\begin{align}
		\drho\rho{XP}I^2
		&= \Tr_\rho (XP- I)\dagg(XP-I)
		\\&= \Tr (PX\dagg- I)(XP-I)\rho
		\\&= \Tr (PX\dagg- I)(X-I)\rho
		\\&= \Tr \rho(PX\dagg- I)(X-I)
		\\&= \Tr_\rho(X\dagg- I)(X-I)
		\\&= \drho\rho XI^2.
	\end{align}
	This gives the first equality; a similar manipulation gives the second.
	\item By unitary of $U$, we can apply \eqref{item:state-dependent-distance-inverse} to get
	\begin{equation}
		\drho\rho U{V\dagg U' V}
		= \drho\rho {U\dagg V\dagg U' V} I.
	\end{equation}
	Next we apply \eqref{item:state-dependent-distance-isometry} to obtain
	\begin{equation}
		\drho\rho U{V\dagg U' V}
		= \drho{V\rho V\dagg} {VU\dagg V\dagg U' V V\dagg} {VV\dagg}.
	\end{equation}
	Now we notice that $VV\dagg$ is a projection with $(VV\dagg) V\rho V\dagg = V\rho V\dagg$, so we apply both parts of \eqref{item:state-dependent-distance-projection-is-identity}:
	\begin{equation}
		\drho{V\rho V\dagg} {VU\dagg V\dagg U' V V\dagg} {VV\dagg}
		=\drho{V\rho V\dagg} {VU\dagg V\dagg U'} {I}.
	\end{equation}
	Finally, by unitary of $U'$, we can apply \eqref{item:state-dependent-distance-inverse} to get
	\begin{equation}
		\drho{V\rho V\dagg} {VU\dagg V\dagg U'} {I}
		=\drho{V\rho V\dagg} {VU\dagg V\dagg} {(U')\dagg}.
	\end{equation}
	Taking adjoints and chaining equalities recovers the desired equation.
\end{enumerate}
\end{proof}

We now use some of the properties of the state-dependent distance to give an approximate version of Lemma \ref{lemma:unitary-observable-strategy-exact} from Section \ref{sec:linear-constraint-games}.

\begin{lemma}[Observable form for LCS game strategies, approximate version]\label{lemma:unitary-observable-strategy}
	Suppose that $\set{A_e^{(v)}}, \set{B_e}, \rho$ is a strategy presented via observables. Let $p_\text{con}$ be the probability that Alice and Bob pass the consistency check, $p_\text{sat}$ be the probability that Alice and Bob pass the constraint satisfaction check, and $p_\text{win}$ be the probability that they pass both checks. Then we have the immediate bounds
	\begin{equation}
		p_\text{sat} + p_\text{con} - 1 \leq p_\text{win} \leq \min\set{p_\text{sat}, p_\text{con}},
	\end{equation}
	together with the following bounds on $p_\text{sat}$ and $p_\text{con}$ in terms of the strategy:
	
	\begin{align}
		\label{eq:approximate-win-criterion-con}
		\eta &= \E{v,e} \frac14\drho{\rho}{A_e^{(v)}\otimes B_e}{I}^2,
		&\eta &\leq 1-p_\text{con} \leq d^2\eta,
		\\ \label{eq:approximate-win-criterion-sat}
		\mu &= 
		\E{v} \frac14
		\drho{\rho}{\prod_{e}\left(A_e^{(v)}\right)^{H(v,e)}\otimes I}{\w_d^{l(v)}I}^2,
		&\mu &\leq 1-p_\text{sat} \leq d^2\mu.
	\end{align}
\end{lemma}

\begin{proof}[Proof of Lemma \ref{lemma:unitary-observable-strategy}]
	As in the proof of the exact case, let $\tilde B^i_e$ and $\tilde A_v^a$ be projectors onto the eigenspaces of the observables, as in the following spectral decomposition:
	\begin{align}
		B_e:= \sum_j\w_d^{-j}\tilde B^{i}_{e}
		&& 
		A_e^{(v)}:= \sum_i \w_d^i\sum_{a: a(e) =i} \tilde A_{v}^a.
	\end{align}

	Now, we compute
	\begin{align}
		\E{v,e}\Tr_\rho A_e^{(v)}\otimes B_e 
			&= \E{v,e}\sum_{i,j}\w_d^{i-j}\Tr_\rho \left(\sum_{a: a(e) =i} \tilde A_{v}^a\right)\otimes \tilde B^{j}_{e} 
		\\	&= \E{v,e}\sum_k \w_d^k\Pr[a(e) - b \equiv k\mid \text{questions }x=v, y=e].
		\\	&= p_\text{con} + \sum_{k\in \Z_d\minus\set 0} \w_d^k\Pr[a(e) - b \equiv k].
	\end{align}
	Taking real parts and applying the inequalities of complex numbers \ref{lemma:convex-inequality-easy}, \ref{lemma:convex-inequality-hard}, we recover equation \eqref{eq:approximate-win-criterion-con}:
	\begin{align}
		1 - 2 (1 - p_\text{con})
		\leq
		&\E{v,e} \Re\Tr_\rho A_e^{(v)}\otimes B_e
		\leq 
		1 - 2d^{-2}(1-p)
		\\
		4 (1-p_\text{con})
		\geq
		&\E{v,e} \drho{\rho}{A_e^{(v)}\otimes B_e}{I}^2
		\geq 
		4 d^{-2}(1-p_\text{con}).
	\end{align}
	(To get from the first line to the second, we applied Lemma \ref{lemma:state-dependent-distance}\eqref{item:state-dependent-distance-square}.)

	With a similar computation, we get:
	\begin{align}
		 &\E{v}\w_d^{-l(v)}\Braket {\psi |\prod_{e} \left(A_e^{(v)}\right)^{H(v,e)}\otimes I | \psi } \\
		= & \E{v}\sum_k\w_d^{k}\Braket {\psi |\sum_{\substack{a\\ \sum_e H(v,e)a(e)\equiv k}} \tilde A_{v}^a\otimes I |\psi} \\
		= & \E{v}\sum_k \w_d^{k-l(v)}\Pr\left[\sum_{e}H(v,e)a(e)\equiv k\middle| \text{ question }x=v\right]\\
		= & p_\text{sat} + \sum_{k\in \Z_d\minus\set 0} \w_d^k\Pr\left[\sum_{e}H(v,e)a(e)\equiv k+l(v)\right].
	\end{align}
	Again, \eqref{eq:approximate-win-criterion-sat} follows from the above via Lemmas \ref{lemma:convex-inequality-easy} and \ref{lemma:convex-inequality-hard}.
\end{proof}

\subsection{The stability lemma}
\label{subsection:stability-lemma}

We'll use a general stability theorem for approximate representations of finite groups, which will let us take the following approach to robustness. From a quantum strategy winning with high probability, we extract an ``approximate representation'' of the solution group, i.e.\ a map from the group to unitaries which is approximately a homomorphism. The stability theorem lets us conclude that this function is close to an exact representation in the way that the unitaries act on the joint state of the provers, up to a local isometry. Once we have a representation, we'll be able to start applying reasoning analagous to that of \S \ref{subsection:exact-self-testing}. 

We were first made aware of results of this type by \cite{gowers2015inverse}. The result of interest was restated more conveniently in \cite{gowers2017generalizations}. In what follows, $U(\m H)$ will denote the group of unitary operators on the Hilbert space $\m H$. 
\begin{thm}[Informal statement of Theorem 15.2 of \cite{gowers2017generalizations}]\label{thm:gowers-hatami}
	Let $G$ be a finite group and $f: G\to U(\C^n)$ be such that $\norm{f(x)f(y) - f(xy)}_2 \leq \e\sqrt n$ for all $x,y\in G$. Then there exists $m \leq (1+\e^2)n$, an isometry $V: \C^n\to C^m$, and a unitary representation $\s : G \to U(\C^m)$, such that $\norm{f(x) - V^\dagger\s(x)V}_2 \in O(\e\sqrt n)$ for every $x\in G$.
\end{thm}

Applying this theorem directly requires a guarantee on the Hilbert-Schmidt distance between operators. However, experiments with nonlocal games will only give us guarantees on the state-dependent distance $D_\rho$ between operators, where $\rho$ is the state used by the provers. The following variant addresses this concern. The statement and proof are due to Vidick.
\begin{lemma}[\cite{vidick2017approx}]\label{lemma:vidick-gowers-hatami}
	Let $G$ be a finite group, $f: G \to U(\m H_A)$ be such that $f(x\1) = f(x)^\dagger$, $\rho_{AB}$ a state on $\m H_A\otimes \m H_B$  and
	\begin{equation}
		\E{x,y\in G}\drho\rho{f(x)f(yx)^\dagger f(y)\otimes I_B}{I_{AB}} \leq \eta.
	\end{equation}
	Then there is some Hilbert space $\m H_{\hat A}$, an isometry $V: \m H_A\to \m H_{\hat A}$, and a representation $\tau: G\to U(\m H_{\hat A})$ such that
	\begin{align}
		\E{x\in G} \drho{\rho}{f(x)\otimes I_B}{ V^\dagger \tau(x)V\otimes I_B} 
		&\leq \eta\text{, or equivalently}
		\E{x\in G} \drho{(V\otimes I_B)\rho (V\otimes I_B)\dagg}{Vf(x)V\dagg\otimes I_B}{\tau(x)\otimes I_B} 
		&\leq \eta.
	\end{align}
\end{lemma}

\noindent
Notice the lack of a dimension bound on $\hat A$. From the proof one can check that the dimension of $\hat A$ is at most $\abs G^2$ times the dimension of $A$. We won't use any dimension bound explicitly, and proving a tight dimension bound takes considerable effort. 
We give a self-contained proof of Lemma \ref{lemma:vidick-gowers-hatami}. 

\begin{proof}[Proof of Lemma \ref{lemma:vidick-gowers-hatami}, \cite{vidick2017approx}]
	Let $\s$ vary over irreducible representations of $G$. For each $\s$, let $n_\s$ be the dimension of $\s$.
	We define a generalized Fourier transform of $f$, which acts on irreps of $G$, by 
	\begin{equation}
		\hat f(\s) = \E{x\in G} f(x) \otimes \bar {\s (x)} \in \m L(C^d\otimes \C^{n_\s}).
	\end{equation}
	Let $\m H_{A_1A_2} = \bigoplus_\s \C^{n_\s}_{A_1}\otimes \C^{n_\s}_{A_2}$. (Notice that the dimension of $\m H_{A_1A_2}$ is $\abs G$ by Fact \ref{fact:character-dimension}.) For each $\s$, define a state $\ket{\r{EPR}_\s} = \frac1{\sqrt {n_\s}}\sum_i^{n_\s}\ket{ii}$ in the $\s$-summand of $\m H_{A_1A_2}$. (Notice that the $\ket{\r{EPR}_\s}$ form an orthonormal family.)
	Let $\m H_{A_3} = \Span\set{\ket \s}$ be a Hilbert space of dimension equal to the number of inequivalent irreps of $G$. Finally, we define the Hilbert space $\m H_{\hat A}$, isometry $V: \m H_A \to \m H_{\hat A}$, and representation $\tau: G\to U(\m H_{\hat A})$ from the statement of the lemma. 
	\begin{align}
		{\m H}_{\hat A} &= {\m H}_A\otimes {\m H}_{A_1}\otimes {\m H}_{A_2}\otimes {\m H}_{A_3},
		\\V &= \sum_\s n_\s(\hat f(\s)_{AA_1}\otimes I_{A_2})(I_A\otimes \ket\eprrho_{A_1A_2}\otimes \ket\s_{A_3}),
		\\\tau(x) &= I_{AA_1} \otimes \sum_\s (\s(x)_{A_2} \otimes \proj \s_{A_3}).
	\end{align}
	It's clear that $\tau$ is a unitary representation. We check that $V$ is an isometry:
	\begin{align}
		V^\dagger V 
		&= \sum_\s n_{\s}^2(I_A \otimes \bra\eprrho)\hat f (\s)^\dagger \hat f(\s)(I_A\otimes \ket\eprrho)
		\\&= \sum_\s n_{\s} \Tr_{A_1}\hat f (\s)^\dagger \hat f(\s)
		\\&= \b E_x f (x)^\dagger  f(x) 
		\\&= I_A.
	\end{align}
	Now we compute the pullback of $\tau$ along $V$:
	\begin{align}
		V^\dagger \tau(x) V 
		&= \sum_\s n_{\s}^2(I_A \otimes \bra\eprrho) (\hat f (\s)^\dagger \hat f(\s) \otimes \s(x))(I_A\otimes \ket\eprrho)
		\\&= \sum_\s n_{\s}^2 \b E_{y,z\in G}f(y)^\dagger f(z)\bra\eprrho\s(y)^T\bar{\s(z)} \otimes \s(x) \ket\eprrho
		\\&= \sum_\s n_{\s}  \b E_{y,z\in G} \Tr(\s(x)^T\s(y)^T\bar{\s(z)})f(y)^\dagger f(z)
		\\&= \b E_{y\in G}\sum_\s n_{\s}   \b E_{z\in G}\Tr(\s(yxz\1))f(y)^\dagger f(z)
		\\&= \b E_{y\in G} f(y)^\dagger f(yx),
	\end{align}
	where the last equality follows from Fact \ref{fact:orthogonality}. 
	
	Then it follows from properties of $\drho\rho\cdot\cdot$ that
	\begin{align}
		\E{x\in G} \drho{\rho}{f(x)\otimes I_B}{ V^\dagger \tau(x) V\otimes I_B}
		&=
		\E{x,y\in G} \drho{\rho}{f(x)f(yx)^\dagger f(y)\otimes I_B}{I} 
		\\&\leq \eta.
	\end{align}
	The equivalence of the two forms of the conclusion follows from Lemma \ref{lemma:state-dependent-distance}\eqref{item:state-dependent-distance-isometry-switching}.
\end{proof}
\noindent Notice that we can also use the lemma with the isometry acting on the state instead of the representation, since
	\begin{align}
		\E{x\in G} \drho{\rho}{f(x)\otimes I_B}{V^\dagger \tau(x) V\otimes I_B}
		&=\E{x\in G} \drho{\rho}{f(x)V^\dagger \tau(x)^\dagger V\otimes I_B}{ I}
		\\&=\E{x\in G} \drho{V\rho V^\dagger}{Vf(x)V^\dagger \tau(x)^\dagger VV^\dagger\otimes I_B}{I}
		\\&=\E{x\in G} \drho{V\rho V^\dagger}{Vf(x)V^\dagger \tau(x)^\dagger\otimes I_B}{I}
		\\&=\E{x\in G} \drho{V\rho V^\dagger}{Vf(x)V^\dagger\otimes I_B}{\tau(x)\otimes I_B}.
	\end{align}
	Here the last two equalities are applications of Lemma \ref{lemma:state-dependent-distance}(\ref{item:state-dependent-distance-isometry},\ref{item:state-dependent-distance-projection-is-identity}).

\subsection{Quantitative van Kampen lemma}
\label{subsection:quant-van-kampen}

In order to apply the stability lemma of the previous subsection, we need an error bound averaged over the whole solution group. From playing an LCS game, we learn an error bound averaged over the generators and relations. In order to go from the latter to the former, we need a bound on how much work is required to build up the individual group elements from its generators and relations. In particular, we'll use the following quantitative version of the van Kampen lemma introduced in \S \ref{subsection:van-kampen-diagrams}.

\begin{prop}\label{lem:quant-van-kampen}
	Suppose $G = \Braket{S:R}_{\Z_d}$ and $\m P$ is a $G$-picture witnessing the equation $w = J^a$.
	Then the equation $w = J^a$ is true, and can be proven by starting with the equation $1= 1$ and applying the following steps in some order:
	\begin{itemize}
		\item at most twice for each appearance of generator $s$ in $\m P$, conjugate both sides of the equation by $s$, and
		\item exactly once for each appearance of the relation $J^{-a}r\in R$ in $\m P$, right-multiply the left-hand side of the equation by $r$ and multiply the right-hand side by $J^a$.
	\end{itemize}
\end{prop}

It suffices to prove this only for group pictures whose edges and vertices form a connected graph. For graphs with more than one connected component, we can split the picture into subpictures, apply the lemma, and then glue them back in the obvious way.

The proof proceeds via a simple algorithm---we prove the validity of the relation witnessed by the group picture by starting from a subpicture (which witnesses a different relation), and inductively growing it to the whole picture. This can be thought of as a graphical way to prove the validity of the equation witnessed by the group picture, with each step in the algorithm corresponding to a rearrangement of the starting relation. The algorithm then terminates when the subpicture has grown to the full picture, and the starting relation has been transformed into the relation witnessed by the picture. We will then keep track of the steps in the algorithm to verify that Proposition \ref{lem:quant-van-kampen} is true. We describe the algorithm precisely in \ref{algorithm:quant-van-kampen}. We expect, however, that most readers will be satisfied by examining the example application of the algorithm in Figure \ref{fig:quant-van-kampen}.

\begin{figure}

	    \resizebox{\textwidth}{!}{\includegraphics{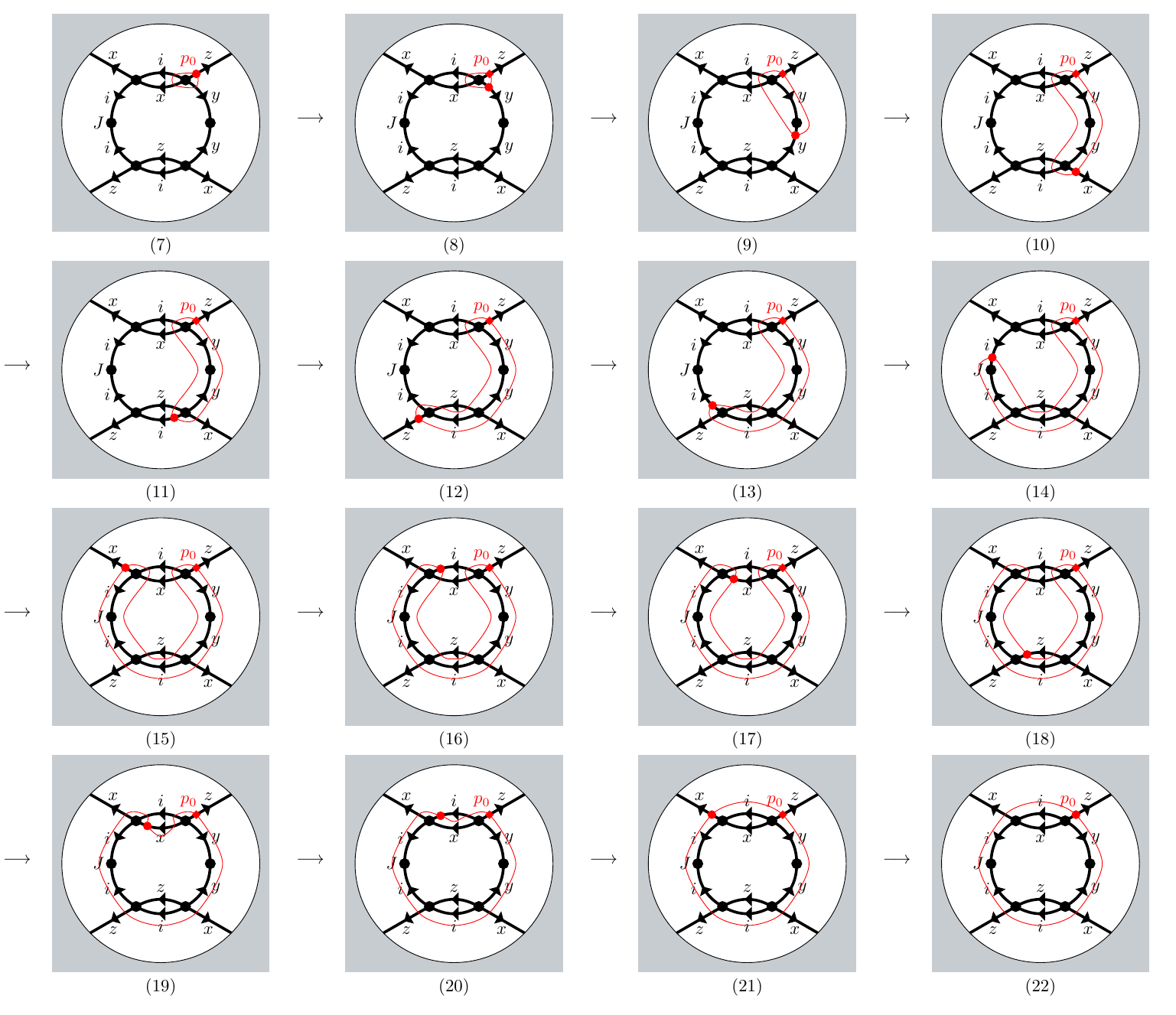}}

\caption{An instantiation of Algorithm \ref{algorithm:quant-van-kampen} on a group picture for the anticommutation of the qubit Pauli $X$ and $Z$ using the relations $ixyz = 1, y^2 = 1, i^2 = J$. 
Each diagram, together with its associated equation, represents the state of the memory at some step of the algorithm. Equation \eqref{eq:quant-van-kampen-12} is witnessed by the outer picture, and its truth is established by the whole sequence of equations.}

\begin{multicols}{3}
	~\begin{align}
		\label{eq:quant-van-kampen-0}
		ixyz &= 1\\
		zixy &= 1\label{eq:quant-van-kampen-1}\\
		zixy\1 &= 1\label{eq:quant-van-kampen-2}\\
		zixzix &= 1\label{eq:quant-van-kampen-3}\\
		xzixzi &= 1\label{eq:quant-van-kampen-4}
	\end{align}
	~\begin{align}
		xzixzz\1iz &= 1\label{eq:quant-van-kampen-5}\\
		zxzixzz\1i &= 1\label{eq:quant-van-kampen-6}\\
		zxzixzz\1i\1 &= J\label{eq:quant-van-kampen-7}\\
		zxzixzz\1x\1i\1x &= J\label{eq:quant-van-kampen-8}\\
		xzxzixzz\1x\1i\1 &= J\label{eq:quant-van-kampen-8b}
	\end{align}
	~\begin{align}
		i\1xzxzixzz\1x\1 &= J\label{eq:quant-van-kampen-8c}\\
		x\1i\1xzxzixzz\1 &= J\label{eq:quant-van-kampen-9}\\
		i\1xzxzixx\1 &= J\label{eq:quant-van-kampen-10}\\
		xzxzii\1 &= J\label{eq:quant-van-kampen-11}\\
		zxzx &= J\label{eq:quant-van-kampen-12a}\\
		xzxz &= J\label{eq:quant-van-kampen-12}
	\end{align}
\end{multicols}
\end{figure}
\label{fig:quant-van-kampen}
In order to define the algorithm, we set up some terminology:
The \emph{bubble} is the boundary of the expanding subpicture. A \emph{bubble-intersection} is the intersection between the bubble and an edge of the picture. The \emph{pointer} is a (vertex, edge) pair. In our diagrams, we'll draw it as a dot at the bubble-intersection at the edge near the vertex.
To \emph{advance the pointer} is to move the pointer from its current location to the next bubble-intersection clockwise around the bubble. 

We'll work informally with smooth curves. This approach can be rigorized with notions from differential topology---see e.g.\ \cite{guillemin2010differential} for an introduction to the subject. See \cite{slofstra2016tsirelson} for a more careful topological treatment of group pictures. 
Alternatively, one can use graph embeddings where all the vertices lie at integer coordinates and all curves are piecewise linear, and then argue constructively from there.

\begin{algorithm}\label{algorithm:quant-van-kampen}
	
First, pick an edge $e_0$ incident on the boundary of $\m P$. Let $v_0$ be the interior vertex incident to $e_0$. Initialize the bubble so that $v_0$ is the only vertex inside it, and each edge going out of $v_0$ has exactly one bubble-intersection. Initalize the pointer at the bubble-intersection with $e_0$; call this initial point $p_0$. Additionally, initialize variables $w \in \m F(S)$  a word in the generators and $j \in \braket J$ will be some power of $J$. Set $w$ to be the counterclockwise product of the labels of the edges around $v_0$; pick the order so that the rightmost letter corresponds to the lcoation of the pointer. Set $j$ to be label of $v_0$.

Repeat the following until the pointer returns to $p_0$.
Let $(v,e)$ be the location of the pointer. Let $s$ be the group element labeling $e$. Let $v'$ be the other vertex incident on $e$.
\begin{itemize}
	\item If $v'$ is on the boundary of $\m P$, advance the pointer. Additionally, replace $w$ by $sws\1$, canceling an $ss\1$ term that appears. (In the example of Figure \ref{fig:quant-van-kampen}, this happens immediately after states \eqref{eq:quant-van-kampen-0}, \eqref{eq:quant-van-kampen-3}, \eqref{eq:quant-van-kampen-5}, \eqref{eq:quant-van-kampen-8}, \eqref{eq:quant-van-kampen-12a}.)
	\item If $v'$ is not inside the bubble, continuously deform the bubble to contain $v'$.
	Move the pointer to $(v',e)$ and then advance the pointer.

	Additionally, let $r$ be the counterclockwise  product of the edges around $v'$, starting with $e$. Let $l$ be the label of $v'$. Replace $w$ by $wr$, canceling an $ss\1$ term that appears. Replace $j$ by $J^lj$. 
	(In the example of Figure \ref{fig:quant-van-kampen}, this happens immediately after states \eqref{eq:quant-van-kampen-1}, \eqref{eq:quant-van-kampen-2}, \eqref{eq:quant-van-kampen-4}, \eqref{eq:quant-van-kampen-6}, \eqref{eq:quant-van-kampen-7}.)
	\item If $v'$ is inside the bubble, advance the pointer. 

	If the pointer is now on $(v',e)$, continuously deform the bubble to contain $e$ and move the pointer back to the most recently visited intersection which still exists. Additionally, cancel an $ss\1$ term that was already present in $w$, and replace $w$ by $(s')\1 ws'$, where $s'$ is the generator associated with the final location of the pointer. (In the example of Figure \ref{fig:quant-van-kampen}, this happens immediately after states \eqref{eq:quant-van-kampen-9}, \eqref{eq:quant-van-kampen-10}, \eqref{eq:quant-van-kampen-11}.)

	If instead the pointer is not on $(v',e)$, replace $w$ by $sws\1$, canceling an $ss\1$ term that appears. 
	(In the example of Figure \ref{fig:quant-van-kampen}, this happens immediately after states \eqref{eq:quant-van-kampen-8b}, \eqref{eq:quant-van-kampen-8c}.)
\end{itemize}
\end{algorithm}

\begin{lemma}\label{lemma:quant-van-kampen-induction}
	After each iteration of the main loop of Algorithm \ref{algorithm:quant-van-kampen}, all of the following hold:
\begin{enumerate}[(i)]
	\item The equation $w=j$ is true in $G$.
	\item\label{lemma:quant-van-kampen-induction-2} The equation witnessed by the picture whose boundary is the bubble is $w=j$.
	\item\label{lemma:quant-van-kampen-induction-3} On the counter-clockwise arc from $p_0$ to the pointer, each bubble-intersection is a previous location of the pointer.
	\item\label{lemma:quant-van-kampen-induction-4} On the counter-clockwise arc from $p_0$ to the pointer, there is at most one bubble-intersection with each edge of the graph.
	\item The rightmost letter of the word on the left-hand side of the equation is the group element associated with the pointer.
\end{enumerate}
\end{lemma}
\begin{proof}
	\begin{enumerate}[(i)]
		\item The initial equation is true since it is a relation from the group presentation. Each step of the algorithm preserves truth of the equation, since it multiplies the sides of the equation by equal things.
		\item This is true of the intial equation. To see that each step of the algorithm preserves the property, we examine by cases. If the algorithm moves the pointer but not the bubble, then it cyclically permutes the letters on one side of the equation. This is okay, since ``the equation witnessed by a picture'' is only defined up to cyclic permutation.

		If the algorithm moves the bubble by including a new vertex, then the equation witnessed by the bubble changes by replacing the label of one edge at that vertex by the product of the rest of the edge-labels at that vertex.
		This is also how the algorithm changes the equation, multiplying by the relation of the vertex and canceling the $ss\1$ term of the edge to be replaced.

		If the algorithm moves the bubble to include an edge but no vertices, then the equation witnessed by the bubble changes by canceling an $ss\1$ term for that edge. This is also how the algorithm updates the equation.
		\item 
		This is true at the initial step, since the open arc is empty.
		Each time we move the pointer, we move it in the clockwise direction. Whenever we create new bubble-intersections by including a new vertex, we place the pointer at the counter-clockwise-most bubble-intersection at that vertex.
		\item 
		It suffices to check that whenever the pointer is on an edge $e$ with two bubble-intersections, the algorithm immediately moves the bubble so that there are $0$ bubble intersections with that edge.
		Assume inductively that the condition has been true in all the previous steps of the algorithm. We claim that there are no bubble-intersections on the counter-clockwise arc from the current pointer to the other bubble-intersection on edge $e$. 

		First, we must see that no vertex is enclosed by edge $e$ and the arc. 
		By the inductive hypothesis, any edge intersecting this arc does so at most once. By \eqref{lemma:quant-van-kampen-induction-3}, this bubble-intersection is a previous location of the pointer. If this edge were incident on a vertex in the region of interest, then the algorithm would have moved the bubble to enclose that vertex. Therefore, there are no vertices in the region of interest.

		Now suppose some edge $e'$intersects the arc. Since there is no vertex enclosed by the arc and edge $e$, $e'$ must either intersect $e$ or intersect the arc again. The former contradicts planarity of the graph. The latter contradicts the inductive hypothesis.
		\item 




		This is proved by casework similar to the proof of \eqref{lemma:quant-van-kampen-induction-2}.
	\end{enumerate}
\end{proof}

\begin{lemma}
The algorithm always terminates. During the runtime, the pointer leaves from each (vertex, edge) pair at most once. When the algorithm terminates, the equation witness by the bubble is the same as the equation witnessed by the picture. 
\end{lemma}
\begin{proof}

	By parts \eqref{lemma:quant-van-kampen-induction-3} and \eqref{lemma:quant-van-kampen-induction-4} of Lemma \ref{lemma:quant-van-kampen-induction}, we have that the pointer visits each (vertex, edge) pair at most once before visiting $p_0$ twice. But the algorithm terminates when it visits $p_0$ for a second time, never getting the chance to leave it a second time. 

	Once it terminates, the whole bubble is comprised of the counterclockwise arc from the pointer to $p_0$, since those are the same point. Then every edge intersecting the bubble does so in at most one place, and every bubble-intersection is a previous location of the pointer. Therefore, the interior of the bubble contains any vertex attached to an edge which has a bubble-intersection. So all of the edges intersecting the bubble are also edges intersecting the boundary of the whole picture. 

	Conversely, we claim that every vertex is contained in the interior of the bubble. This implies that every edge intersecting the boundary of the picture also intersects the bubble.
	 To see the claim, suppose towards a contradiction that there's a vertex outside the bubble. Take a simple (i.e. loop-free) path from that vertex to a vertex in the interior of the bubble. This path intersects the bubble at an edge which does not intersect the boundary of the picture; contradiction. 

	Since the bubble contains the same vertices and intersects the same edges as the picture, they witness the same equation. 

\end{proof}

\subsection{Quantitative stabilizer state bounds}
\label{subsection:quant-stabilizer-state-bounds}


To finish our collection of tools, we show that if a state is approximately stabilized by the simultaneous action of an irreducible group representation on two tensor factors, then the state is almost maximally entangled between those factors. This will allow us to deduce self-testing of the provers' state from self-testing of their operators.




\begin{lemma}
	Let $\tau:\Gamma \to U(\C^d)$ be an irreducible representation with $\Gamma$ a finite group. Then the maximally entangled state can be characterized as a uniform combination of operators from the image of $\tau\otimes \bar{\tau}$. In particular,
	\begin{equation}
        \proj\epr = \E{g\in \Gamma}\tau(g) \otimes \bar{\tau(g)}.
	\end{equation}
\end{lemma}
\begin{proof}
	We'll show four intermediate equations via simple computations.
	\begin{enumerate}
		\item $\rho_{AB} = \rho_{AB}\dagg$
		\item $\Tr \rho_{AB} = 1$
		\item $\rho_{AB}^2 = \rho_{AB}$
		\item $\Tr_B \rho_{AB}$ is maximally mixed.
	\end{enumerate}
	The first two items assert that $\rho_{AB}$ is a density matrix. The third shows that it is in fact pure. The fourth tells us that the state is maximally entangled across the $A/B$ cut. This characterizes the state.

	Our main trick for the whole proof will be to relabel the index of summation defining $\rho_{AB}$. 
	To prove the first item, we use the relabeling $x\mapsto x\1$. 
	\begin{align*}
		\rho_{AB} 
		&= \E x \tau(x)_A \otimes \bar{\tau(x)}_B
		\\&= \E x \tau(x\1)_A \otimes \bar{\tau(x\1)}_B
		\\&= \E x \tau(x)_A\dagg \otimes \bar{\tau(x)}_B\dagg
		\\&= \left[\E x \tau(x)_A \otimes \bar{\tau(x)}_B\right]\dagg
		\\&= \rho_{AB}\dagg.
	\end{align*}
	(Notice we've used the fact that $\tau(x)$ is unitary; this is one of several parts of the proof that relies on the finiteness of $\Gamma$.) Now define the character $\chi(x):= \Tr \tau(x)$ to compute:
	\begin{align*}
		\Tr \rho_{AB}
		&= \Tr \E x \tau(x)_A \otimes \bar{\tau(x)}_B
		\\&= \E x \chi(x) \bar{\chi(x)}
		\\&= 1.
	\end{align*}
	The final equation is true for the character of any irreducible representation character, and is referred to as the ``second orthogonality relation" in Dummit and Foote \cite{dummit2004abstract}. For the second item,
	\begin{align*}
		\rho_{AB}^2 
		&= \left(\E x \tau(x)_A \otimes \bar{\tau(x)}_B\right)^2
		\\&= \E x\E y \tau(x)_A\tau(y)_A \otimes \bar{\tau(x)}_B\bar{\tau(y)}_B
		\\&= \E x\left[\E y \tau(xy)_A \otimes \bar{\tau(xy)}_B\right]
		\\&= \E x\left[\E y \tau(y)_A \otimes \bar{\tau(y)}_B\right]
	\end{align*}
	In the last line, we used the relabeling $y \mapsto xy$. Continuing, we have
	\begin{align*}
		&= \E x\rho_{AB}
		\\&= \rho_{AB}.
	\end{align*}
	Now define $\rho_A = \Tr_B \rho_{AB}$. Let $y\in \G$ be arbitrary and use the relabeling $x \mapsto yxy\1$:
	\begin{align*}
		\rho_A 
		&= \E x \bar{\chi(x)}\tau(x)
		\\&= \E x \bar{\chi(yxy\1)}\tau(yxy\1)
		\\&= \E x \bar{\chi(x)}\tau(y)\tau(x)\tau(y)\1
		\\&= \tau(y)\left[\E x \bar{\chi(x)}\tau(x)\right]\tau(y)\1
		\\&= \tau(y)\rho_A\tau(y)\1.
	\end{align*}
	So $\rho_A$ commutes with $\tau(y)$ for all $y$. By Schur's lemma (Fact \ref{fact:schur}), $\rho_A$ is a scalar multiple of identity. Since $\Tr\rho_A = 1$, we know that $\rho_A$ is in fact the maximally mixed state. 
	
	Since the maximally entangled state of local dimension $d$ on systems $A$ and $B$ is the unique pure state such that the partial trace over either system gives a maximally mixed state, this concludes our proof. 
\end{proof}
\begin{cor}\label{lemma:stabilizer-state}
	Let $\m H_A \cong \m H_B \cong \C^d$. Let $\rho_{ABC}$ be a state on $\m H_A\otimes \m H_B \otimes \m H_C$. Let $\rho_{AB} = \Tr_C\rho_{ABC}$. 

	Let $\Gamma$ be a finite group. Suppose that for each $g\in \Gamma$, $\drho{\rho_{AB}}{\tau(g)\otimes \bar{\tau(g)}}I \leq \eta$. 

	Then there is a state $\rho_{\text{aux}}$ such that $\norm{\rho_{ABC} - \proj\epr\otimes \rho_\text{aux}}_1 \leq 6\eta^2$.
\end{cor}
\begin{proof}
By linearity, we compute
\begin{align}
	\drho{\rho_{AB}}{\proj \epr}I^2
	&=\drho{\rho_{AB}}{\E{g\in \Gamma}\tau(g)\otimes\bar{\tau(g)}}I^2
	\\&=\E{g \in \Gamma}\drho{\rho_{AB}}{\tau(g)\otimes\bar{\tau(g)}}I^2
	\\& \leq \eta^2.
\end{align}
An application of Lemma \ref{lemma:entanglement-monogamy} completes the proof.
	 
\end{proof}

\subsection{Robust self-testing}\label{subsection:robust-self-testing}

Now we prove a robust self-testing theorem for linear constraint system games. First, we specify precisely what we mean by robust self-testing.

\begin{definition}[Robust self-testing for LCS games]\label{defn: robust-self-testing}
	Let $G$ be an LCS game and $\set{\tilde A_e^{(v)}}, \set{\tilde B_e}, \ket\psi$ be a strategy presented via observables. Let $\d:\R \to \R$ be a continuous function with $\d(0) = 0$. We say that $G$ 
	\emph{self-tests} 
	the strategy with
	\emph{perfect completeness}
	and
	\emph{$\d$-robustness}
	if:
	\begin{itemize}
		\item The strategy $\set{\tilde A_e^{(v)}}, \set{\tilde B_e}, \ket\psi$ wins the game with probability $1$, and
		\item for every strategy $\set{A_e^{(v)}}, \set{B_e}, \rho$ which wins with probability at least $1-\e$, there is a local isometry $V = V_A\otimes V_B$ and auxiliary state $\rho_{\text{aux}}$ such that for every $e,v$ with $H(e,v)\neq 0$,
	\begin{align}
		\label{eq:robustness-condition-state}
		\norm{V\rho V^\dagger - \proj \psi\otimes \rho_{\text{aux}}}_1 &\leq \d(\e),\\
		\drho{V\rho V^\dagger}{V_A A_e^{v}V_A^\dagger\otimes I_A}{(\tilde A_e^{v}\otimes I) \otimes I_B}^2 &\leq \d(\e)\text{, and}\label{eq:robustness-condition-alice}\\
		\drho{V\rho V^\dagger}{I_A\otimes V_BB_eV_B^\dagger}{I_A\otimes (\tilde B_e\otimes I)}^2 &\leq \d(\e).\label{eq:robustness-condition-bob}
	\end{align}

	\end{itemize}
\end{definition}

We restrict our attention only to LCS games with sufficiently nice solution groups.
\begin{definition}\label{definition:group-test}
	Let $\G$ be a finite solution group over $\Z_d$ and $\tau:\G \to U(\C^d)$ an irreducible representation of $\G$ with $\tau(J) = \w_dI$. We say that $\G$ \emph{group-tests $\tau$}  
	if:
	\begin{itemize}
		\item every representation of $\G$ which sends $J \mapsto \w_dI$ is equivalent to $\tau$, and
		\item every irreducible representation of $\G$ sends $J$ to $\w_d^{j}I$ for some $j\in \Z_d$.
	\end{itemize}
\end{definition}
Our second condition may seem artificial. What we really need is the existence of some $\d$ such that if $\s$ is any irreducible not equivalent to $\tau$, then $\norm{\s(J)-\tau(J)}_2 \geq \delta$. This condition gives that to us with $\delta = \Theta(d\1)$.

\begin{theorem}\label{thm:robust-self-testing}
	Let $G$ be an LCS game over $\Z_d$ with vertex set $V$, edge set $E$, and constraints given by  $H:V\times E \to \Z_d$ and $l: V\to \Z_d$. Let $\G$ be the solution group of $G$. Suppose that:
	\begin{enumerate}[(i)]
		\item \label{assumption:bounded-degree}
		 $\abs E \leq \Delta$, $\abs V \leq \Delta$, and
		 each equation has at most $\Delta$ variables with multiplicity, i.e.\ \mbox{$\forall v:\sum_e\abs{H(v,e)} \leq \Delta$},
		\item \label{assumption:small-pictures}there is a canonical form $\can$ such that\footnote{We'll also need a technical assumption that for all $x\in \G$, $\can(Jx) = J\can(x)$ or $\can(Jx) = J^{1-d}\can(x)$.} 
		every equation of the form  $\can(e)e\1 = 1$ for $e\in E$ or $\can(g)\can(gh)\1\can(h) = 1$ for $g,h\in \G$ is witnessed by a $\G$-picture in which 
		 each variable is used at most $\Delta$ times
		 and 
		 each relation is used at most $\Delta$ times,
		\item \label{assumption:w-p-irreps}
		 $\G$ group-tests $\tau: \G \to U(\C^{d^n}) $ in the sense of Definition \ref{definition:group-test}.
	\end{enumerate}
	Then $G$ self-tests the strategy $\tilde A_e^{(v)} = \tau(e), \tilde B_e = \bar{\tau(e)}, \ket\psi = \ket{\r{EPR}_{d^n}}$
	with perfect completeness and $O(d^2\Delta^{10}\e)$-robustness.

\end{theorem}
The fact that the strategy wins the game with probability $1$ is proven as Proposition \ref{prop: perfect strategy}. The special case of $\e=0$ is the main result of \S \ref{subsection:exact-self-testing}. We remark that although the robustness bound does not seem to depend directly on $n$, typically $\Delta$ does.

The theorem is stated for finite solution groups. Using the stability lemma from \cite{de2017operator}, almost every part of the proof goes through for amenable\footnote{A countable group is \emph{amenable} if it admits a finitely-additive translation-invariant probability measure which is defined on every subset. For a finite group $\G$, we can take this measure as the familiar $\E{x\in \G}$. More exotic examples exist among infinite groups.} groups. However, we crucially use that the length of proofs of group equations is bounded by a constant, which is not true for infinite groups. It seems plausible that this barrier can be overcome; we leave this to future work. To avoid overloading notation, we stated Theorem \ref{thm:robust-self-testing} with sub-optimal bounds. A version of it with tighter, but more notation-involved, bounds is stated and proved in the Appendix as Theorem \ref{thm:robust-self-testing-appendix}.

We break the proof into several lemmas. In the statement of each lemma, we point out which of the assumptions of the main theorem we use. Before we proceed with the proof, we fix some useful notation.  We write $\prod_{i=1}^ng_i$ for the ordered product $g_1g_2\cdots g_n$. 
We write $r_v$ for the relation in $\G$ corresponding to equation $v\in V$, note that this is a word in the generators, say $r_v = s_1s_2\ldots s_n$. We write $\prod_{s\in r_v}f(r_v)$ for the ordered product $f(s_1)f(s_2)\cdots f(s_n)$. 
For $e,e' \in E$, we say $e\sim e'$ if they share an equation, i.e.\ there is some $v$ such that $H(v,e) \neq 0 \neq H(v,e')$. Furthermore, for each edge $e$, we fix a special vertex $v_e$ such that $H(v_e,e)\neq 0$.

\begin{lemma}[Assumption \eqref{assumption:bounded-degree}]\label{lemma:Bs-are-approximate-conjugate-operator-solution}
	$\set{B_e}$ is an ``approximate conjugate operator solution'' in the following sense:
	\begin{align}
	\label{eq:Bs-approximately-satisfy-constraints}
		\sum_v \drho{\rho}{
		\prod_{e\in r_v} I_A\otimes B_e}{\w_d^{-l(v)}I_{AB}} 
		&\leq 4\Delta^2\sqrt\e.
		\\
	\label{eq:Bs-approximately-commute}
		\sum_{e,e':e\sim e'} \drho{\rho}{I_A\otimes [B_e,B_{e'}]}{I_{AB}}
		&\leq 4\Delta^3\sqrt\e
		.\
	\end{align}

	Furthermore, we have similar inequalities for Alice, 
	\begin{align}
	\label{eq:As-approximately-satisfy-constraints}
		\sum_v\drho{\rho}{\prod_{e\in r_v}A_e^{(v_e)}\otimes I_B}{\w_d^{l(v)} I_{AB}}
		&\leq 8\Delta^2\sqrt{\e},
		\\
	\label{eq:As-approximately-commute}
		\sum_{e,e':e\sim e'} \drho{\rho}{\left[A_e^{(v_e)},A_{e'}^{(v_{e'})}\right]\otimes I_B}{I_{AB}}
		&\leq 4\Delta^3\sqrt\e
		.
	\end{align}
	Finally, these ``solutions'' are consistent in the sense that
	\begin{equation}
	\label{eq:As-and-Bs-approximately-consistent}
		\sum_e \drho{\rho}{A_e^{(v_e)}\otimes I_B}{I_A\otimes B_e^\dagger}
		 \leq 2 \Delta^2\sqrt\e.
	\end{equation}
\end{lemma}
Equation \eqref{eq:Bs-approximately-satisfy-constraints} says that Bob's operators approximately satisfy the group equations induced by the constraints. Equation \eqref{eq:Bs-approximately-commute} says that Bob's operators approximately commute whenever they share an equation.
\begin{proof}
	Recalling the consistency criterion \eqref{eq:approximate-win-criterion-con}, we have 
	\begin{align}
		\frac14\E{e,v}\drho{\rho}{A_e^{(v)}\otimes B_e}{I}^2 \leq \e.
	\end{align}
	An application of Cauchy-Schwarz 
	and then Lemma \ref{lemma:state-dependent-distance}\eqref{item:state-dependent-distance-inverse} 
	gives
	\begin{align}
		\sum_{e,v:H(v,e)\neq 0} \drho{\rho}{A_e^{(v)}\otimes B_e}{I} 
		&\leq 2 \abs E \abs V\sqrt\e \nonumber
		\\
		\sum_{e,v:H(v,e)\neq 0} \drho{\rho}{A_e^{(v)}\otimes I_B}{I_A\otimes B_e^\dagger}
		 &\leq 2 \Delta^2\sqrt\e. \label{eq:Bs-are-approximate-operator-solution-1}
	\end{align}
	Inequality \eqref{eq:As-and-Bs-approximately-consistent} can be obtained from inequality \eqref{eq:Bs-are-approximate-operator-solution-1} by dropping some (nonnegative) terms from the left-hand side.
	Similarly, we extract the following from the constraint satisfaction criterion \eqref{eq:approximate-win-criterion-sat}.
	\begin{align}
	\label{eq:Bs-are-approximate-operator-solution-2}
		\sum_v \drho{\rho}{
		\prod_{e\in r_v} A_e^{(v)} \otimes I}{\w_d^{l(v)}I}
		\leq 2 \sqrt \e \abs V.
	\end{align}
	Applying a triangle inequality and taking inverses (see \ref{lemma:state-dependent-distance}\ref{item:state-dependent-distance-triangle},\ref{item:state-dependent-distance-chaining}) to the previous two equations yields
	\begin{align}
		\sum_v \drho{\rho}{
		\prod_{e\in r_v} I_A\otimes B_e}{\w_d^{-l(v)}I_{AB}} 
		&\leq 2\Delta^2\sqrt\e +2 \abs V\sqrt\e
		,
	\end{align}
	establishing Equation \eqref{eq:Bs-approximately-satisfy-constraints}. 
	We use the same strategy for the commutators. First, note that Alice's operators commute exactly, i.e.\
	\begin{equation}
	\label{eq:Bs-are-approximate-operator-solution-3}
	\drho{\rho}{\left[A_e^{(v)}, A_{e'}^{(v)}\right]\otimes I_B} {I_{AB}} 
	= 0\text{ if }e\sim e'.
	\end{equation}
	Then we can chain triangle inequalities to deduce a bound on the magnitude of Bob's commutators:
	\begin{equation}
	\drho{\rho}{I_A\otimes \left[B_e,B_{e'}\right]}{I}
	\leq 
	2\drho{\rho}{A_e^{(v)}\otimes B_e}{I_{AB}}
	+2\drho{\rho}{A_{e'}^{(v)}\otimes B_{e'}}{I_{AB}}.
	\end{equation}
	Since we know the right-hand-side to be small on average, we sum over all equations and then apply Equation \eqref{eq:Bs-are-approximate-operator-solution-1}:
	\begin{align}
	\sum_{\substack{e,e'\\e\sim e'}}\drho{\rho}{I_A\otimes \left[B_e,B_{e'}\right]}{I}
	&\leq
	2\Delta \sum_{\substack{e,v\\H(v,e)\neq 0}}
	\drho\rho{A_e^{(v)}\otimes B_e}{I_{AB}}	
	\\
	&\leq 4\Delta^3\sqrt \e
	,
	\end{align}
	establishing Equation \eqref{eq:Bs-approximately-commute}.
	In the first line, we used that each equation has at most $\Delta$ variables. In the second line, we applied inequality \eqref{eq:Bs-are-approximate-operator-solution-1}.

	By reasoning on the $B$ system, we proved 
	Equations \eqref{eq:Bs-approximately-satisfy-constraints}, \eqref{eq:Bs-approximately-commute}
	from Equations \eqref{eq:Bs-are-approximate-operator-solution-1}, 
	\eqref{eq:Bs-are-approximate-operator-solution-2},
	\eqref{eq:Bs-are-approximate-operator-solution-3}. The same arguments on the $A$ system prove Equations \eqref{eq:As-approximately-satisfy-constraints}, \eqref{eq:As-approximately-commute}
	from Equations
	\eqref{eq:Bs-approximately-satisfy-constraints},
	\eqref{eq:Bs-approximately-commute},
	\eqref{eq:Bs-are-approximate-operator-solution-1}.

\end{proof}

In order to apply the stability lemma of \S \ref{subsection:stability-lemma}, we need to construct a function from the solution group to the group of unitaries. 
We already have functions defined on the generators: those which send $e\mapsto A_e^{(v_e)}$ and $e\mapsto B_e$. We would like to say ``extend $f_A$ and $f_B$ to all of $\G$ by multiplication''. However, these functions are not quite homomorphisms, so different choices of how to ``extend by multiplication'' define different functions. We'll use our canonical form $\can$ to make that choice in a consistent way. 
\begin{definition}
\label{def:4.19}
	Define $f_A:\G\to U(\m H_A)$ and $f_B:\G\to U(\m H_B)$ by
	\begin{align}
		f_A(g) = \begin{cases}
			\w_d I, &\text{ if }g=J\\
			A_e^{(v_e)}, &\text{ if }\can(g)=e\\
			\prod_{s\in \can(g)} f_A(s), &\text{ otherwise,}
		\end{cases}
		\\
		f_B(g) = \begin{cases}
			\w_d\1 I, &\text{ if }g=J\\
			B_e, &\text{ if }\can(g)=e\\
			\prod_{s\in \can(g)} f_B(s), &\text{ otherwise.}
		\end{cases}
	\end{align}
\end{definition}
Notice that we may have generators $e$ of the group for which $f_B(e) \neq B_e$. For example, in our canonical form for the Magic Square game defined in section \ref{sec:specific-games}, we'll have $\can(e_3) = e_1\1e_2\1$. If the equation $B_1B_2B_3 = I$ does not hold exactly, we have that $f(e_3) = B_1\1B_2\1 \neq B_3$. However, we do want $f(e_3)$ to be close to $B_3$. This is the content of the first item of the next lemma.

\begin{lemma}[Assumption \eqref{assumption:small-pictures}]\label{lemma:canonical-form-implies-stability}
	Suppose that $\set{A_e^{(v_e)}}$ and $\set{B_e}$ $\eta$-satisfy the relations from $R$ in the sense that
	\begin{equation}
	\label{lemma:canonical-form-implies-stability-1}
		\sum_v\drho{\rho}{\prod_{e\in r_v}A_e^{(v_e)}\otimes I_B}{\w_d^{l(v)} I_{AB}} \eta\text{, and }
		\sum_{e,e':e\sim e'} \drho{\rho}{\left[A_e^{(v_e)},A_{e'}^{(v_{e'})}\right]\otimes I_B}{I_{AB}} \leq \eta
	\end{equation}
	And similarly for $B$ we have 
	\begin{equation}
	\label{lemma:canonical-form-implies-stability-1-B}
	\sum_v \drho{\rho}{
		\prod_{e\in r_v} I_A\otimes B_e}{\w_d^{-l(v)}I_{AB}} \leq \eta\text{, and }
		\sum_{e,e':e\sim e'} \drho{\rho}{I_A\otimes [B_e,B_{e'}]}{I_{AB}}
	\end{equation}
	 Furthermore, suppose that $\set{A_e^{(v_e)}}$ and $\set{B_e}$ are $\eta$-consistent, i.e.\ 
	 \begin{equation}
	 \label{lemma:canonical-form-implies-stability-2}
	 	\sum_e \drho{\rho}{A_e^{(v_e)}\otimes B_e}{I_{AB}} \leq \eta.
	 \end{equation}
	  Then 
	  \begin{enumerate}[(1)]
	  	\item 
	  for all $e\in E$, the operators $f_A(e)$ and $f_B(e)$ are close to the operators used by Alice and Bob, i.e.\ 
	  \begin{align}
	  \label{eq:canonical-form-implies-stability-conclusion-1}
	  	\drho{\rho}{f_A(e)\otimes I_B}{ A_e^{(v_e)}\otimes I_B } &\leq 8\Delta \eta,
	  	\\
	  	\drho{\rho}{I_A \otimes f_B(e)}{ I_A \otimes B_e} &\leq 8\Delta\eta. 
	  \end{align}
	\item \label{item:canonical-form-implies-stability}$f_A$ and $f_B$ are suitable for application of the stability lemma \ref{lemma:vidick-gowers-hatami}, i.e.\ for all $x,y\in \G$, 
	\begin{align}
	\label{eq:canonical-form-implies-stability}
		\drho\rho{f_A(x)f_A(yx)\1f_A(y)\otimes I_B}{I_{AB}} &\leq 64\Delta^2\eta,
		\\
		\drho\rho{I_A\otimes f_B(x)f_B(yx)\1f_B(y)}{I_{AB}}&\leq 64\Delta^2\eta.
	\end{align}
	\item \label{item:canonical-form-implies-consistency}$f_A$ and $f_B$ are consistent, i.e.\  for all $x\in \G$,
	\begin{equation}
	\label{eq:canonical-form-implies-consistency}
		\drho\rho{f_A(x)\otimes f_B(x)}{I_{AB}} \leq 4\Delta\eta.
	\end{equation}
  \end{enumerate}

\end{lemma}
For our purposes, it would suffice to prove items \eqref{item:canonical-form-implies-stability} and \eqref{item:canonical-form-implies-consistency} on average over $g$ and $h$. This may make the upper bound smaller, but the authors presently know of no families of groups for which this improvement is better than a constant factor. 
\begin{proof}
	By the quantitative van Kampen lemma (Lemma \ref{lem:quant-van-kampen}), any identity of the form $\can(e)e\1 = 1$ has a proof using at most $2\Delta$ conjugations by each generator and at most $\Delta$ right-multiplications by each relation. In this proof, we replace each instance of a generator $e$ with the corresponding Bob operator $I\otimes B_{e}$, and replace the equality by a bound of the $\drho\rho\cdot\cdot$-distance between the two sides. By at most $2\Delta$ applications of Lemma \ref{lemma:state-dependent-distance}(\ref{item:state-dependent-distance-right-multiplication}) and at most $\Delta$ applications of Lemma \ref{lemma:state-dependent-distance}(\ref{item:state-dependent-distance-conjugation}), we get the bound on the $\drho\rho\cdot\cdot$-distance stated in Equation \eqref{eq:canonical-form-implies-stability-conclusion-1}. Repeating the same proof for identities of the form $\can(x)\can(yx)\1\can(y)$ but now starting from the $f_B(e)$ instead of the $B_e$ gives the second part of Equation \eqref{eq:canonical-form-implies-stability}. The same argument with the tensor factors reversed gives the first part.

	Finally, we obtain, using Lemma \ref{lemma:state-dependent-distance}\eqref{item:state-dependent-distance-chaining},
	\begin{align}
		\drho\rho{f_A(g)\otimes f_B(g)}{I_{AB}} 
		&= \drho{\rho}{\prod_{e\in \can(g)}{ A_e^{(v_e)}\otimes B_e}}{I_{AB}}
		\\&\leq \sum_{e\in \can(g)} 
		\drho\rho{A_e^{(v_e)}\otimes B_e}{I_{AB}}.
	\end{align}
	Each word in the canonical form must use at most $\Delta$ occurences of each generator, since all such occurences appear in a group picture with the word on the boundary. So we have
	\begin{align}
		\sum_{e\in \can(g)} 
		\drho\rho{A_e^{(v_e)}\otimes B_e}{I_{AB}}
		&\leq \Delta\sum_{e\in E} 
		\drho\rho{A_e^{(v_e)}\otimes B_e}{I_{AB}}
		\\&\leq \Delta \eta.
	\end{align}
\end{proof}

\begin{lemma}[Assumption \eqref{assumption:w-p-irreps}]
\label{lemma:fixing-the-J}
	Let $f:\G\to \m L(\m H)$ be such that $f(1)$ is a projection and $f(Jx) = \w_df(x)$ for all $x$. 
	Let $\s: \G \to U(\m H)$ be a representation. Let $\rho$ be a state on $\m H$. Finally, suppose $\E{x}\drho\rho{f(x)}{\s(x)} \leq \eta$. 

	Then there is a projection $P$ such 
	that $P$ commutes with $\s(x)$ for each $x$, $P\s(J)P = \w_dP$, and $\drho\rho PI \leq d \eta$. The same holds if we replace $\w_d$ by $\w_d\1$.\footnote{Indeed, we could replace $\w_d$ by $\w_d^k$ for any $k$ coprime to $d$.}
\end{lemma}
\begin{proof}
	Decompose $\s = \bigoplus_i \s_i$ as a sum of irreducibles. 
	For each $j \in \Z_d$, let $P_j$ be the projection onto the $\w_d^j$-eigenspace of $\s(J)$.
	Notice that these decompositions are compatible in the following sense: for each $i,j$, the map $x\mapsto P_j\s_i(x)P_j$ is either the all $0$-map or it is a representation on the range of $P_j$ sending $J$ to $\w_d^jI$. It follows that $x\mapsto P_1 \s(x)P_1$ is an operator solution. Now we compute, using inequality \ref{lemma:convex-inequality-hard},
	\begin{align}
	\frac12\drho\rho PI^2 
	&= 1 - \Re\Tr_\rho P
	\\
	\frac12\drho\rho PI^2
	&\leq \frac14d^2\left[
		1 - \Re\Tr_\rho\w_d\1\s(J)
	\right]
	\\
	\frac12\drho\rho PI^2
	&\leq \frac12\left(\frac d2\drho\rho{\s(J)}{\w_dI}\right)^2
	\\
	\drho\rho PI
	&\leq \frac d2\drho\rho{\s(J)}{\w_dI}.
	\end{align}
	The lemma will be established if we can show $\drho\rho{\s(J)}{\w_dI} \leq 2 \eta$. We first use the fact that expectation is invariant to multiplication by $J$.
	\begin{align}
		\E{x}\drho\rho{f(x)}{\s(x}
		&=\E{x}\drho\rho{f(Jx)}{\s(Jx)}
		\\&=
		\E{x}\drho\rho{\w_df(x)}{\s(J)\s(x)}
		\\&=
		\E{x}\drho\rho{f(x)}{w_d\1\s(J)\s(x)}
	\end{align}
	Next we use the triangle inequality and the unitarity of $\s(x)$. 
	\begin{align}
		\E{x}\drho\rho{\s(x)}{\w_d\1\s(J)\s(x)}
		&\leq 
		\E{x}\drho\rho{f(x)}{\s(x)}
		\E{x}\drho\rho{f(x)}{\w_d\1\s(J)\s(x)}
		\\
		\E{x}\drho\rho{\s(x)}{\w_d\1\s(J)\s(x)}
		&\leq2\eta. 
		\\
		\E{x}\drho\rho{\s(J)}{\w_dI}
		&\leq2\eta. 
	\end{align}
	Notice that the argument of the expectation on the left-hand side does not depend on $x$, so we have 
	\mbox{$\drho\rho{\s(J)}{\w_dI} \leq 2\eta$} unconditionally.
\end{proof}

\begin{proof}[Proof of Theorem \ref{thm:robust-self-testing}]
	By Lemma \ref{lemma:Bs-are-approximate-conjugate-operator-solution}, $f_A$ and $f_B$ each satisfy conditions \eqref{lemma:canonical-form-implies-stability-1}, \eqref{lemma:canonical-form-implies-stability-1-B} and \eqref{lemma:canonical-form-implies-stability-2} 
	with $\eta_1 = 2^4\Delta^3\sqrt \e$. By Lemma \ref{lemma:canonical-form-implies-stability}, $f_A$ and $f_B$ each satisfy the condition of the stability lemma \ref{lemma:vidick-gowers-hatami} with $\eta_2 =2^{10}\Delta^5\sqrt\e$. Applying the stability lemma, we get representations $\s_A,\s_B$ and isometries $W_A,W_B$ such that 
	\begin{align}
	\label{eq:fa-close-to-something}
		\E{x} \drho{\rho}{f_A(x)\otimes I_B}{ W_A^\dagger \s_A(x)W_A\otimes I_B} &\leq \eta_2\text{, and} \\
		\E{x} \drho{\rho}{I_A\otimes f_B(x)}{ I_A \otimes  W_B^\dagger \s_B(x)W_B} &\leq \eta_2. 
	\end{align}
	Recall that $f_A(J) = \w_d$ and $f_B(J) = \w_d\1$. Note that furthermore 
	\begin{equation}
		f_A(Jx) = \w_df_A(x)\text{ and } f_B(Jx) = \w_d\1f_B(x)\text{ for any }x\in\G.
	\end{equation}
	Now we apply Lemma \ref{lemma:fixing-the-J} with $\s=\s_A,\s_B$, $f(x) = W_Af_A(x)W_A^\dagger, W_Bf_B(x)W_B^\dagger$ on the states $(W_A\otimes I_B)\rho (W_A^\dagger \otimes I_B)$,$(I_A\otimes W_B)\rho (I_A\otimes W_B^\dagger)$, respectively. Let $P_A$ and $P_B$ be the resulting projectors. One can check that $x \mapsto P_A\s_A(x)P_A$ is an operator solution, while  $x \mapsto P_B\s_B(x)P_B$ is a conjugate operator solution. By assumption \eqref{assumption:w-p-irreps}, we can apply Lemma \ref{lem:unique-operator-solution} to get isometries ${\tilde W}_A, {\tilde W}_B$ such that
	\begin{align} 
	\label{eq:tau-construction-by-isometry}
	{\tilde W}_AP_A\s_A(x)P_A{\tilde W}_A^\dagger &= \tau(x)\otimes I\text{, and }
	\\
	{\tilde W}_BP_B\s_B(x)P_B{\tilde W}_B^\dagger &= \bar{\tau(x)}\otimes I.
	\end{align}

	Let $V_A = \tilde W_AW_A,V_B = \tilde W_BW_B, V = V_A\otimes V_B$. We compute:
	\begin{align}
		\eta_2 
		&\geq
		\E{x} \drho{(W_A\otimes I)\rho (W_A^\dagger\otimes I)}{W_Af_A(x)W_A^\dagger\s_A(x)^\dagger\otimes I_B}{I_{AB}}
		&& \text{derived from Equation \eqref{eq:fa-close-to-something}}
		\\
		(d+1)\eta_2
		&\geq
		\E{x} \drho{(W_A\otimes I)\rho (W_A^\dagger\otimes I)}{W_Af_A(x)W_A^\dagger\s_A(x)^\dagger P_A\otimes I_B}{I_{AB}}
		&& \text{right-multiply }P_A
		\\&=
		\E{x} \drho{(W_A\otimes I)\rho (W_A^\dagger\otimes I)}{W_Af(x)W_A^\dagger P_A\s_A(x)P_A\otimes I_B}{I_{AB}}
		&& \text{commute }P_A\text{ past }\s(x)
		\\&=
		\E{x} \drho{V\rho V^\dagger}{V_Af_A(x)W_A^\dagger(\tilde W_A^\dagger\tilde W_A) P_A\s_A(x)^\dagger P_A\tilde W_A^\dagger\otimes V_BV_B^\dagger}{\tilde W_A W_A \otimes V_BV_B\dagg}
		&& \text{conjugate by } \tilde W_A\otimes V_B
		\\&=
		\E{x} \drho{V\rho V^\dagger}{V_Af_A(x)W_A^\dagger(\tilde W_A^\dagger\tilde W_A) P_A\s_A(x)^\dagger P_A\tilde W_A^\dagger\otimes I_B}{I_{AB}}
		&& \text{apply Lemma \ref{lemma:state-dependent-distance}\eqref{item:state-dependent-distance-projection-is-identity} twice } 
		\\&=
		\E{x} \drho{V\rho V^\dagger}{V_Af_A(x)V_A^\dagger\otimes I_B}{\tilde W_A P_A\s_A(x) P_A\tilde W_A^\dagger\otimes I_B}
		&& \text{apply Lemma \ref{lemma:state-dependent-distance}\eqref{item:state-dependent-distance-inverse}}
		\\&=
		\E{x} \drho{V\rho V^\dagger}{V_Af_A(x)V_A^\dagger\otimes I_B}{(\tau(x) \otimes I) \otimes I_B}
		&& \text{apply Equation \eqref{eq:tau-construction-by-isometry}.}
	\end{align}
	The same proof works for the $B$ objects, yielding
	\begin{equation}
	\label{eq:robust-self-testing-proof-2}
		\E{x} \drho{V\rho V^\dagger}{I_A \otimes V_Bf_B(x)V_B^\dagger}{I_A \otimes (\bar{\tau(x)} \otimes I)}
		\leq (d+1)\eta_2.
	\end{equation}
	Recalling equation \eqref{eq:canonical-form-implies-consistency} and taking an expectation, we have
	\begin{align}
	\label{eq:robust-self-testing-proof-1}
		4\Delta\eta_1 
		&\geq
		\E{x}\drho{\rho}{f_A(x)\otimes f_B(x)}{I_{AB}}
		\\&=
		\E{x}\drho{V\rho V\dagg}{V_Af_A(x) V_A\dagg\otimes V_Bf_B(x) V_B\dagg}{I_{AB}}
		\\&=
		\E{x}\drho{V\rho V\dagg}{V_Af_A(x) V_A\dagg\otimes I_B}{I_A\otimes V_Bf_B(x)\dagg V_B\dagg}.
	\end{align}
	Weakening the previous inequality for convenience, we have
	\begin{equation}
		\E{x}\drho{V\rho V\dagg}{V_Af_A(x) V_A\dagg\otimes I_B}{I_A\otimes V_Bf_B(x)\dagg V_B\dagg}.\leq \eta_2.
	\end{equation}
	Applying three triangle inequalities gives us
	\begin{equation}
		\label{eq:robust-self-testing-proof-3}
	 	\E x \drho{V\rho V^\dagger}{(\tau(x)\otimes I)_A\otimes (\bar{\tau(x)}\otimes I)_B}{I_{AB}} 
		\leq 3d\eta_2
	 \end{equation}
	 which says that $(\tau(x)\otimes I)_A \otimes (\bar{\tau(x)}\otimes I)_B$ approximately stabilizes $V\rho V^\dagger$ on average. Now we see that a similar bound holds pointwise. We use a change of variable and the homomorphism property of $\tau$; this is the same technique used in the proof of Lemma \ref{lemma:fixing-the-J}.
	\begin{align}
		\E x \drho{V\rho V^\dagger}{(\tau(x)\otimes I)\otimes (\bar{\tau(x)}\otimes I)}{I} 
		&\leq 3d\eta_2
		\\
		\E x \drho{V\rho V^\dagger}{(\tau(yx)\otimes I)\otimes (\bar{\tau(yx)}\otimes I)}{I} 
		&\leq 3d\eta_2
		\\
		\label{eq:robust-self-testing-proof-4}
		\E x \drho{V\rho V^\dagger}{(\tau(y)\tau(x)\otimes I)\otimes (\bar{\tau(y)\tau(x)}\otimes I)}{I} 
		&\leq 3d\eta_2
		\\
		\label{eq:robust-self-testing-proof-5}
		\E x \drho{V\rho V^\dagger}{(\tau(y)\otimes I)\otimes (\bar{\tau(y)}\otimes I)}{I} 
		&\leq 6d\eta_2.
	\end{align}
	The final equation follows from right-multiplying the previous equation by the inverse of Equation \eqref{eq:robust-self-testing-proof-3}. Since the expression has no $x$ dependence, we can drop the average and draw the same conclusion pointwise.

	Now we use the finiteness of the group and apply Corollary \ref{lemma:stabilizer-state}.  We trace out irrelevant subsystems and then apply the conclusion of that lemma:
	\begin{align}
		\forall y\drho{V\rho V^\dagger}{(\tau(y)_{A_1}\otimes I_{A_2})\otimes (\bar{\tau(y)}_{B_1}\otimes I_{B_2})}{I} 
		&\leq 6d\eta_2,
		\\
	\Rightarrow	\,\,\,\,\forall y\drho{\Tr_{A_2B_2}V\rho V^\dagger}{\tau(y)_{A_1}\otimes \bar{\tau(y)}_{B_1}}{I} 
		&\leq 6d\eta_2,
		\\
		\Rightarrow \,\,\,\, \norm{\Tr_{A_2B_2}V\rho V^\dagger - \proj\epr^{\otimes n}\otimes \rho_\text{aux}}_1
		&\leq 6^3(d\eta_2)^2, \,\,\,\,\,\,\,\text{by Corollary \ref{lemma:stabilizer-state}}.
	\end{align}
	This establishes the robustness condition \eqref{eq:robustness-condition-state} with $\d(\e) = O(d^2\eta_2^2) = O(d^2\Delta^{10}\e)$.

	Next, we show the other robustness conditions. It'll suffice to find that $f$ is close to $\tau$ pointwise. 
	Equations \eqref{eq:robust-self-testing-proof-2}, \eqref{eq:robust-self-testing-proof-1} with a triangle inequality give
	\begin{align}
		\E{x}&\drho{V\rho V^\dagger} {V_Af_A(x)V_A^\dagger\otimes (\bar{\tau(x)} \otimes I)_B}{I}\leq 2d\eta_2\text{, and}
		\\
		\E{x}&\drho{V\rho V^\dagger} {(\tau(x) \otimes I)_A\otimes V_Bf_B(x)V_B^\dagger}{I} \leq 2d\eta_2.
	\end{align}
	From here, we argue only on the $A$ side. The argument for the $B$ side is analogous. Applying a change of variable and then multiplying gives
	\begin{align}
		\E{x}\drho{V\rho V^\dagger}{V_Af_A(ex)V_A^\dagger\otimes (\bar{\tau(ex)} \otimes I)}{I} &\leq 2d\eta_2,
		\\
		\E{x}\drho{V\rho V^\dagger}{V_Af_A(ex)f_A(x)^\dagger V_A^\dagger\otimes (\bar{\tau(e)} \otimes I)}{I} &\leq 4d\eta_2.
	\end{align}
	By Equation \eqref{eq:canonical-form-implies-stability},
	\begin{equation}
	\drho{V\rho V^\dagger}{V_Af_A(ex)f_A(x)^\dagger V_A^\dagger\otimes I_B}{V_Af_A(e) V_A^\dagger\otimes I_B} \leq \eta_2.
	\end{equation}
	Using this, Equation \eqref{eq:robust-self-testing-proof-5}, and two triangle inequalities gives
	\begin{equation}
		\drho{V\rho V^\dagger}{V_Af_A(e)V_A^\dagger\otimes I_A}{(\tau(e)\otimes I) \otimes I_B}
		\leq 11d\eta_2.
	\end{equation}
	From the conclusion of Lemma \ref{lemma:canonical-form-implies-stability}, we know that $f_A(e)$ is $\eta_2$-close to $A_e^{(v_e)}$. One more triangle inequality establishes robustness conditions \eqref{eq:robustness-condition-alice}, \eqref{eq:robustness-condition-bob} with $\d = O(d^2\eta_2^2) = O(d^2\Delta^{10}\e)$.

\end{proof}

\section{On the failure of Magic Square and Pentagram for $d\neq 2$}\label{sec:impossibility}
	One can generalize the magic square and magic pentagram games by taking the constraints and the answers in the game to be mod $d$ (instead of simply mod $2$). A previous version of this paper falsely claimed that these generalizations are pseudotelepathy games, and that moreover our self-testing theorem \ref{thm:robust-self-testing} applies to them. It is instead the case that for any $d \neq 2$, both the magic square and magic pentagram games are not pseudotelepathy games. The following theorem establishes this.

\begin{thm}
    Let $\G$ be the solution group of the magic square game or the magic pentagram game over $\Z_d$. Then $\G$ satisfies $J^2 = 1$. In particular, if $d$ is odd, then $J = 1$ and $\G$ is abelian. 
\end{thm}

\begin{proof}
    One can show that for any pair $\set{x,z}$ of generators which do not share a constraint,
    we have $[x,z] = J$. (See Lemmas \ref{prop:four-edges} and \ref{prop:six-edges}.)
    Applying the same observation with the role of $x$ and $z$ swapped shows that $[z,x] = J.$ For general group commutators we have that $[x,z] = [z,x]\1$. In particular $J = J\1$ or equivalently, $J^2 = 1 = J^d$. If $d$ is odd, then $J^{d+1} = (J^2)^{\frac{d+1}{2}} = 1 = J^d$. This implies $J=1$. Since the commutator subgroup of $\G$ is equal to the trivial subgroup $\braket{J}$ (see Lemmas \ref{lemma:commutator-subgroup-square} and \ref{lemma:commutator-subgroup-pentagram}), $\G$ is abelian.
\end{proof}

Note that a solution group with $J=1$ has no operator solution, even in the commuting operator model of entanglement. Separately, an abelian group has an operator solution iff it has a classical solution.

In a manuscript to appear shortly after this one, Joel Wallman \cite{wallman2019} shows that there is no pseudotelepathy LCS game whose ideal operators are tensor products of Paulis mod $d$ for $d\neq 2$. 
	
\section{Self-testing of specific games}\label{sec:specific-games}
We now apply the results of the previous section to conclude robustness for a specific family of games.
We must both understand the representation theory of their abstract solution groups and the combinatorics of the presentations for those groups. Even though our general robust self-testing theorem holds for LCS games mod $d$, we are currently only aware of applications of it to examples of LCS games mod $2$. In this section, we show applications of our theorem to the magic square and magic pentagram games mod $2$, and to certain parallel versions of these games.

\subsection{The qudit pauli group}
In this subsection, we formally introduce the Pauli group. We state definitions and prove properties for the Pauli group mod $d$. However, we will later only utilize such properties for the Pauli group mod $2$. As mentioned earlier, a manuscript by Joel Wallman, to appear shortly after this one, shows that there does not exist any pseudotelepathy LCS game mod $d$, for $d\neq 2$, whose ideal strategy consists of products of Pauli operators.
\begin{definition}\label{definition:n-qudit-pauli-group}
	The $n$-qudit Pauli group of local dimension $d$ is denoted $\m P_d^{\otimes n} := \braket{S:R}_{\Z_d}$ and presented with generators and relations
	\begin{align}
		S = \set{x_i,z_i\;i\leq n}
		&&
		R = \set{J\1[x_i,z_i], [x_i,x_j], [z_i,z_j], [x_i,z_j]\;i \neq j \leq n}
	\end{align}
\end{definition}

We aim to show that the Pauli group is suitable for applying the results from Section \ref{sec:self-testing}. 

\begin{definition}
\label{def:representations-pauli-group}
	We now define maps $\tau_l^{(n)}: \m P_d^{\otimes n} \to U(\C^{d})^{\otimes n}$ as 
	\begin{align}
	\tau_l^{(n)}(J) &= \w_d^l I
	,
	\\\tau_l^{(n)}(x_i) &= \underbrace{I\otimes \cdots I}_{i-1} \otimes X^l \otimes \underbrace{I \otimes\cdots \otimes I}_{n-i}
	,
	\\\tau_l^{(n)}(z_i) &= \underbrace{I\otimes \cdots I}_{i-1} \otimes Z \otimes \underbrace{I \otimes\cdots \otimes I}_{n-i}.
	\end{align}
\end{definition}
\begin{lemma}
\label{lemma:representations-pauli-group}
	The $\set{\tau_l^{(n)}\;l\in \Z_d\minus\set0}$ are $d-1$ inequivalent representations of dimension $d^n$. 
\end{lemma}
\begin{proof}
	To see that they are representations, it suffices to check the commutation and anticommutation relations. To see that they are inequivalent, see that their characters differ at $J$, since $\Tr\tau_l^{(n)}(J) = \w_d^ld^n$. 
\end{proof}

\begin{prop}\label{prop:representations-pauli-group}
	$\m P_d^{\otimes n}$ group-tests $\tau_1^{(n)}$ in the sense of Definition \ref{definition:group-test}.
\end{prop}

To prove this, we first establish the following lemma, which will let us count the elements of $\pauli$.

\begin{lemma}
\label{lemma:canonical-form}
	There is a canonical form $\can: \m P_d^{\otimes n} \to \m F(S)$ which sends each element to a string of the form
	\begin{equation}
	 	J^{a_1}\prod_{i=1}^nx_i^{a_{2i}}z_i^{a_{2i+1}}, a_i\in \Z_d.
	 \end{equation}
\end{lemma}
\begin{proof}
	First, we see that each element can be written this way. Start with an arbitrary word representing the element and apply the commutation and anticommutation relations to get the $x_i$ and $z_i$ in order. Finish by commuting all of the $J$s to the front and applying the relations $s^d = 1$ to get all of the exponents to lie in $\Z_d$. 

	Next, we see that different words represent different group elements. Suppose that 
 	\begin{equation}\label{eq:canonical-form}
 		J^{a_1}\prod_{i=1}^nx_i^{a_{2i}}z_i^{a_{2i+1}} = J^{b_1}\prod_{i=1}^nx_i^{b_{2i}}z_i^{b_{2i+1}}.
 	\end{equation}
 	Then by various applications of the (twisted) commutation relations, we have
 	\begin{equation}\label{eq:canonical-form-uniqueness}
 		J^{c_1} = \prod_{i=1}^nx_i^{a_{2i} - b_{2i}}z_i^{a_{2i+1}-b_{2i+1}}
 	\end{equation}
 	for some $c_1 \in \Z_d$. The left hand side is always central, but the right hand side is central only if $a_i = b_i$ for all $i\in [2,2n+1]$. 
 	(Suppose for example that $a_3 - b_3 \neq 0$, so that the power of $z_1$ is nonzero. Then the right hand side fails to commute with $x_1$.) 
 	In this case, we can see that in fact $c_1 = a_1 - b_1$, so equation \eqref{eq:canonical-form-uniqueness} holds only if $J^{c_1} = 1$ in the group. 
 	 But Proposition \ref{lemma:representations-pauli-group} gives us a representation in which $J$ and $1$ are represented by distinct matrices. Therefore, equation \eqref{eq:canonical-form-uniqueness} holds only when $a_i = b_i$ for all $i$.
\end{proof}

Thanks to the canonical form, we can easily compute the size of $\m P_d^{\otimes n}$.
\begin{cor}
	$\m P_d^{\otimes n}$ has $d^{2n+1}$ elements.
\end{cor}
\begin{proof}[Proof of Proposition \ref{prop:representations-pauli-group}]
	We'll check that $\m P_d^{\otimes n}$ has exactly $d-1$ irreducible representations of dimension $d^{n}$, each sending $J$ to a different nontrivial $d\th$ root of unity. All other irreducible representations are $1$-dimensional and send $J$ to $1$. 

	We'll complete the character table of $\pauli$. Now that we know the size of the group, we can check via Fact \ref{fact:irreducibility-criterion} that the representations of Lemma \ref{lemma:representations-pauli-group} are irreducible. 

	Next, we notice that the commutator subgroup $[\m P_d^{\otimes n},\m P_d^{\otimes n}]$ is equal to $\braket J$, the cyclic subgroup generated by $J$. This has order $d$, so by Fact \ref{fact:1-dim-irreps}, there are $d^{2n}$ irreps of dimension $1$ which send $J$ to $1$. Now we add the squares of the dimensions of our irreps and see that they saturate equation \eqref{eq:character-dimension}.
	\begin{equation}
		\abs\pauli = d^{2n+1} = (d-1)\cdot (d^n)^2 + (d^{2n})\cdot(1)^2 = \sum_\s (\dim \s)^2.
	\end{equation}
	Therefore, we've found all irreducible representations of $\pauli$. 
\end{proof}

\begin{lemma}\label{lemma:small-group-pictures-pauli-n}
	Let $\can$ be the canonical form from Lemma \ref{lemma:canonical-form}. Then each equation $\can(x)\can(yx)\1\can(y) = 1$ is witnessed by a $\pauli$-picture in which
	each generator and each relation appears at most $18d^2n$ times.
\end{lemma}
\begin{proof}
	See Figure \ref{fig:small-group-pictures-pauli-n}. Starting from arbitrary $x, y$, we compute $\can(x), \can(y), \can(yx)\1$. Draw a group picture whose boundary is $\can(x)\can(yx)\1\can(y)$ up to $J$ terms. Link each positive term from $\can(x)$ and $\can(y)$ with an appropriate negative term from $\can(yx)\1$. In the case that there are more positive terms than negative terms, link them with each other using a relation of the form $s^d=1$. At each intersection of links, add a vertex with either a commutation relation or an anticommutation relation. This subdivides the links into edges, giving us a valid group picture. Now we compute its size.

	There are $2n$ generators and each generator has multiplicity at most $d$ in each of $\can(x), \can(xy)\1, \can(y)$. Therefore, we draw at most $3d\cdot 2n$ links in the above drawing process. Each link intersects each other link at most once, so each link is subdivided into at most $6dn$ edges. Each generator labels at most $3d$ links, so there are at most $18d^2n$ edges with a given label. We must also count the uses of the generators. Recall that each relation involves only two generators. Therefore, each relation is used at most $(3d)^2$ times---once for each pair of links labelled by the generators in the relation.
\end{proof}

\begin{figure}
	\resizebox{\textwidth}{!}{
		\begin{tabular}{
		m{0.3\textwidth}m{0.03\textwidth}m{0.3\textwidth}m{0.03\textwidth}m{0.3\textwidth}
		}
			\resizebox{0.3\textwidth}{!}{\includegraphics{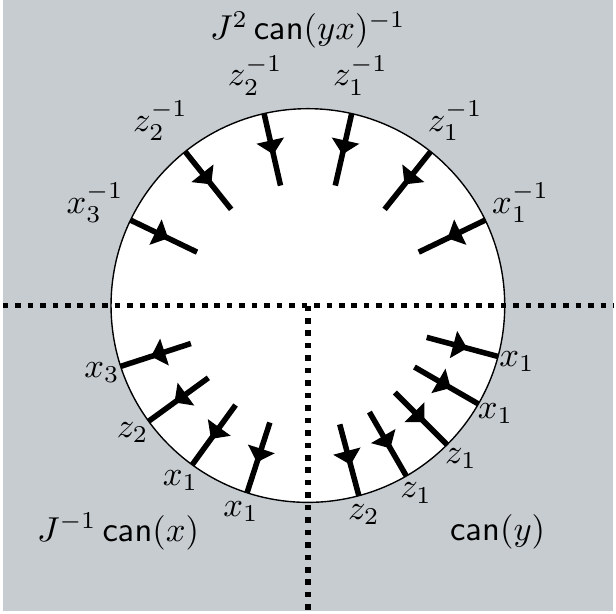}}
		    & $\longrightarrow$ &
			\resizebox{0.3\textwidth}{!}{\includegraphics{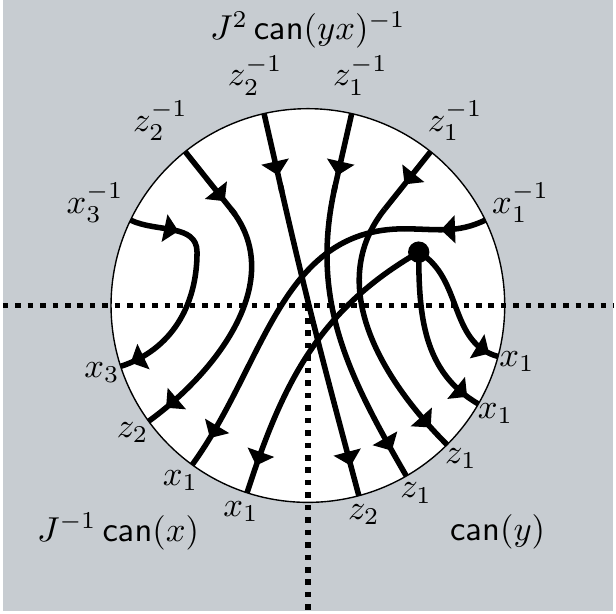}}
		    & $\longrightarrow$ &
			\resizebox{0.3\textwidth}{!}{\includegraphics{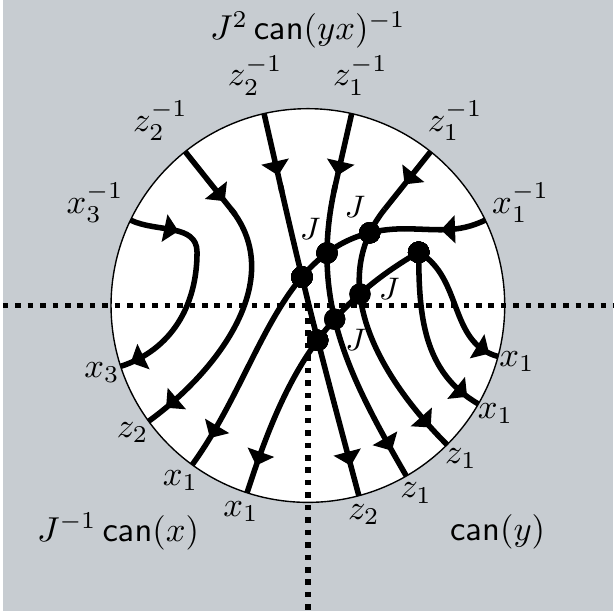}}
		\end{tabular}
	}
	\caption{Set $d = n = 3$ and $x = J\1x_1^2z_2x_3$, $y = x_1^2z_1^2z_2$. Then $\can(yx)\1 = J^{-2}x_3^{-1}z_2^{-2}z_1^{-2}x_1^{-1}$. The group picture witnesses that $J\1\can(x) J^2\can(yx)\1 \can(y) = J^4 = J$, from which it follows by scalar multiplication that $\can(x)\can(yx)\1\can(y) = 1$.}
	\label{fig:small-group-pictures-pauli-n}
\end{figure}

\subsection{Self-testing the Magic Square}
Recall the definition of the Magic Square game from Example \ref{example:magic-square}. 

\begin{definition}[Ideal strategy for Magic Square LCS game $\pmod 2$] 
	See Figure \ref{eq:mermin-peres-magic-square}. Let $A_e$ be the operator which appears on the right-hand side in the same spot as variable $e$ appears on the left-hand side. Set $A_e^{(v)}:= A_e$ for all $v$. Then set $B_e = \bar{A_e}$ (where any choice of basis works for the conjugation).  Set $\ket\psi = \ket{\text{EPR}}^{\otimes 2}$. We define $\{A_e^{(v)}\}, \{B_e \}, \ket\psi$ to be the \emph{ideal strategy} for the Magic Square game $\pmod 2$.
\end{definition}
Notice that the $B_e$ are defined only up to local isometry, because of the freedom in the choice of basis for conjugation.

\begin{figure}[!b]
	\begin{center}
	\begin{tabular}{m{0.4\textwidth}m{0.4\textwidth}}
		\resizebox{!}{0.22\textheight}{\includegraphics{MagicSquare-figure30.pdf}}
		&
		\resizebox{!}{0.22\textheight}{\includegraphics{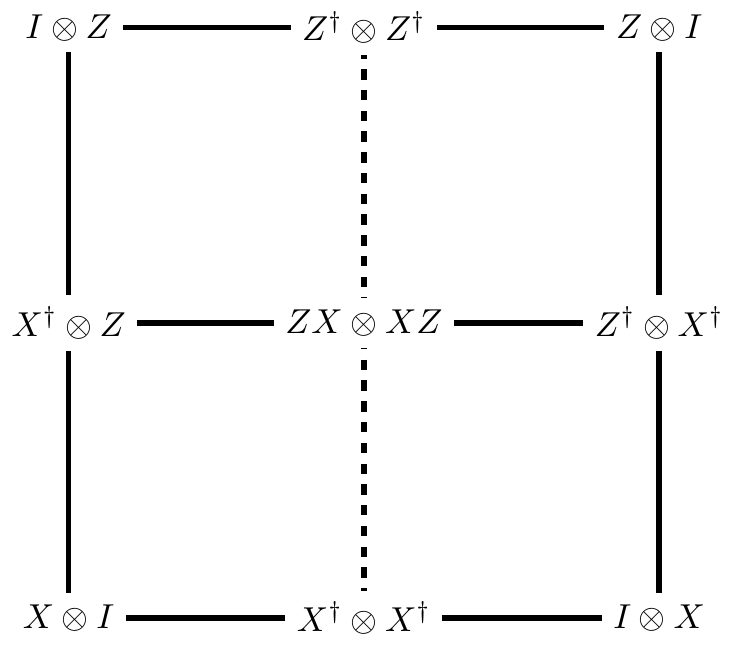}}
	\end{tabular}
	\end{center}
	\caption{The standard operator solution for the Magic Square.}
	\label{eq:mermin-peres-magic-square}
\end{figure}

The robust self-testing theorem for the Magic Square game is the following.
\begin{thm}\label{thm:robust-self-testing-square}
The Magic Square game mod $2$ self-tests the ideal strategy with perfect completeness and $O(\e)$-robustness.
\end{thm}
To prove this, we'll make a direct application of Theorem \ref{thm:robust-self-testing}. However, we'll use the tighter bounds stated in the appendix as Theorem \ref{thm:robust-self-testing-appendix}. We will check that all of its conditions are satisfied by a series of lemmas. 
Throughout, let $\G_2$ be the solution group for the Magic Square game over $\Z_2$. We'll start by identifying $\G_2$ as a group of Pauli operators.

\begin{prop}\label{prop:solution-group-square}
	$\G_2\cong \mathcal{P}_{2}^{\otimes 2}$.
\end{prop}

\begin{cor}\label{cor:G2-satisfies-conditions-4-5}
	$\G_2$ satisfies condition \eqref{assumption:group-test-appendix} of Theorem \ref{thm:robust-self-testing-appendix}. (This is the same as condition \eqref{assumption:w-p-irreps} of Theorem \ref{thm:robust-self-testing}.)
\end{cor}
\begin{proof}[Proof of corollary]
	Let $\tau = \tau_1^{(2)}$ as defined in Definition \ref{def:representations-pauli-group}. By Proposition \ref{prop:representations-pauli-group}, $\G_2$ group-tests $\tau$, giving \eqref{assumption:group-test-appendix}.
\end{proof}

We prove Proposition \ref{prop:solution-group-square} with two lemmas.

\begin{lemma}\label{lemma:commutator-subgroup-square}
	The commutator subgroup $[\G_2,\G_2]$ is $\Braket J$, 
	the cyclic subgroup generated by $J$.
\end{lemma}
\begin{proof}
	First, note that $J$ commutes with everything by construction.
	Next, see that each pair of generators of $\G_2$ has a commutator which is a power of $J$, and that $J$ commutes with all generators. If $w_1,w_2$ are words in the generators, then it holds by induction on the lengths of the words that $w_1w_2 = J^aw_2w_1$ for some $a\in \Z_2$. This proves the inclusion $\G_2' \seq \Braket J$. The reverse inclusion is immediate.
\end{proof}

\begin{lemma}\label{prop:four-edges}
	For generators $s_1,s_2 \in \G_2$, say that the pair $\pair{s_1}{s_2}$ is \emph{intersecting} if the corresponding edges in the constraint graph are incident on a common vertex.
 	Let $x_1,x_2,z_1,z_2$ be any generators of $\G_2$ such that $\set{x_1,x_2},\set{z_1,z_2},\set{x_1,z_2},\set{z_1,x_2}$ are interesecting pairs, while $\set{x_1,z_1},\set{x_2,z_2}$ are not.
 	Then
 	\begin{enumerate}
 		\item\label{item:prop-four-edges-1} $[x_1,z_1] = J = [x_2,z_2]$, and
 	 	\item\label{item:prop-four-edges-2} $\set{x_1,x_2,z_1,z_2, J}$ generates $\G_2$.
 	\end{enumerate}
 \end{lemma} 
 \begin{proof}

 	\ref{item:prop-four-edges-1}. If $x_1$ and $z_1$ are any pair of edges not sharing a vertex, then the group picture of Figure \ref{fig:k33-generators} establishes the twisted commutation relation. If $x_2$ and $z_2$ are any other pair of edges which do not share a vertex, then there is an automorphism of the graph $K_{3,3}$ sending $x_1\mapsto x_2$ and $z_1\mapsto z_2$. Therefore, we can draw the same group picture with a different labeling to prove that $x_2$ and $z_2$ share the same twisted commutation relation.
 	
 	\ref{item:prop-four-edges-2}. See Figure \ref{fig:k33-generators}. Suppose some vertex has only one black edge. Then the group element labeling the black edge is equal to some product of $J$ and the group elements labeling the blue edges at that vertex. So the group generated by the blue edges and $J$ contains the black edge. By the sequence of pictures in Figure \ref{fig:k33-generators}, we see that the four blue edges, together with $J$, generate all nine of the edges. Therefore, they generate all of $\G_2$.
 \end{proof}
From here on, we fix the identification $x_1 = e_7, x_2 = e_9, z_1 = e_3, z_2 = e_1$ (c.f.\ Figure \ref{eq:mermin-peres-magic-square}).

\begin{proof}[Proof of Proposition \ref{prop:solution-group-square}]
	We have the same set of generators for both groups. This gives a surjective function $\mathcal{P}_2^{\otimes 2}\to \G_2$. We've seen that the generators of $\G_2$ satisfy the relations defining $\mathcal{P}_2^{\otimes 2}$; this implies that the function is a group homomorphism. All that remains to check is that the map is injective, i.e.\ has trivial kernel. This holds if the relations of $\G_2$ hold for the preimages of the $e_i$ in $\mathcal{P}_2^{\otimes 2}$. This follows from the fact that the square of operators \eqref{eq:mermin-peres-magic-square} is a Mermin--Peres magic square in the usual sense, i.e.\ operators in the same row or column commute, the products across each row and down the first two columns are $I$, and the product down the last column is $-I$. *Notice that this step fails for the Magic Square game mod $d \neq 2$.*
	
\end{proof}

\begin{figure}[!b]
	\begin{tabular}{m{0.3\textwidth}m{0.05\textwidth}m{0.3\textwidth}m{0.05\textwidth}m{0.3\textwidth}}
	\resizebox{0.3\textwidth}{!}{\includegraphics{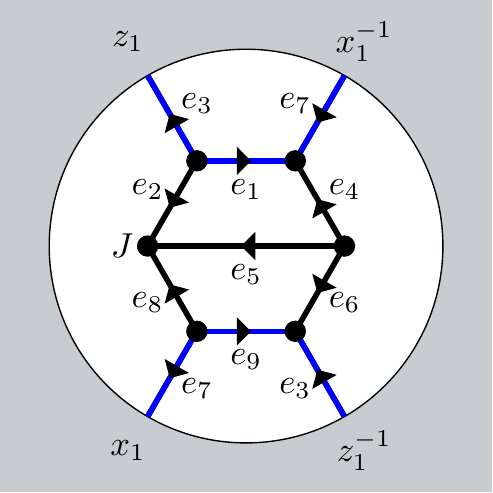}}
        &$\longrightarrow$
	&\resizebox{0.3\textwidth}{!}{\includegraphics{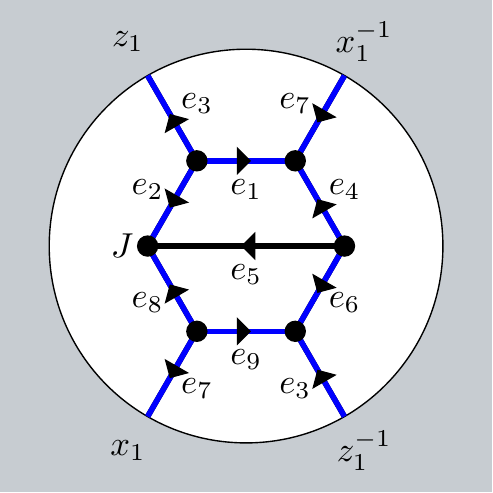}}
        &$\longrightarrow$
	&\resizebox{0.3\textwidth}{!}{\includegraphics{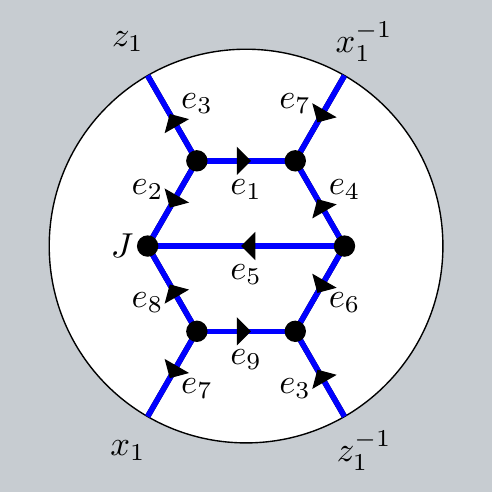}}
	\end{tabular}

    \caption{The group picture proves that $x_1z_1x_1\1z_1\1 = J$ in the solution group for the magic square with the identification $x_1 = e_7, x_2 = e_9, z_1 = e_3, z_2 = e_1$. (Compare Figure \ref{eq:mermin-peres-magic-square}.) The blue-colored edges illustrate that $\set{x_1, z_1, x_2, z_2, J}$ generates the solution group for the magic square.}
    \label{fig:k33-generators}
\end{figure}

\begin{lemma}\label{lemma:small-group-pictures-square}
	Suppose $\m P$ is a $\mathcal{P}_2^{\otimes 2}$-picture in which each generator and relation appears at most $m$ times. Then there is a $\G_2$-picture $\m P'$ witnessing the same equation in which each generator and relation appears at most $3m$ times.
\end{lemma}
This allows us to control the size of group pictures for any relation in $\G_2$ which uses only the letters $x_1,x_2,z_1,z_2,J$. For relations using the other generators, we'll use Lemma \ref{lemma:G2-m0-equals-1}
\begin{proof}
	The generators labeling $\m P$ can be reinterpreted as generators of $\G_2$. $\m P$ has at most $2m$ twisted commutation relations, and the rest of the relations are already relations of $\G_2$. Form $\m P'$ by replacing each twisted commutation relation with a $\G_2$-group picture of the form of Figure \ref{fig:k33-generators}. Each subpicture replacement adds at most one use of each generator and relation.
\end{proof}

\begin{figure}
	\caption{
	The left-hand picture proves that $e_6(x_2z_1) = 1$. This is equivalent to proving $e_6 = \can(e_6) = z_1\1x_2\1$.
	}
	\label{fig:G2-m0-equals-1}
	\begin{center}
	\begin{tabular}{m{0.3\textwidth}m{0.3\textwidth}m{0.3\textwidth}}	
		\resizebox{0.3\textwidth}{!}{\includegraphics{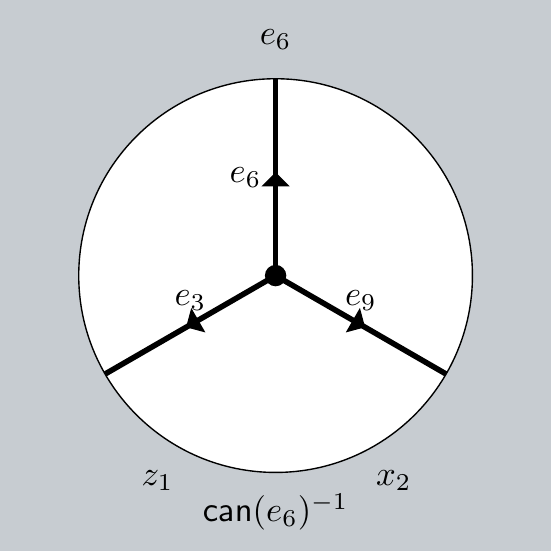}}
		&
		&\resizebox{0.3\textwidth}{!}{\includegraphics{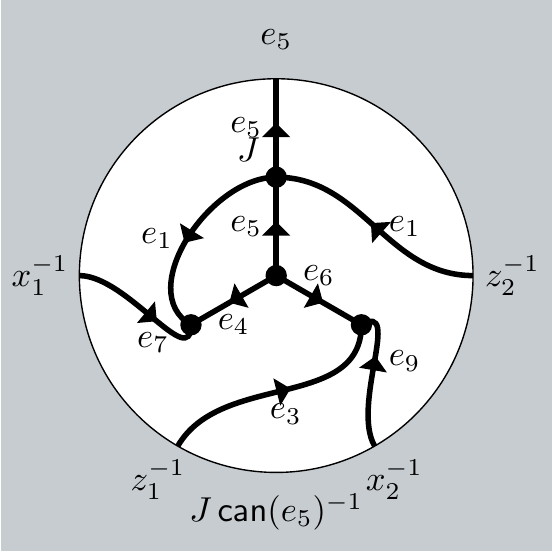}}
	\end{tabular}
	\end{center}
	\caption{
	The right-hand picture proves that $e_5(x_1z_1x_2z_2)\1 = J$. Multiplying both sides by $J\1$, we see that this is equivalent to proving $e_5 = \can(e_5) = Jx_1z_1x_2z_2$. The picture has been drawn with a twisted commutation relation between $e_1$ and $e_5$. To get a valid $\G_2$-picture, this relation must be expanded to a subpicture of the form of Figure \ref{fig:k33-generators}, just as in the proof of Lemma \ref{lemma:small-group-pictures-square}. 
	The new picture thus formed uses each generator and relation at most $3$ times.
	}
	\label{fig:G2-m0-equals-1-b}
\end{figure}

Let $\can$ be the canonical form from \ref{lemma:canonical-form} composed with the isomorphism $\G_2 \cong \mathcal{P}_2^{\otimes 2}$.
\begin{lemma}\label{lemma:G2-m0-equals-1}
	For each generator $e\in E$, the equation $\can(e)e\1 = 1$ has a group picture in which each generator and relation appear at most $3$ times. 
\end{lemma}
\begin{proof}

	Either $\can(e) = e$ as words already, or there is a picture similar to one of the pictures in Figures \ref{fig:G2-m0-equals-1},\ref{fig:G2-m0-equals-1-b}. (Here by ``similar'' we mean ``identical up to relabeling of edges''.) 
\end{proof}

\begin{proof}[Proof of Theorem \ref{thm:robust-self-testing-square}]
We want to apply Theorem \ref{thm:robust-self-testing-appendix}, so we check each of its conditions. The magic square has at most $3$ variables in each equation, so we can take $l_0 = 3$ in condition \eqref{assumption:bounded-degree-appendix}. By Lemma \ref{lemma:G2-m0-equals-1}, we can take $m_0 = 3$ in condition \eqref{assumption:small-pictures-appendix}. By Lemmas \ref{lemma:small-group-pictures-square} and \ref{lemma:small-group-pictures-pauli-n}, we can take $m = 108\cdot 2^2$  in condition \eqref{assumption:small-pictures-w-appendix}. The final two conditions were shown to hold in Corollary \ref{cor:G2-satisfies-conditions-4-5}. We hence apply Theorem \ref{thm:robust-self-testing-appendix} to get the desired conclusion.
\end{proof}

\subsection{Self-testing the Magic Pentagram}
Recall the definition of the Magic Pentagram game from Example \ref{example:magic-pentagram}. 

\begin{definition}[Ideal strategy for Magic pentagram $\pmod 2$] 
	In Figure \ref{fig:magic-pentagram}, associate each operator in the left-hand pentagram with the corresponding variable in the right-hand pentagram. Set $A_e^{(v)}$ to the operator corresponding to $e$, and denote the latter by $A_e$, so that we have $A_e^{(v)} = A_e$ for all $v$. Then set $B_e = \bar{A_e}$ (where any choice of basis works for the conjugation). 
\begin{figure}[h]
	\begin{tabular}{m{0.4\textwidth}m{0.1\textwidth}m{0.4\textwidth}}
		\resizebox{0.4\textwidth}{!}{
			\includegraphics{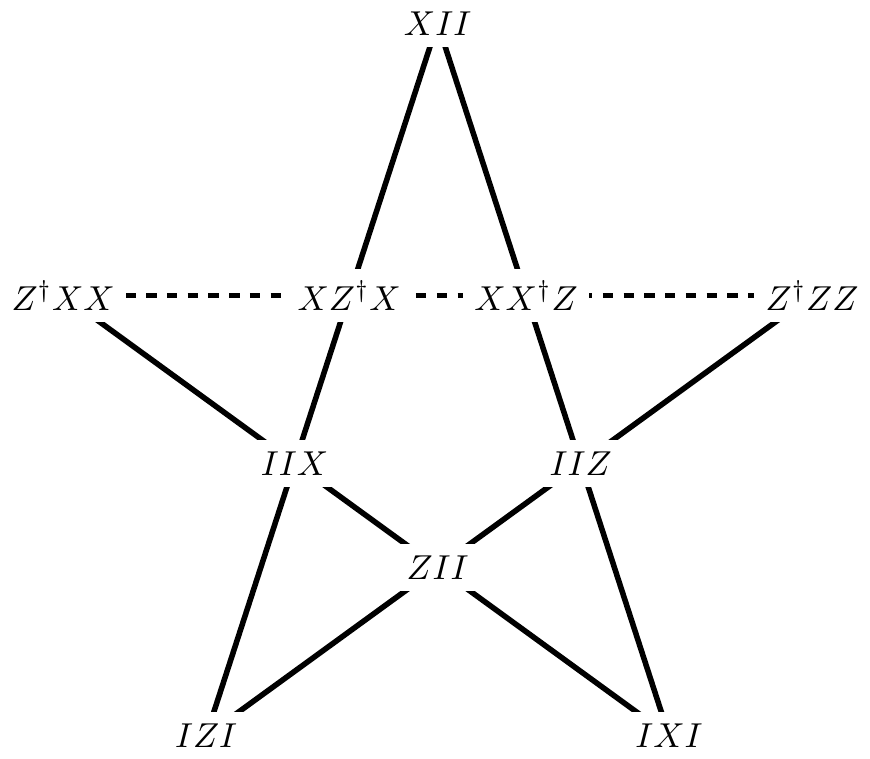}
	   	}
	   	&&
		\resizebox{0.4\textwidth}{!}{
			\includegraphics{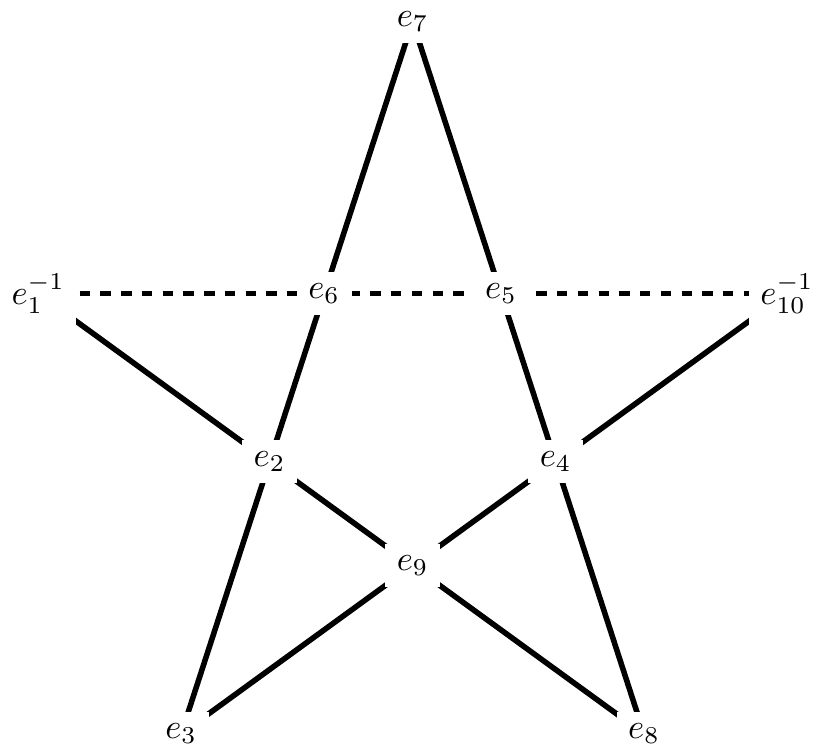}
	   	}
	\end{tabular}
	\caption{The standard operator solution for the Magic Pentagram.}
	\label{fig:magic-pentagram}
\end{figure}

	Set $\ket\psi = \ket{\text{EPR}_2}^{\otimes 3}$. We define $\{A_e^{(v)}\}, \{B_e \}, \ket\psi$ to be the \emph{ideal strategy} for the Magic Pentagram game.
\end{definition}

The robust self-testing theorem for the Magic Pentagram game is the following.
\begin{thm}\label{thm:robust-self-testing-pentagram}
The Magic Pentagram game mod $2$ self-tests the ideal strategy with perfect completeness and $O(\e)$-robustness.
\end{thm}
Again, to prove this, we will make a direct application of Theorem \ref{thm:robust-self-testing}, but we will use the tighter bounds stated in the appendix as Theorem \ref{thm:robust-self-testing-appendix}.

Let $\G_3$ be the solution group for the Magic Pentagram. We give the proof details only where they differ from the Magic Square case.

\begin{prop}\label{prop:solution-group-pentagram}
	$\G_3\cong \mathcal{P}_2^{\otimes 3}$. 
\end{prop}
\begin{cor}\label{cor:G3-satisfies-conditions-4-5}
	$\G_2$ satisfies condition \eqref{assumption:group-test-appendix} of Theorem \ref{thm:robust-self-testing-appendix} with $\tau = \tau_1^{(3)}$ from Definition \ref{def:representations-pauli-group}.
\end{cor}

\begin{lemma}\label{lemma:commutator-subgroup-pentagram}
	The commutator subgroup $[\G_3,\G_3]$ is $\Braket J$, 
	the cyclic subgroup generated by $J$.
\end{lemma}

\begin{lemma}\label{prop:six-edges}
 	Let $x_1,x_2,x_3,z_1,z_2,z_3$ be any generators of $\G_3$ such that in the linear constraint graph, the edge pairs $\set{x_i,x_j},\set{z_i,z_j},\set{x_i,z_j},i\neq j$ are \emph{intersecting} (see Lemma \ref{prop:four-edges}), while the edge pairs $\set{x_i,z_i}$ are not. 
 	Then
 	\begin{enumerate}
 		\item\label{item:prop-six-edges-1} $[x_i,z_i] = J $, and
 	 	\item\label{item:prop-six-edges-2} $\set{x_i,z_i, J\;i\leq 3}$ generates $\G_3$.
 	\end{enumerate}
 \end{lemma} 
 \begin{proof}

 	\ref{item:prop-six-edges-1}.
 	If $x_1$ and $z_1$ are any pair of edges not sharing a vertex, then the group picture of Figure \ref{fig:k5-generators} establishes the twisted commutation relation. If $x_i$ and $z_i$ are any other pair of edges which do not share a vertex, then there is an automorphism of the graph $K_{5}$ sending $x_1\mapsto x_i$ and $z_1\mapsto z_i$. Therefore, we can draw the same group picture with a different labeling to prove that $x_i$ and $z_i$ share the same twisted commutation relation.
 	
 	\ref{item:prop-six-edges-2}.
 	See Figure \ref{fig:k5-generators}, which is interpreted the same way as Figure \ref{fig:k33-generators} from the Magic Square case.
 \end{proof}
We fix the identification $x_1 = e_7, z_1 = e_9, x_2 = e_8, z_2 = e_3, x_3 = e_2, z_3 = e_4$ (c.f.\ Figure \ref{fig:magic-pentagram}.)

\begin{figure}
    \caption{
    The leftmost group picture proves that $x_1z_1x_1\1z_1\1 = J$ in $\G_3$, with $x_1 = e_7, z_1 = e_9$. Identifying further $x_2 = e_8, z_2 = e_3, x_3 = e_2, z_3 = e_4$ and following the color of the edges shows that $\set{x_i,z_i, J\;i\leq 3}$ generates $\G_3$.
    }
    \label{fig:k5-generators}
	\begin{center}
	\begin{tabular}{
	m{0.3\textwidth}
	m{0.03\textwidth}
	m{0.3\textwidth}
	m{0.3\textwidth}
	}
		\resizebox{0.3\textwidth}{!}{
			\includegraphics{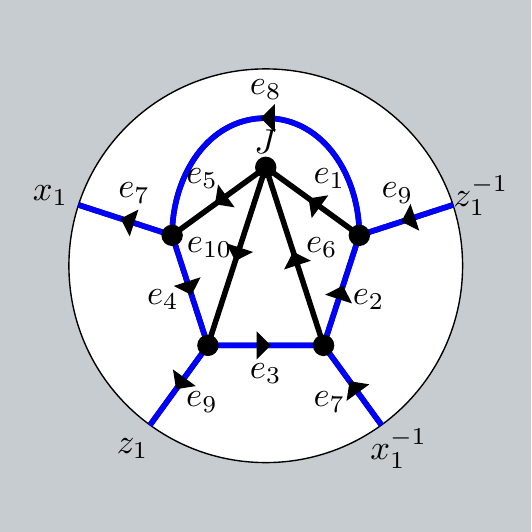}
	    	}
	     &$\longrightarrow$
	 	 &\resizebox{0.3\textwidth}{!}{
			\includegraphics{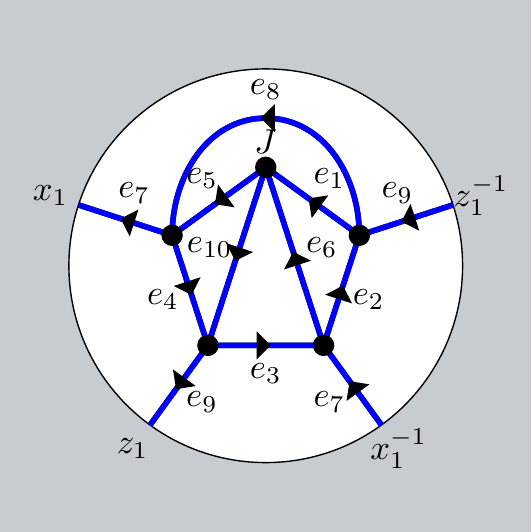}
	        }
		& \resizebox{0.3\textwidth}{!}{
			\includegraphics{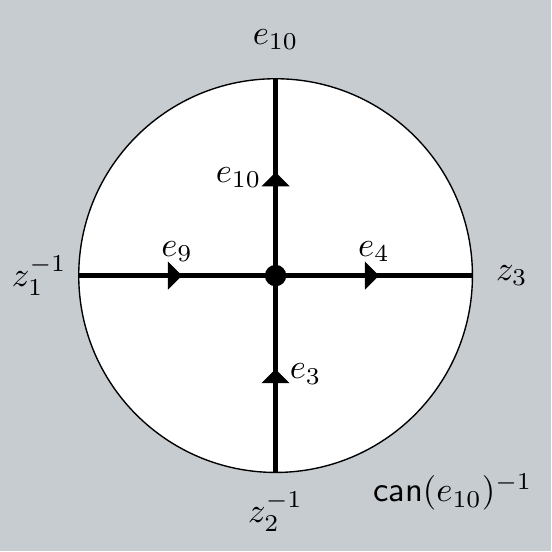}
	        }
	\end{tabular}
	\end{center}
	\caption{
		The rightmost figure is a $\G_3$-picture showing $\can(e_{10}) = z_1z_2z_3\1$.
	}
	\label{fig:G3-m0-equals-1}
\end{figure}

\begin{proof}[Proof of Proposition \ref{prop:solution-group-pentagram}]
	As in the Magic Square case, all that remains to check is that the generators of $\mathcal{P}_2^{\otimes 3}$ satisfy the relations of $\G_3$. This amounts to checking that the pentagram of operators in Figure \ref{fig:magic-pentagram} is a $2$-dimensional Mermin Magic Pentagram in the usual sense, i.e.\ operators on the same line commute, the alternating products across the four solid lines are each $I$, and the alternating product across the dashed line is $-I$.
\end{proof}

\begin{lemma}\label{lemma:small-group-pictures-pentagram}
	Suppose $\m P$ is a $\mathcal{P}_2^{\otimes 3}$-picture in which each generator and each relation appears at most $m$ times. Then there is a $\G_3$-picture $\m P'$ witnessing the same equation in which each generator and each relation appears at most $4m$ times.
\end{lemma}
\begin{proof}
	We use the same proof as in the Magic Square case, but now there are up to $3m$ twisted commutation relations in $\m P$.
\end{proof}

Let $\can$ be the canonical form from Lemma \ref{lemma:canonical-form} composed with the isomorphism $\G_3 \cong \mathcal{P}_2^{\otimes 3}$.
\begin{lemma}\label{lemma:G3-m0-equals-1}
	For each generator $e\in E$, the equation $\can(e)e\1 = 1$ has a group picture in which each generator and relation appear at most once. 
\end{lemma}
\begin{proof}

	Either $\can(e) = e$ as words already, or there is a picture similar to the picture in Figure \ref{fig:G3-m0-equals-1}.
\end{proof}

\begin{proof}[Proof of Theorem \ref{thm:robust-self-testing-pentagram}]
Again we wish to apply Theorem \ref{thm:robust-self-testing-appendix}, so we check each of its conditions. The Magic Pentagram constraint system has at most $4$ variables in each equation, so we can take $l_0 = 4$. Let $\can$ be the canonical form for $\mathcal{P}_2^{\otimes 3}$ from Lemma \ref{lemma:canonical-form} composed with the group isomorphism $\mathcal{P}_2^{\otimes 3} \cong \G_3$. By Lemma \ref{lemma:G3-m0-equals-1}, we can take $m_0 = 1$. By Lemmas \ref{lemma:small-group-pictures-square} and \ref{lemma:small-group-pictures-pauli-n}, we can take $m = 162\cdot 2^2$. That the final two conditions are satisfied is Corollary \ref{cor:G3-satisfies-conditions-4-5}. Applying Theorem \ref{thm:robust-self-testing-appendix} gives the desired statement.
\end{proof}

\subsection{Self-testing $n$ pairs of maximally entangled qubits and $n$-qubit Paulis}

Self-testing in parallel has gained recent interest both as a potential tool in cryptographic protocols, and as a simple way to witness high-dimensionality of a quantum system. Various recent results have shown self-testing of $n$ maximally entangled pairs of qubits, and associated $n$-qubit Pauli measurements, in particular using copies of the Magic Square game \cite{coladan17, CN16}. In this section, we show how our general self-testing result of Theorem \ref{thm:robust-self-testing} allows to produce a similar result.

We will introduce first a notion of product of LCS games. Informally, we form the product of LCS games by adding new equations to enforce commutativity between each variable in one game with each variable in the others. 
The main motivation for this definition is that we can express the solution group of the product as an appropriate product of the solution groups. First, we introduce our notion of product for groups over $\Z_d$ (we will only make use of it mod $2$ though).

\begin{definition}
	Let $G_i = \braket{S_i:R_i}_{\Z_d}$ be a family of groups presented over $\Z_d$. Define their \emph{product over $\Z_d$} as
	\begin{align}
		\prod^{\Z_d}_i G_i :=
		\Braket{
		\bigsqcup_i S_i:
		\bigsqcup_i R_i \sqcup R_\text{prod}
		}_{\Z_d},
		&&
		R_\text{prod} = \set{[s,s']\;s\in S_i, s'\in S_j, i\neq j}.
	\end{align}
	Here the symbol $\sqcup$ denotes the disjoint union.
\end{definition}
One can check that this product is the categorical product in the category which has objects the groups presented over $\Z_d$ and has maps the group homomorphisms that send $J\mapsto J$. Therefore it obeys the usual properties one expects from a product. In particular it has an equivalent definition as a repeated application of an associative binary product $\overset{\Z_d}\times$. The following is easy to check from Definition \ref{definition:n-qudit-pauli-group}.

\begin{lemma}
	The $n$-qudit Pauli group is the $n$-fold product over $\Z_d$ of the $1$-qudit Pauli group, i.e.\
	\begin{equation}
		\m P_d^{\otimes n} = \prod^{\Z_d}_{i\in[n]}\m P_d^{\otimes 1}.
	\end{equation}
	 As a corollary, 
	 $\m P_d^{\otimes n_1} \overset{\Z_d}\times \m P_d^{\otimes n_2} = \m P_d^{\otimes(n_1+n_2)}$.
\end{lemma}

\begin{definition}[LCS game product]
\label{def:lcs-game-product}
	Let $G_i=\LCS(\mathbf H_i, l_i, \Z_d)$, $i\in [n]$ be LCS games over $\Z_d$ with $\mathbf H_i=(H_i,V_i,E_i)$. We define their \emph{product LCS game} as $\prod_iG_i:= \LCS(\mathbf H, l, \Z_d)$, where $\mathbf H = (H,V,E)$ andbbb
	\begin{align}
	V &= \left(\bigsqcup_i V_i\right)\sqcup V_\text{prod}
	&V_\text{prod} &= \set{v_{xy}\; x\in E_i, y\in E_j, i\neq j}
	\\E &= \left(\bigsqcup_i E_i\right)\sqcup E_\text{prod}
	&E_\text{prod} &= \set{e_{xy}\; x\in E_i, y\in E_j, i\neq j}
	\end{align}
	In words, we add one equation and one variable for each pair of variables $x,y$ living in distinct factor games. We call the new variable $e_{xy}$, and the equation is
	\begin{equation}
		x + y - e_{xy} = 0.
	\end{equation}
	The definition of the equations can be formalized as follows.
	\begin{align}
	l(v) &= \begin{cases}
		l_i(v), &\text{ if }v\in V_i\\
		0, 		&\text{ if }v\in V_\text{prod}\\
		\end{cases}
	&
	H(v,e) &= \begin{cases}
		H_i(v,e), &\text{ if }v \in V_i\text{ and }e\in E_i\\
		1, &\text{ if }v = v_{xy}\text{ and }e \in \pair xy\\
		-1, &\text{ if }v = v_{xy}\text{ and }e = e_{xy}\\
		0, &\text{ otherwise }
		\end{cases}
	\end{align}
\end{definition}

\begin{lemma}
\label{lemma:lcs-game-product}
Let $G_0 = \prod_{i=1}^nG_i = \LCS(\mathbf H_0, l_0, \Z_d)$.
~\begin{enumerate}
	\item If each of the $G_i$ satisfy 
	$\forall v:\sum_e\abs{H_i(v,e)} \leq \Delta_i$ then $G_0$ satisfies the same with
	$\Delta_0 = \max\left(\set{\Delta_i}\cup\set{3}\right)$.
	\item $\abs{E_0} \leq (\sum_i\abs{E_i})^2$
	\item $\abs{V_0} = \abs {E_0} + \sum_i \abs {V_i}$
	\item\label{item:LCS-product-is-functorial} Let $\G_i = \G(\mathbf H_i,l_i,\Z_d)$. Then $\G_0 = \prod\limits^{\Z_d}_{i>0}\G_i$. 
\end{enumerate}
\end{lemma}
\begin{proof}
	We prove only \eqref{item:LCS-product-is-functorial}, which is less straightforward than the rest.

	Let $\G = \prod\limits_{i>0}^{\Z_d}\G_i$. There's a clear inclusion map $\iota:\G\into \G_0$, since the generators of the former are a subset of the generators of the latter.  To see that $\iota$ is a homomorphism, we show that the relations of $\G$ are a subset of the relations of $\G_0$. For each commutation relation introduced by the group product, there is an equation in the product containing the same variables. So the solution group $\G_0$ has a corresponding commutation relation.
	To see that this map is surjective, notice that $e_{xy} = xy$ in $\G_0$, so all of the generators of $\G_0$ lie in the image of $\iota$. To see that $\iota$ is injective, check that every relation in the presentation $\G_0$ is already true of the pre-image elements in $\G$. 
\end{proof}

\begin{definition}
	Let $G_2$ and $G_3$ be the Magic Square and Magic Pentagram LCS games over $\Z_2$, respectively.
	For $n\geq 4$, construct $G_n$ as an LCS game product as follows:
	\begin{align}
		G_{2k} := \prod_{i\in [k]}G_2,
		&&
		G_{2k+1} :=\prod(\underbrace{G_2,\ldots, G_2}_{k-1}, G_3),
	\end{align}
	where $\prod$ is the LCS game product from Definition \ref{def:lcs-game-product}. 
\end{definition}
From Lemma \ref{lemma:lcs-game-product} and properties of $G_2$ and $G_3$, we can deduce basic properties of $G_n$.
\begin{lemma}~
\label{lemma:basic-Gn-properties}
	\begin{enumerate}
		\item $G_n$ satisfies $\forall v:\sum_e \abs{H_n(v,e)} \leq 4$.
		\item $\abs{E_n} \leq \abs{V_n} \leq 25n^2$.
		\item $G_n$ has solution group $\G_n \cong \m P_2^{\otimes n}$
	\end{enumerate}
\end{lemma}

Next, we describe the winning strategy for the game $G_n$. First, we give an abstract description, and then we unpack it into a concrete description. Understanding either description should suffice to appreciate Theorem \ref{thm:self-testing-pauli-LCS}.

\begin{definition}[Ideal strategy for the game $G_n$, abstract]
	Let $\tau_n^{(1)}: \mathcal{P}_2^{\otimes n}\to U(\C^2)^{\otimes n}$ be as in Definition \ref{def:representations-pauli-group}. Then let $\tau$ be composition of that map with the isomorphism $\mathcal{P}_2^{\otimes n}\cong \G_n$. The ideal strategy is that which follows from applying the construction of Proposition \ref{prop: perfect strategy} to the operator solution $\tau$.
\end{definition}

As a structural hint for what follows, notice that each observable measured by the provers is a Pauli operator of weight at most $5$. (It is either a magic square operator, a magic pentagram operator, a tensor product of magic square operators, or a tensor product of a Magic Square operator with a Magic Pentagram operator.) 

Recall by the definition of product game, that, for $n = 2k$, $V_n = \bigsqcup_{i \in [k]} V_2^{(i)} \sqcup V_\text{prod}$ and $E_n = \bigsqcup_{i \in [k]} E_2^{(i)} \sqcup E_\text{prod}$ where $V_2^{(i)}$ and $E_2^{(i)}$ are the vertex and edge sets for the $i$th copy of the magic square game, and $V_\text{prod} = \set{v_{xy}, x\in E_i, y\in E_j, i\neq j}$ and $E_\text{prod} = \set{e_{xy}, x\in E_i, y\in E_j, i\neq j}$. Similarly, for $n=2k+1$, $V_n = \bigsqcup_{i \in [k-1]} V_2^{(i)} \sqcup V_3 \sqcup V_\text{prod}$ and $E_n = \bigsqcup_{i \in [k-1]} E_2^{(i)} \sqcup E_3 \sqcup V_\text{prod}$ where $V_3$ and $E_3$ are the vertex and edge sets corresponding to the magic pentagram game. 

\begin{definition}[Ideal strategy for the game $G_n$]
Let $\{A^{\text{sq}}_e\}$, $\{B^{\text{sq}}_e\}$ and $\{A^{\text{pt}}_e\}$, $\{B^{\text{pt}}_e\}$  be observables from the ideal strategies of the magic square and magic pentagram games mod $d$ respectively.
\begin{itemize}
\item For $n = 2k$, the $A_e^{(v)}$ are observables on $\C^{2^n}$ (which we think of as $k$ copies of $\C^{2^2}$). If $v \in V_2^{(i)}$ and $e \in E_2^{(i)}$, let $A_e^{(v)} = (A^{\text{sq}}_e)_i \otimes I := A_e$, where the subscript indicates that the observable acts on the $i$th of the $k$ copies, and the identity is on everything else; if $v = v_{ey}$ or $v = v_{xe} \in V_\text{prod}$ and $e \in E_2^{(i)}$, let $A_e^{(v)} =  (A^{\text{sq}}_e)_i \otimes I := A_e$; if $v = v_{xy} \in V_\text{prod}$ and $e = e_{xy}$ for $x \in E_2^{(i)}$ and $y \in E_2^{(j)}$, then let $A_e^{(v)}$ = $(A^{\text{sq}}_x)_i \otimes (A^{\text{sq}}_y)_j \otimes I := A_e$. Finally, let $B_e = \bar{A_e}$.
\item For $n = 2k+1$, the observables are on $\C^{2^n}$ (which we think of as $k-1$ copies of $\C^{2^2}$ and one copy of $\C^{2^3}$). The only changes from the even case are the following: if $v \in V_3$ and $e \in E_3$, let $A_e^{(v)} = (A^{\text{pt}}_e)_k \otimes I := A_e$, where the $k$ subscript denotes the last $\C^3$ register; if $v = v_{ey}$ or $v = v_{xe} \in V_\text{prod}$ and $e \in E_3^{(i)}$, let $A_e^{(v)} =  (A^{\text{pt}}_e)_k \otimes I := A_e$; if $v = v_{xy} \in V_\text{prod}$ and $e = e_{xy}$ for $x \in E_2^{(i)}$ and $y \in E_3$, then let $A_e^{(v)}$ = $(A^{\text{sq}}_x)_i \otimes (A^{\text{pt}})_k \otimes I := A_e$, and similarly for the symmetric case. As before, let $B_e = \bar{A_e}$.
\end{itemize}
Set $\ket\psi = \ket{\text{EPR}_2}^{\otimes n}$. Define $\set{A_e^{(v)}}$, $\set{B_e}$, $\ket\psi$ to be the ideal strategy.
\end{definition}

\begin{thm}\label{thm:self-testing-pauli-LCS}

The product game $G_n$ mod $2$ self-tests the ideal strategy with perfect completeness and $O(n^{10}\e)$-robustness.
\end{thm}

\begin{lemma}\label{lemma:small-group-pictures-pauli-LCS}
	Suppose $\m P$ is a $\mathcal{P}_2^{\otimes n}$-picture in which each relation and each generator appears at most $m$ times. Then there is a $\G_n$-picture $\m P'$ witnessing the same equation in which each relation and each generator appears at most $4m$ times. 
\end{lemma}
\begin{proof}
	As in Lemmas \ref{lemma:small-group-pictures-square} and \ref{lemma:small-group-pictures-pentagram}, we take $\m P$ and replace the twisted commutation relations with small subpictures. There are at most $3m$ twisted commutation relations from each factor game $G_2$ or $G_3$, and each one is replaced by a $\G_2$- or $\G_3$-picture. Each of these replacements adds at most one instance of each generator and relation.
\end{proof}

Let $\can$ be the canonical form from \ref{lemma:canonical-form} composed with the isomorphism $\G_n \cong \mathcal{P}_2^{\otimes n}$.
\begin{lemma}
\label{lemma:small-group-pictures-w-pauli-LCS}
	For each generator $e\in E_n$, the equation $\can(e)e\1 = 1$ has a group picture in which each generator and relation appear at most three times. 
\end{lemma}
\begin{proof}
	If $e$ comes from a Magic Square factor or a Magic Pentagram factor, then we apply Lemma \ref{lemma:G2-m0-equals-1} or \ref{lemma:G3-m0-equals-1}, respectively. If $e = e_{xy}$ is an auxiliary variable, then we glue the pictures for $x$ and $y$ together.
\end{proof}

\begin{proof}[Proof of Theorem \ref{thm:self-testing-pauli-LCS}]
	We again seek to apply theorem \ref{thm:robust-self-testing-appendix}, so we check each of its conditions. From Lemma \ref{lemma:basic-Gn-properties}, we have 
	$l_0 = 4$. 
	From Lemma \ref{lemma:small-group-pictures-w-pauli-LCS}, we have $m_0 = 3$.  
	From Lemmas \ref{lemma:small-group-pictures-pauli-n} and \ref{lemma:small-group-pictures-pauli-LCS}, we have $m = 72\cdot 2^2n$.
	Finally, it follows from item $3$ of Proposition \ref{lemma:basic-Gn-properties} that $\G_n$ group-tests $\tau_1^{(n)}$, whose image contains an isomorphic copy of $\mathcal{P}_2^{\otimes n}$.  Then, applying Theorem \ref{thm:robust-self-testing-appendix} gives the desired bound.
	
\end{proof}

\section{Concluding remarks}
We have presented a general robust self-testing theorem for a certain class of linear constraint system games, using the group-theoretic approach of Cleve, Liu and Slofstra. We specialized this theorem to the cases of the Magic Square and Magic Pentagram games mod $2$, obtaining robust self-testing theorems for respectively two and three pairs of maximally entangled qubits and the associated ideal measurements. Furthermore, we applied our theorem to a certain $n$-fold product of these games to obtain a robust self-testing theorem for a tensor product of $n$ maximally entangled qubits and the associated $n$-qubit Pauli measurements. The following are some remaining open questions.

\begin{question}
	For which $d$ is there a non-local game self-testing the maximally entangled state of local dimension $d$?
\end{question}
In \cite{coladangelo2018generalization}, one of us gave a family of Bell inequalities whose maximal violation self-tests the maximally entangled game of local dimension $d$ for any $d > 2$. However, this translates into non-local games which are not pseudotelepathy (i.e. they do not have a perfect strategy). 
If one considers just pseudotelepathy games, it is known that there are no non-local games which can be won perfectly with dimension $2$ but not with dimension $1$ (i.e.\ with a classical strategy) \cite{brassard2004minimum}; this answers the above question in the negative for $d=2$. For $d=3$, there is a pseudotelepathy game based on Kochen-Specker sets for which there is no classical perfect strategy and there is a perfect quantum strategy using the maximally entangled state of qutrits. \cite{cleve2004consequences} To the best of the authors' knowledge, the question of rigidity is open for this game.

Restricting our attention to just to LCS games, we can reduce the above question to the following:
\begin{question}
	For any $d$, is there a solution group $\G$ presented over $\Z_d$ such that every irreducible representation of $\G$ is either $1$-dimensional or $d$-dimensional, sending $J$ to a nontrivial $d\th$ root of unity? 
\end{question}
$\m P_d^{\otimes 1}$ has this property, but it can be shown that it is not a solution group. 
Classifying the groups with the above property seems like a nontrivial question in group theory.

The following question represents, in the authors' opinion, the main barrier to understanding the full power of general LCS games.
\begin{question}
	Which groups are solution groups? In particular, are there any nonabelian finite solution groups other than the groups of qubit paulis?
\end{question}
William Slofstra gave a partial answer to this question by proving \emph{embedding theorems.} First, he proved \cite{slofstra2016tsirelson} that all finitely presented groups embed into some solution group over $\Z_2$. Next, he proved \cite{slofstra2017set} that a somewhat broad class of groups, including all finite symmetric groups, embed into binary solution groups such that the embedding preserves the finite-dimensional approximate representation structure. These embeddings fail to preserve finiteness of the group, however.


\bibliography{ms}

\begin{appendices}

\section{Some Inequalities}
We record here inequalities which we need and whose proofs are not particularly enlightening. We'll use $\Re z:= \frac{z + \bar z}2$ denote the real part of $z\in \C$.
\begin{lemma}
\label{lemma:convex-inequality-easy}
	Let $d\geq 2$ be an integer and $\a = \sum_i \a_i \w_d^i$, where $\a_i$ are nonnegative reals with $\sum_i \a_i = 1$. Suppose $\a_0 \geq 1-\e$. Then $1 - \Re\a \leq \abs{1-\a} \leq 2\e$.
\end{lemma}
\begin{proof}
	Make repeated applications of the inequality $\abs{a+b} \leq \abs a + \abs b$. 
	\begin{align}
		\abs{\a - 1}
		& \leq (1-\a_0) + \sum_i\abs{\a_i\w_d^i}
		\\& = 2(1-\a_0) 
		\\&\leq 2\e.
	\end{align}
\end{proof}
Now we prove a converse which is slightly more technical.
\begin{lemma}
\label{lemma:convex-inequality-hard}
	Let $d\geq 2$ be an integer and $\a = \sum_i \a_i \w_d^i$, where $\a_i$ are nonnegative reals with $\sum_i \a_i = 1$. 
	Suppose $1- \Re\a \leq \e$. Then $1-\a_0 \leq \frac12d^2\e$.

\end{lemma}
The quadratic dependence on $d$ is optimal. Take for example $\a = \a_0 + (1-\a_0) \frac{\w_d+\w_d\1}2$. 
\begin{proof}
Recalling that $\Re\w_d = \cos\frac{2\pi}d$, we establish the following inequality for all integers $d\geq 2$.
\begin{equation}
	1 - \cos \frac{2\pi}p \geq \frac{2}{d^2}.
\end{equation}
It suffices to use the fourth order Taylor series for cosine and the inequality $\pi^2 - \frac{\pi^4}{3d^2} \geq 1$, true for $d\geq 1.92$.
\begin{equation}
	1 - \cos \frac{2\pi}d \geq \frac{2^2\pi^2}{2!d^2} - \frac{2^4\pi^4}{4!d^4} 
	\geq \frac{2}{d^2}\left[\pi^2 - \frac{\pi^4}{3d^2}\right]
	\geq \frac{2}{d^2}.
\end{equation}

\noindent
From this we conclude that $\frac{1}{1 - \Re\w_d} \leq \frac12d^2$.
 We'll also use that the primitive $d\th$ root of unity has maximal real part among the $d\th$ roots of unity, i.e.\ $\Re \w_d^i \leq\Re \w_d$ for all $i\neq 0$. (We write in two columns to save space. Read the left column first.)
	\begin{align}
		1-\e& \leq \Re \a
		& \a_0(1-\Re\w_d) &\geq 1 - \Re\w_d - \e
		\\	& \leq \a_0 + \sum_{i\neq 0} \a_i \Re (\w_d^i)
		& \a_0 &\geq 1 - \frac{\e}{1 - \Re\w_d}
		\\	& \leq \a_0 + (1-\a_0)(\Re \w_d)
		& \a_0 &\geq 1 - \frac{1}{1 - \Re\w_d} \e
		\\	& \leq \a_0(1- \Re\w_d) + \Re \w_d
		& \a_0 &\geq 1 - \frac12d^2 \e.
	\end{align}
\end{proof}

\begin{lemma}\label{lemma:entanglement-monogamy}
	Let $\m H_A, \m H_B, \m H_C$ Hilbert spaces. Let $\rho_{ABC}$ be a state on $\m H_A\otimes \m H_B\otimes \m H_C$. Let $\rho_{AB} = \Tr_C \rho_{ABC}$. Let $\ket\psi_{AB}$ be a pure state on $\m H_A\otimes \m H_B$. Suppose that 
	\begin{equation}
	\label{eq:assumption-entanglement-monogamy}
		\drho{\rho_{AB}}{\proj \psi_{AB}}{I_{AB}}^2 \leq \e. 
	\end{equation}
	Then there is some state $\rho_{\text{aux}}$ on $\m H_C$ such that 
	\begin{equation}
		\norm{\rho_{ABC} - \proj\psi_{AB}\otimes \rho_{\text{aux}}}_1 \leq 6\e.
	\end{equation}
\end{lemma}
\begin{proof}

	Let $\ket\phi_{ABCC'}$ be a purification of $\rho_{ABC}$, i.e.\ suppose that $\Tr_{C'} \proj\phi = \rho_{ABC}$. We examine a Schmidt decomposition of $\ket\phi$, cutting along subystems $AB/CC'$. Let
	\begin{equation}
		\ket{\phi}_{ABCC'} = \sum_i \sqrt{\l_i} \ket{i_{AB}}\otimes\ket{i_{CC'}}.
	\end{equation}
	where $\l_i >0$ for all $i$.
	Tracing out $C'$, we have
	\begin{equation}
		\rho_{ABC} = \sum_i \l_i \proj{i_{AB}}\otimes\Tr_{C'}\proj{i_{CC'}}.
	\end{equation}
	Now let $\rho^{(i)}_{\text{aux}} = \Tr_C\proj{i_{CC'}}$. One can compute the distance between $\rho_{ABC}$ and $\proj{i_{AB}}\otimes\rho^{(i)}_{\text{aux}}$ as
	\begin{equation}
	\label{eq:entanglement-monogamy-1}
		\frac12\norm{\rho_{ABC} - \proj{i_{AB}}\otimes\rho^{(i)}_{\text{aux}}}_1 = 1- \l_i.
	\end{equation}
	By the same computation,
	\begin{equation}
	\label{eq:entanglement-monogamy-2}
		\frac 12\norm{\rho_{AB} - \proj{i_{AB}}}_1 
		= 1- \l_i.
	\end{equation}
	\noindent
	The $\l_i$ are the eigenvalues of $\rho_{AB}$; let $\l_1$ be the greatest. Then we have
	\begin{align}
		\l_1 
		&\geq \braket{\psi|_{AB}\rho_{AB}|\psi}_{AB} 
		\\&= 1 - \drho\rho{\proj\psi}{I}^2. 
		\\ &\geq 1-\epsilon
	\end{align}
where we applied assumption \eqref{eq:assumption-entanglement-monogamy} to get the last line. We use the following inequality, valid for arbitrary $\rho$ and $\ket\psi$,
	\begin{align*}
		\frac12 \norm{\rho - \proj \psi}_1 
		& = \frac12 + \Tr\rho^2 - \sandwich\psi\rho
		\\& \leq 1 - \sandwich\psi\rho,
	\end{align*}
	to conclude that 
	\begin{equation}
	\label{eq:entanglement-monogamy-3}
		\frac12\norm{\rho_{AB} - \proj\psi_{AB}}_1 \leq \e.
	\end{equation}
	Finally, we apply several triangle inequalities, to obtain the following. (To clarify any confusion, here $\proj {1_{AB}}$ is $\proj {i_{AB}}$ when $i = 1$):
	\begin{align}
		\frac12 \norm{\rho_{AB} - \proj\psi_{AB}}_1
		&\leq \e
		& \text{Equation \eqref{eq:entanglement-monogamy-3}}
		\\
		\frac12\norm{\proj {1_{AB}} - \proj\psi_{AB}}_1
		&\leq 2\e
		& \text{Triangle inequality with Equation \eqref{eq:entanglement-monogamy-2}}
		\\
		\frac12 \norm{\proj {1_{AB}}\otimes \rho^{(1)}_{\text{aux}} - \proj\psi_{AB}\otimes \rho^{(1)}_{\text{aux}}}_1
		&\leq 2\e
		& \text{Tensoring }\rho^{(1)}_{\text{aux}}
		\\
		\frac12 \norm{\rho_{ABC} - \proj\psi_{AB}\otimes \rho^{(1)}_{\text{aux}}}_1
		&\leq 3\e
		& \text{Triangle inequality with Equation \eqref{eq:entanglement-monogamy-1}}.
	\end{align}
This concludes the proof.
\end{proof}


\section{Tighter bounds via more parameters}
In Section \ref{sec:self-testing}, we introduced one complexity parameter for LCS games and gave a robustness bound in terms of that parameter. Here, we give a tighter robustness bound at the expense of cumbersome bookeeping of parameters. We give new statements of the lemmas from \S \ref{subsection:robust-self-testing}. The proofs are essentially the same, and are omitted. We give the subscript $0$ to parameters which are typically constant.

\begin{theorem}
\label{thm:robust-self-testing-appendix}
	Let $G$ be a linear constraint game over $\Z_d$ with vertex set $V$, edge set $E$, and constraints given by  $H:V\times E \to \Z_d$ and $l: V\to \Z_d$. Let $\G$ be the solution group of $G$. Suppose that:
	\begin{enumerate}[(i)]
		\item 
		\label{assumption:bounded-degree-appendix}
		 each equation has at most $l_0$ variables with multiplicity, i.e.\ $\forall v:\sum_e\abs{H(v,e)} \leq l_0$,
		\item 
		
		\label{assumption:small-pictures-w-appendix}
		there is a canonical form $\can$ such that every equation of the form  $\can(e)e\1 = 1$ for $e\in E$ is witnessed by a $\G$-picture in which each generator and relation appears at most $m_0$ times,
		\item 
		\label{assumption:small-pictures-appendix}
		 every equation of the form $\can(g)\can(gh)\1\can(h) = 1$ $g,h\in \G$ is witnessed by a $\G$-picture 
		 proving in which 
		 each generator and 
		 each relation is used at most $m$ times,
		\item \label{assumption:group-test-appendix}
		 $\G$ group-tests $\tau: \G \to U(\C^{d^n}) $ in the sense of Definition \ref{definition:group-test}.
		 \item \label{assumption:pauli-in-image-appendix}The image of $\tau$ contains an isomorphic copy of the Pauli group $\m P_d^{\otimes n}$.
	\end{enumerate}	
Then $G$ self-tests the strategy $\tilde A_e^{(v)} = \tau(e), \tilde B_e = \bar{\tau(e)}, \ket\psi = \ket{\r{EPR}_{d^n}}$
	with perfect completeness and $O\left((m_0l_0dm\abs E \abs V)^2 \e\right)$-robustness.
\end{theorem}

\begin{proof}
	The proof is the same as the proof of Theorem \ref{thm:robust-self-testing}, but with different parameters.
	Using Lemmas \ref{lemma:Bs-are-approximate-conjugate-operator-solution-appendix} and \ref{lemma:canonical-form-implies-stability-appendix}, we can get $\eta_1 = 2^4l\abs E \abs V \sqrt \e, \eta_2 = 2^{10}m_0l_0m\abs E \abs V \sqrt \e$. The rest of the argument goes through unmodified.
\end{proof}

\begin{lemma}[c.f.\ Lemma \ref{lemma:Bs-are-approximate-conjugate-operator-solution}]
\label{lemma:Bs-are-approximate-conjugate-operator-solution-appendix}
	$\set{B_e}$ is an ``approximate conjugate operator solution'' in the following sense:
	\begin{align}
	\label{eq:Bs-approximately-satisfy-constraints-appendix-appendix}
		\sum_v \drho{\rho}{
		\prod_{e\in r_v} I_A\otimes B_e}{\w_d^{-l(v)}I} 
		&\leq 4l_0\abs E \abs V\sqrt\e.
		\\
	\label{eq:Bs-approximately-commute-appendix-appendix}
		\sum_{\substack{e,e'\\e\sim e'}} \drho{\rho}{I_A\otimes [B_e,B_{e'}]}{I}
		&\leq 4l_0 \abs E\abs V\sqrt\e
		.\
	\end{align}
	Furthermore, $\set{A_e^{(v_e)}}$ is an ``approximate operator solution'' in the same sense with a slightly worse parameter, i.e.\ 
	\begin{align}
	\label{eq:As-approximately-satisfy-constraints-appendix-appendix}
		\sum_v\drho{\rho}{\prod_{e\in r_v}A_e^{(v_e)}\otimes I_B}{\w_d^{l(v)}}
		&\leq 8l_0\abs E\abs V\sqrt{\e},
		\\
	\label{eq:As-approximately-commute-appendix-appendix}
		\sum_{\substack{e,e'\\e\sim e'}} \drho{\rho}{\left[A_e^{(v_e)},A_{e'}^{(v_{e'})}\right]\otimes I_B}{I}
		&\leq 8l_0 \abs E\abs V\sqrt\e
		.
	\end{align}
	Finally, these ``solutions'' are consistent in the sense that
	\begin{equation}
	\label{eq:As-and-Bs-approximately-consistent-appendix-appendix}
		\sum_e \drho{\rho}{A_e^{(v_e)}\otimes B_e}{I}
		 \leq 2 \abs E \abs V\sqrt\e.
	\end{equation}
\end{lemma}

\begin{lemma}[c.f.\ Lemma \ref{lemma:canonical-form-implies-stability}]
\label{lemma:canonical-form-implies-stability-appendix}
	
	Suppose that $\set{A_e^{(v_e)}}$ and $\set{B_e}$ $\eta$-satisfy the relations from $R$ in the sense that
	\begin{align}
	\label{lemma:canonical-form-implies-stability-1-appendix}
		\sum_{r\in R} \drho{\rho}{\prod_{e\in r}{A_e^{(v_e)}\otimes I}}I \leq \eta, 
		&&\text{ and }
		&&
		\sum_{r\in R} \drho{\rho}{I\otimes \prod_{e\in r}{B_e}}I \leq \eta.
	\end{align}
	 Furthermore, suppose that $\set{A_e^{(v_e)}}$ and $\set{B_e}$ are $\eta$-consistent in the sense that
	 \begin{equation}
	 \label{lemma:canonical-form-implies-stability-2-appendix}
	 	\sum_e\drho\rho{A_e^{(v_e)}\otimes B_e}I \leq \eta.
	 \end{equation}
	 Then 
 \begin{itemize}
	\item $f_A$ and $f_B$ are consistent, i.e.\ for all $x\in \G$,
	\begin{equation}
	\label{eq:canonical-form-implies-consistency-appendix-appendix}
		\drho\rho{f_A(x)\otimes f_B(x)}I \leq m\eta.
	\end{equation}
 	\item $f$ is close to Alice and Bob's strategy pointwise, i.e.\ for all $e\in E$,
	  \begin{align}
	  \label{eq:canonical-form-implies-stability-conclusion-1-appendix-appendix}
	  	\drho{\rho}{f_A(e)\otimes I_B}{ A_e^{(v_e)}\otimes I_B } &\leq 5m_0\eta
	  	\\
	  	\drho{\rho}{I_A \otimes f_B(e)}{ I_A \otimes B_e} &\leq 5m_0\eta. 
	  \end{align}
	  \item $f$ is ``approximately a homomorphism'', i.e.\ for all $x,y\in \G$, 
	\begin{align}
	\label{eq:canonical-form-implies-stability-appendix-appendix}
		\drho\rho{f_A(x)f_A(yx)\1f_A(y)\otimes I_B}I &\leq 17mm_0\eta,
		\\
		\drho\rho{I_A\otimes f_B(x)f_B(yx)\1f_B(y)}I &\leq 17mm_0\eta.
	\end{align}
 \end{itemize}

\end{lemma}

\end{appendices}

\end{document}